\documentclass[lettersize,journal]{IEEEtran}
\usepackage{amsmath,amsfonts}
\usepackage{algorithmic}
\usepackage{algorithm}
\usepackage{array}
\usepackage[caption=false,font=normalsize,labelfont=sf,textfont=sf]{subfig}
\usepackage{textcomp}
\usepackage{stfloats}
\usepackage{url}
\usepackage{verbatim}
\usepackage{graphicx}
\usepackage{cite}
\usepackage{multirow}
\usepackage{subfig}
\usepackage{pifont} % 用于打钩和打叉
\usepackage{booktabs}
\usepackage{multirow}
\usepackage{placeins}
\newcommand{\cmark}{\ding{51}} % 打钩
\newcommand{\xmark}{\ding{55}} % 打叉
\newtheorem{theorem}{Theorem}
\newtheorem{lemma}{Lemma}
\newtheorem{proof}{Proof}
\newtheorem{definition}{Definition}
\begin{document}

\title{VFEFL: Privacy-preserving federated learning
against malicious clients via verifiable functional encryption}

\author{
	Nina Cai,
	Jinguang Han,
	Weizhi Meng
}
        % <-this % stops a space



\maketitle

\begin{abstract}
	Federated learning is a promising distributed learning paradigm that enables collaborative model 
	training without exposing local client data, thereby protecting data privacy. 
	However, it also brings new threats and challenges. 
	The advancement of model inversion attacks has rendered the plaintext transmission 
	of local models insecure, while the distributed nature of federated learning makes 
	it particularly vulnerable to attacks raised by malicious clients.  
	To protect data privacy and prevent malicious client attacks, 
	this paper proposes a privacy-preserving Federated Learning framework
	based on Verifiable Functional Encryption (VFEFL), without a non-colluding 
	dual-server assumption or additional trusted third-party.
	Specifically, we propose a novel Cross-Ciphertext Decentralized Verifiable Functional Encryption (CC-DVFE) scheme 
	that enables the verification of specific relationships over multi-dimensional ciphertexts. 
	This scheme is formally treated, in terms of definition, security model and security proof.
	Furthermore, based on the proposed  CC-DVFE scheme, we design a privacy-preserving 
	federated learning framework that incorporates a novel robust aggregation rule 
	to detect malicious clients, enabling the effective training of high-accuracy models 
	under adversarial settings.
	Finally, we provide the formal analysis and empirical evaluation of VFEFL. 
	The results demonstrate that our approach achieves the desired privacy protection, 
	robustness, verifiability and fidelity, while eliminating the reliance on non-colluding dual-server 
	assumption or trusted third parties required by most existing methods.
\end{abstract}


\section{Introduction}
\IEEEPARstart{F}{ederated}
learning\cite{FL2017} is an advanced and promising distributed learning paradigm 
that enables multiple clients to collaboratively train models without sharing 
local data. It embodies the principles of focused data collection and minimization, 
thereby mitigating many of the systemic privacy risks and costs 
inherent in traditional, centralized machine learning approaches.
With the enactment of privacy protection and data security regulations in various countries, 
as well as the increase in public privacy awareness, the application of federated 
learning has gradually expanded e. g., healthcare, finance, and 
edge computing.

Although the federated learning framework offers a degree of privacy protection for 
client data during training, it brings new privacy and security issues \cite{OpenFL}. 
One primary concern is the potential for privacy leakage 
through the local models uploaded by clients. With the advent and ongoing advancement 
of attack techniques, such as model inversion attacks\cite{GIA,MIA,GIA2}, submitting local 
models in plaintext form is no longer regarded as a secure method. 
Accordingly, obfuscation and encryption mechanisms have been adopted to address this issue.
Another threat is attacks from malicious clients. Due to its distributed architecture, 
federated learning  is susceptible to adversarial manipulation by malicious clients.
A malicious client may make the accuracy of the global model 
drop rapidly by uploading malicious models, 
or destroy the decryption by uploading a malicious decryption key or a wrong ciphertext.
The above issues are important for federated learning systems 
to achieve their core goal (training high quality models while preserving data privacy). 
Privacy-preserving mechanisms (such as encryption schemes) can potentially lead 
to difficulties in detection and verification, while detecting and verifying 
encrypted models of the clients may compromise specific information. 
Existing solutions \cite{AegisFL,FheFL} cannot be deployed 
in the basic federated learning architecture (one server and multiple clients), 
and need help from two non-colluding servers or additional trusted third parties, 
which introduces additional limitations and consumption.

\IEEEpubidadjcol
Therefore it is very interesting and challenge to design  a federated learning scheme 
that realizes privacy protection and resists malicious client attacks in the basic federated learning architecture.
This paper proposes a privacy-preserving federated learning scheme 
against malicious clients based on decentralized verifiable functional encryption.
Our contributions can be summarized as follows:
\begin{enumerate}
	\item We employ a verifiable functional encryption scheme to encrypt local models in the federated learning, 
		ensuring data privacy and correctness during encryption and decryption.
		The scheme supports verifiable evaluation of relations over multidimensional ciphertexts, 
		making it well suited for federated learning scenarios.
		Furthermore, we formalize its definition and security model, and reduce its security 
		in the proposed model to well-established computational hardness assumptions.
	\item We propose a privacy-preserving federated learning scheme named VFEFL 
		that protects local model privacy and defends against malicious clients. 
		Compared with existing approaches, it enables verifiable client-side aggregated weights 
		and can be integrated into standard federated learning architectures to enhance trust. 
		Theoretical analysis demonstrates that VFEFL provides strong guarantees in 
		both privacy and robustness.
	\item We implement a prototype of VFEFL and conduct comprehensive experiments 
		to evaluate its fidelity, robustness, and efficiency. 
		Robustness is evaluated under both targeted and untargeted poisoning attacks. 
		Experimental results demonstrate that VFEFL effectively defends against 
		such attacks while preserving model privacy.
\end{enumerate}

% \IEEEPARstart{T}{his} file is intended to serve as a ``sample article file''
% for IEEE journal papers produced under \LaTeX\ using
% IEEEtran.cls version 1.8b and later. The most common elements are covered in the simplified and updated instructions in ``New\_IEEEtran\_how-to.pdf''. For less common elements you can refer back to the original ``IEEEtran\_HOWTO.pdf''. It is assumed that the reader has a basic working knowledge of \LaTeX. Those who are new to \LaTeX \ are encouraged to read Tobias Oetiker's ``The Not So Short Introduction to \LaTeX ,'' available at: \url{http://tug.ctan.org/info/lshort/english/lshort.pdf} which provides an overview of working with \LaTeX.

\section{RELATED WORK}
\subsection{Functional Encryption}
Functional encryption (FE) \cite{FE_def, FE} is a paradigm supports selective computation 
on encrypted data, while traditional public key encryption only implements all-or-nothing decryption method. 
An authority in a FE scheme can generate restricted decryption keys to allow receiver
to learn specific functions of the encrypted messages and nothing else.
Goldwasser\cite{MIFE2014} first introduced multi-input functional encryption (MIFE), which 
is a generalization of functional encryption to the setting of multi-input functions. 
Michel Abdalla \cite{MIFE2017} constructed the 
first MIFE for a non-trivial functionality with polynomial security loss 
under falsifiable assumptions for a super-constant number of slots. 

However, whereas the input can be large, the original definition of FE 
limits the source of the encrypted data to a single client, which is not compatible with
a wide range of real-world scenarios where data needs to be collected from multiple sources.
Multi-client functional encryption (MCFE)\cite{DMCFE} was proposed to addresses this issue. 
Further, decentralized multi-client functional encryption (DMCFE)\cite{DMCFE} eliminated
reliance on a trusted party. 
In the DMCFE scheme instantiated over an asymmetric pairing group, 
each client independently samples $s_i \leftarrow \mathbb{Z}^2_p$, 
and all clients jointly generate random matrices $T_i \leftarrow \mathbb{Z}_p^{2 \times 2}$ 
such that $\sum_{i \in [n]} T_i = 0$.
$s_i$ and $T_i$ forms the secret key held by each client. 
This interactive procedure allows clients to locally construct their own keys 
without relying on a trusted third party. 

The decentralized architecture of DMCFE brings forth additional challenges posed by potentially dishonest clients, 
prompting the proposal of verifiable functional encryption (VFE) \cite{VFE} as a solution.
Nguyen et al.\cite{VMCFE} proposed a verifiable decentralized 
multi-client functional encryption scheme for inner products (VDMCFE), 
along with its formal definitions and security proofs. 
This work extended the DMCFE scheme by constructing the 
Class Group-Based One-Time Decentralized Sum (ODSUM) and 
Label-Supporting DSUM (LDSUM) schemes,
leveraging the Combine-then-Descend technique to improve the efficiency 
of correctness proofs for functional decryption key shares.  
Additionally, it incorporated zero-knowledge proofs such as range proofs to ensure 
the correctness of ciphertexts and the validity of the encrypted content.

\subsection{Federated Learning}
\begin{table*}[ht]
	\centering
	\caption{Comparison with other schemes}
	\label{related}
	\vspace{1mm}
	\begin{tabular}{lccccc}
		\toprule
		Proposed Works & Techniques & Robustness & Self-Contained & Verifiability & Fidelity \\
		\midrule
		Biscotti\cite{Biscotti}     & DP+SS      & \cmark & \xmark & \cmark & \xmark \\
		SecProbe\cite{SecProbe}     & DP         & \cmark & \cmark & \xmark & \xmark \\
		ESFL\cite{ESFL}             & DP         & \cmark & \cmark & \xmark & \xmark \\
		SecureFL\cite{SecureFL}     & PLHE+SS    & \cmark & \xmark & \cmark & \cmark \\
		AegisFL\cite{AegisFL}       & HE         & \cmark & \xmark & --     & \cmark \\
		RVPFL\cite{RVPFL}           & HE         & \cmark & \xmark & \cmark & \cmark \\
		FheFL\cite{FheFL}           & FHE        & \cmark & \cmark & \xmark & \cmark \\
		FEFL\cite{FEFL}             & 2DMCFE     & \xmark & \cmark & \xmark & \cmark \\
		BSR-FL\cite{BSRFL}          & NIFE       & \cmark & \xmark & \xmark & \cmark \\
		Ours                        & CC-DVFE    & \cmark & \cmark & \cmark & \cmark \\
		\bottomrule
	\end{tabular}
	\vspace{1mm}

	\begin{minipage}{\linewidth}
		\textit{Note:}
		\textbf{Self-contained} means the scheme works in a single-server setup without trusted third parties.
		\textbf{Verifiability} ensures the client follows protocol specifications.
		\textbf{Fidelity \cite{FLtrust}} means no accuracy loss when no attack exists.
	\end{minipage}
\end{table*}

\textbf{Federated Learning against Malicious Clients.} 
Robust aggregation rules are one of the most common means of resisting malicious clients in federated learning. 
Next, we review some of the existing Byzantine robust aggregation rules (ARG).
\begin{itemize}
	\item \textbf{FedAvg\cite{FL2017}.} FedAvg  is the most basic aggregation rule 
		constructed when federated learning was first proposed. It directly computes 
		the weighted sum of all local model to obtain the global model, where the 
		weight of each client is determined by the proportion of its local dataset 
		size to the total dataset size. 
		% Formally,
		% $$
		% 	W^t=\sum_{i=0}^{n} \frac{\|D_i\|}{N}W^t_i
		% $$
		% where $N$ is the total number of training examples. 
		However, FedAvg 
		is not resistant to malicious client attacks, where the adversary manipulates 
		the global model at will through submitted local models.
	\item \textbf{Krum\cite{Krum}.} This AGR selects global model update 
		based on the squared distance score.
		Formally, for client $i$, the score is defined as  
		$s(i) = \sum_{i \to j, n-f-2} \| W_i - W_j\|^2$
		where $f$ represents the maximum number of malicious clients that can be tolerated
		and $i \to j$ denotes that $W_j$ belongs to the $n-f-2$ closest updates to $W_i$.
		Subsequently, the model with the lowest score or the average of several models with 
		the lowest scores was chosen as the global model.
	% \item \textbf{Bulyan\cite{bulyan}.} In this AGR, the mean values of a number of 
	% 	parameters closest to the median in each dimension are selected as parameters 
	% 	of the global model for that dimension. Formally, let $\Theta=n-2f$, 
	% 	$\bar{n}=\Theta-2f$, the global gradient is defined as:
	% 	$$
	% 		\forall j \in [1,\cdots, m], \quad W[j] = 
	% 			\frac{1}{\bar{n}} \sum_{W_i \in \mathcal{M} }X[j]
	% 	$$
	% 	where $\mathcal{M}[i]=\underset{\mathcal{R} \subset \mathcal{S},|\mathcal{R}|=\bar{n}}{\arg \min }\left(\sum_{X \in \mathcal{R}}|X[i]-\operatorname{med}[i]|\right)$ 
	% 	and $\operatorname{med}[i]=\underset{m=Y[i], Y \in \mathcal{S}}{\arg \min _{\mathcal{S}}}\left(\sum_{Z \in \mathcal{S}}|Z[i]-m|\right)$
	% \item \textbf{DnC\cite{DnC}.} This aggregation method improves the quality of aggregation 
	% 	by detecting and rejecting outliers through singular value decomposition (SVD).
	\item \textbf{FLTrust\cite{FLtrust}.} The method requires the server to maintain a small clean dataset, 
		which is used to train a trusted model update \(g_0\). A trust score for 
		each local model update \(g_i\) is then computed by evaluating the cosine similarity 
		between \(g_i\) and \(g_0\), followed by applying a ReLU operation. 
		The score is used as the aggregation weight of the local model in the aggregation process. 
		% Formally, the global model update is defined as:
		% $$
		% 	c_i = \frac{\left\langle\boldsymbol{g}_i, \boldsymbol{g}_0\right\rangle}
		% 	{\left\|\boldsymbol{g}_i\right\| \cdot\left\|\boldsymbol{g}_0\right\|}
		% $$
		% $$
		% 	g = \frac{1}{\sum_{j=1}^n \operatorname{ReLU}\left(c_j\right)} 
		% 		\sum_{i=1}^n \operatorname{ReLU} \left(c_i\right) \cdot 
		% 		\frac{\left\|\boldsymbol{g}_0\right\|}{\left\|\boldsymbol{g}_i\right\|} 
		% 		\cdot \boldsymbol{g}_i,
		% $$ 
		% The SecureFL\cite{SecureFL} scheme implements a privacy-enhanced 
		% version of FLTrust through homomorphic encryption.
	\item \textbf{AGR in FheFL\cite{FheFL}.} In this schemes, the aggregation methods use 
		the Euclidean distance between the local model $W^t_i$ and the global model $W^{t-1}$
		of the previous round as a discriminator of the quality of the local model $W^t_i$.
		% The basic idea of this AGR is 
		% \begin{align*}
		% 	d_i^{(t)} &= \|W^{t-1} - W^t_i\|^2 \\
		% 	s_i &= 1 - \frac{d_u^{(t)}}{\sum_{i=1}^{n} d_i^{(t)}} 
		% \end{align*}
		% Final aggregation by weight $s_i$ as 
		% $$
		% 	W^t = \frac{1}{n-1}\sum_{i=1}^{n} s_i W^t_i
		% $$
	\item \textbf{AGR in BSR-FL\cite{BSRFL}.} This aggregation method constructs a qualified aggregation 
	model using secure cosine similarity and an incentive-based Byzantine robustness strategy. 
	In BSR-FL, the aggragated benign model from the \((t-1)\)-th iteration \(W_{base}^{t-1}\), 
	serves as the baseline model. The trust score \(cs^t_i\) for each local model \(W_i^t\) 
	is computed as the cosine similarity between \(W_i^t\) and the baseline \(W_{base}^{t-1}\), 
	and then transformed using a parameterized sigmoid function centered at 0.5. 
		However, the BSR does not verify the length of local models, 
		leaving the potential for spoofing.
\end{itemize}

Existing aggregation rules are either not robust enough to resist adaptive attacks\cite{FLtrust}, 
or too complex to be applicable for fast verification of zero-knowledge proofs.
Therefore, we propose a new aggregation rule that is applicable to zero-knowledge proofs 
and provides an effective defense against attacks such as adaptive attacks.

\textbf{Privacy-Preserving FL against Malicious Clients.}
Although client data is not directly exchanged in the federated learning framework, 
the transmission of local models still introduces potential privacy risks. 
In order to make PPFL resistant to the behavior of malicious participants, 
some works have been proposed to implement the above defense strategies in PPFL. 
The relevant schemes can be categorized into 3 groups according to 
privacy techniques.
A comparison of these previous schemes with our method is summarized in Table \ref{related}.
\begin{itemize}
	\item[-] \textbf{Differential Privacy Based FL.} The Biscotti\cite{Biscotti} combined 
	differential privacy, Shamir secret shares (SS) and Multi-Krum aggregation rule
	to realize a private and secure federated learning scheme based on the blockchain.
	SecProbe\cite{SecProbe} aggregated local models based on the computed utility 
	scores and applies a functional mechanism to perturb the objective 
	function to achieve privacy.
	Miao et al.\cite{ESFL} 
	proposed an efficient and secure federated learning (ESFL) framework based 
	on DP. Since the nature of differential privacy techniques is to hide the 
	sensitive information of an individual by introducing random noise in the 
	data or computational results, it cannot avoid affecting the 
	accuracy of the model.
	\item[-] \textbf{Homomorphic Encryption Based FL.} SecureFL\cite{SecureFL} 
	implemented FLTrust construction under encryption based 
	on packed linear homomorphic encryption. However, in this scheme, the 
	client distributes the local gradient to two servers through secret sharing, 
	and the privacy protection of the gradient must rely on two non-colluding 
	servers. AegisFL\cite{AegisFL} improved the computational efficiency of 
	multiple aggregation rules under ciphertexts by optimizing the homomorphic 
	encryption scheme, but it also fails to eliminate dependency on two non-colluding servers. 
	The RVPFL\cite{RVPFL} scheme combined homomorphic encryption and Diffie-Hellman 
	key agreement to achieve privacy protection and robustness detection, 
	but relies on additional participants.
	% The PBFL\cite{PBFL} scheme, on the other hand, can be implemented on a single 
	% server, but in this scheme, multiple users share a single key, risking model leakage.
	% 待处理
	FheFL\cite{FheFL} built the scheme on homomorphic encryption in a 
	single server. However, a trusted key distribution process is necessary to ensure 
	the secure generation of pairwise keys among users. 
	The scheme does not validate the client behavior, and the malicious clients can 
	easily destroy the decryption through wrong keys.
	% The above two schemes lead to 
	% a certain privacy risk due to key sharing among multiple clients, i.e., a 
	% malicious client can achieve decryption if it obtains the ciphertext from 
	% an honest client. 
	Such schemes often use homomorphic encryption for privacy, 
	computing aggregation on ciphertexts. Decryption typically requires a trusted 
	party or two non-colluding servers. Shared decryption keys risk exposing honest 
	clients' data to key holders through collusion or eavesdropping.
	% In order to achieve Byzantine robustness while maintaining 
	% privacy, such schemes usually use homomorphic encryption schemes to encrypt 
	% the model, and then compute the robust aggregation rules in the 
	% ciphertext state. However, since the result of homomorphic encryption is 
	% in the ciphertext state, it is common to rely on two non-colluding server 
	% setups (where one of them is responsible for decryption) or a trusted third parties
	% to help with the decryption. 
	% In federated learning, homomorphic encryption schemes suffer from an additional 
	% security issue: since all encrypted models share the same decryption key, 
	% the party holding the key (either a client or a server) may obtain the encrypted models 
	% of honest clients through eavesdropping or collusion, 
	% thereby compromising their private data, all without the honest clients' knowledge.
	% or to let the 
	% clients share the key (where multiple clients use the same key) to decrypt the result.
	\item[-] \textbf{Functional Encryption Based FL.} Chang et al.\cite{FEFL} constructed a 
	privacy-preserving federated learning scheme FEFL via functional 
	encryption. Qian et al.\cite{DMCFEFL} designed another 
	privacy-preserving federated learning scheme by combining functional 
	encryption with multiple cryptographic primitives. However, both schemes 
	do not address the detection of malicious clients. The BSR-FL scheme \cite{BSRFL} 
	combined functional encryption and blockchain technology to construct an efficient 
	Byzantine robust privacy-preserving federated learning framework. However, the 
	scheme still requires two non-colluding entities (the task publisher and the server) and does not validate 
	the magnitude of the local models submitted by clients.
\end{itemize}

\section{PROBLEM STATEMENT}
\subsection{System Model}
The VFEFL architecture proposed in this paper is applicable to a wide range of 
privacy-preserving federated learning settings and does not require any trusted third party. 
Its architecture involves two main roles: 
a server $\mathcal{S}$ and $n$ clients $\{\mathcal{C}_i\}_{i \in [n]}$.

\textbf{Server:} The server, equipped with a clean dataset $D_0$ and substantial computational 
	resources, is responsible for detecting maliciously encrypted local models submitted 
	by clients, aggregating the models from clients, and returning the aggregated 
	model to the clients for further training. This process is iteratively performed 
	until a satisfactory global model is achieved.

\textbf{Clients:} Clients have their own datasets and a certain amount of computational 
	power, and it is expected that multiple clients will cooperate to train the model 
	together without revealing their own data and model information.

The notations in the paper are described in Table \ref{notations}

\begin{table}[ht] 
	\centering
	\caption{Notations and Descriptions}\label{notations}
	% \vspace{1mm}
	\begin{tabular}{ll}
	  \toprule
	Notations & Descriptions \\ %& Notations & Descriptions \\
	  \midrule
	  $n$& the number of clients \\
	  $m$ & the dimension of model\\
	  $t$& training epoch \\ 
	  $lr$ & local learning rate\\
	  $R_l$ & the number of local iterations\\
	  $\ell$ & the label in the CC-DVFE\\ 
	  $W^t$& global model \\ 
	  $W_0$& initial global model \\ 
	  $W_i^t$& local model \\  
	  $D_0$& root dataset \\ 
	  $D_i$& local dataset \\
	  $g, w_i, u, v$ & the generators of $\mathbb{G}_1$\\  
	  $h, \hat{v}_1, \hat{v}_2$ & the generators of $\mathbb{G}_2$\\
	%   $g^a$($g \in \mathbb{G}^n$,$a \in \mathbb{Z}_p^n$) & $\prod_{i\in[n]}g_i^{a_i}$\\
	  \bottomrule
	\end{tabular}
	\begin{flushleft}
		\textit{Note:}
		The superscripts of $W_i^t, \vec{x}_i^t, D_i$ 
		denote interactions and the subscripts denote client identifiers.
	\end{flushleft}
\end{table}
\subsection{Design Goals}
% Based on the above privacy preservation and robustness challenges, we propose a 
% federated learning scheme that can be deployed in a single-server setup 
% without a trusted third party for more general application scenarios, 
% where all messages submitted by a client are verified 
% against client spoofing and corruption. 
VFEFL is designed to provide the following performance and security guarantees:
\begin{enumerate}
	\item[-] \textbf{Privacy.} VFEFL should ensure that the server or channel 
		eavesdroppers do not have access to the client's local model to reconstruct the 
		client's private data. 
	\item[-] \textbf{Byzantine Robustness.} VFEFL can resist Byzantine attacks such 
		as Gaussian attack\cite{fang2020local}, the adaptive attack\cite{FLtrust} and 
		Label Flipping attack, 
		and obtain a high-quality global model.
	\item[-] \textbf{Self-contained.} VFEFL can be deployed in the basic 
		federated learning framework, i.e., without relying on two non-colluding servers 
		and  additional third parties, 
		thus gaining more general applicability. 
	\item[-] \textbf{Verifiability.} 
		% VFEFL ensures that the outputs of the clients follows the protocol 
		% specification, including the ciphertexts and the functional key share. 
		% In other words, a malicious client cannot accuse an 
		% honest client or disrupt the aggregation process. 
		% It also guarantees that the system will not be interrupted even if part of the client data is corrupted during transmission.
		VFEFL ensures clients' outputs comply with the protocol, including ciphertexts 
		and functional key shares, preventing malicious clients from framing 
		honest ones or disrupting aggregation. It also ensures system continuity despite 
		partial data corruption during transmission.
	\item[-] \textbf{Fidelity.} VFEFL architecture causes no loss of accuracy in the 
			absence of attacks
\end{enumerate}
We will proof the above properties in Section \ref{analysis}.

\subsection{Security Model}
% According to the Universal Composability (UC) framework\cite{UCS}, 
The protocol $\pi$ securely realizes a function $f$ if, for any adversary $\mathcal{A}$, 
there exists no environment $\mathcal{E}$ can distinguish whether it is interacting 
with $\pi$ and $\mathcal{A}$ or with  the ideal process for $f$ and 
ideal adversary $\mathcal{A}^*$ with non-negligible probability.
Formally, let $\text{REAL}_{\mathcal{A}}^{\mathcal{\pi}}(\lambda, \hat{x})$ and 
$\text{IDEAL}_{\mathcal{A}^*}^{f}(\lambda, \hat{x})$
denote the interaction views of $\mathcal{E}$ in the real-world and ideal-world models, 
respectively. For any $\mathcal{E}$, if the following holds:
$$
	\text{REAL}_{\mathcal{A}}^{\mathcal{\pi}}(\lambda, \hat{x}) \stackrel{c}{\approx} \text{IDEAL}_{\mathcal{A}^*}^{f}(\lambda, \hat{x}),
$$
where \( \stackrel{c}{\approx} \) denotes computational indistinguishability, 
\( \lambda \) represents the security parameter, 
and \( \hat{x} \) includes all inputs.

In the design goal of VFEFL, privacy is defined that neither the 
server nor eavesdroppers can obtain the local models of clients, thus preventing 
them from inferring the clients' data samples during interactions between the 
server and clients. In the execution of VFEFL protocols, an attacker may compromise 
either some clients or the server, thereby gaining access to 
encrypted data and intermediate results. 
With them, 
the attacker seeks to extract private information. 
Thus, we define the ideal function of VFEFL as $f_{\mathrm{VFEFL}}$.
% , and 
% and the details of $f_{\mathrm{VFEFL}}$ are discussed in Section \ref{analysis_privacy}.
The privacy of VFEFL can be formally described as:
$$
	\text{REAL}_{\mathcal{A}}^{\mathrm{VFEFL}}
			(\lambda, \hat{x}_\mathcal{S},\hat{x}_\mathcal{C}, \hat{y}) 
	\stackrel{c}{\approx} 
	\text{IDEAL}_{\mathcal{A}^*}^{f_{\mathrm{VFEFL}}}
			(\lambda, \hat{x}_\mathcal{S},\hat{x}_\mathcal{C}, \hat{y}),
$$
where $\hat{y}$ denotes the $y$ and intermediate results, 
$\hat{x}_\mathcal{S},\hat{x}_\mathcal{C}$ denote the 
inputs of $\mathcal{S}$ and $\mathcal{C}$, respectively.
The formal proof as discussed in Section \ref{analysis_privacy}.

\subsection{Threat Model}
We consider the honest-but-curious server, which executes the protocol faithfully 
but may leverage techniques such as gradient inversion attacks on the client's model
to infer private data.
In this paper, we categorize the clients $\mathcal{C}$ into two groups: 
honest clients $\mathcal{HC}$ and malicious clients $\mathcal{CC}$. 
An honest client honestly trains a local model based on its real dataset in a 
training iteration, encrypts and proofs the model, generates the correct 
functional key share and the proofs,  
and submits them to $\mathcal{S}$. 
A malicious client may submit random or carefully crafted model parameters to degrade the 
accuracy of the global model, or submit incorrect encrypted models or functional key shares 
to cause decryption failure. Additionally, a malicious client may eavesdrop on 
an honest client's secrets by colluding with others. 

Our system requires at least two clients to be honest to ensure security, 
even in the presence of collusion.

\section{PRELIMINARIES}
\subsection{Groups}

\textbf{Prime Order Group.} We use a prime-order group \texttt{GGen}, a probabilistic 
polynomial-time (PPT) algorithm that, given a security parameter \(\lambda\), 
outputs a description \(\mathcal{G} = (\mathbb{G}, p, g)\), where \(\mathbb{G}\) 
is an additive cyclic group of order \(p\), \(p\) is a \(2\lambda\)-bit prime, 
and \(g\) is the generator of \(\mathbb{G}\). For any \(a \in \mathbb{Z}_p\), 
given \(g^a\), it is computationally hard to determine \(a\), a problem known 
as the discrete logarithm problem. For any \(\vec{a}=\{a_1, \cdots, a_n\} \in \mathbb{Z}^n_p\), 
\ \( \vec{g} = \{g_1, \cdots, g_n\} \in \mathbb{G}^n\), 
we denote $\vec{g}^{\vec{a}} = \prod_{i\in[n]}g_i^{a_i}$.

\textbf{Pairing group.} We use a pairing group generator \texttt{PGGen}, a PPT algorithm that 
on input the security parameter $\lambda$ returns a a description
$\mathcal{PG}=\left(\mathbb{G}_1, \mathbb{G}_2, \mathbb{G}_T, p, g, h, e\right)$ 
of asymmetric pairing groups where $\mathbb{G}_1, \mathbb{G}_2, \mathbb{G}_T$ are 
additive cyclic groups of order $p$ for a $2 \lambda$-bit prime $p$, $g$ and $h$ 
are generators of $\mathbb{G}_1$ and $\mathbb{G}_2$, respectively, and 
$e: \mathbb{G}_1 \times \mathbb{G}_2 \rightarrow \mathbb{G}_T$ is an 
efficiently computable (non-degenerate) bilinear group elements. 
For $a, b \in \mathbb{Z}_p$, given $g^a, h^b$, 
one can efficiently compute $e(g, h)^{ab}$ using the pairing $e$.

\textbf{Class Group.} We use a Decisional Diffie-Hellman (DDH) group
with an easy DL subgroup\cite{DLeasy}, which can be instantiated from class groups 
of imaginary quadratic fields. Let GenClassGroup be a pair of algorithms ($\texttt{Gen}, \texttt{Solve}$).
The \texttt{Gen}\cite{TPECDSA} algorithm is a group generator 
which takes as inputs a security parameter $\lambda$ and a prime $p$ and outputs 
the parameters $\left(\tilde{s}, f, \hat{h}_p, \hat{G}, F,\hat{G}^p \right)$.
The set $(\hat{G}, \cdot)$ is a finite abelian group of order $p \cdot s$ 
where the bitsize of $s$ is a function of $\lambda$, with  $gcd(p, s) = 1$.
The algorithm \texttt{Gen} only outputs $\tilde{s}$, which is an upper bound of $s$. 
$\hat{G}^p=\{ x^p, x \in \hat{G}\}$ and 
$F$ is the subgroup of order $p$ of $\hat{G}$, so $\hat{G}= F \times \hat{G}^p$.
The algorithm also outputs $f, h_p $ and $\hat{h}=f \cdot \hat{h}_p$, 
which is the generators of $F, \hat{G}^p $ and $\hat{G}$. 
Moreover, the DL problem is easy in $F$, with an efficient 
algorithm \texttt{Solve}\cite{DLeasy}. According to the Hard subgroup membership (HSM)
assumption\cite{Bandwidth}, it is hard to distinguish random elements of 
$G^p$ in $G$.

\subsection{Zero Knowledge Proof}
	For a polynomial-time decidable relation $\mathcal{R}$, we call $w$ a witness for 
	a statement $u$ if $\mathcal{R}(u;w)=1$. Let $L=\{u | \exists w : \mathcal{R}(u;w)=1\}$
	is a language associated with $\mathcal{R}$. A zero-knowledge proof for $L$ 
	is a protocol between prover $\mathcal{P}$ and verifier $\mathcal{V}$ 
	where $\mathcal{P}$ convinces $\mathcal{V}$ that a common input $u \in L$ 
	without revealing information about a witness $w$. 
	A zero-knowledge proof should satisfy the following three properties: 
	soundness, completeness and perfect honest-verifier zero-knowledge.
	Soundness refers to the probability of $\mathcal{P}$ 
	convincing $\mathcal{V}$ to accept $u \in L$ is negligible when $u \notin L$. 
	Completeness means that if $u \in L$, $\mathcal{V}$ accept
	it with non-negligible probability. And the zero-knowledge implies that $\mathcal{V}$ 
	cannot revealing the information of $w$.
	We denote by
	$$
	\mathsf{PoK}\{(w): \varphi(w)\}
	$$
	a proof of knowledge of a secret value $w$ such that the public predicate 
	$\varphi(w)$ holds. 

	Notably, an interactive proof-of-knowledge protocol can be transformed into a 
	non-interactive proof-of-knowledge through the application of the Fiat-Shamir 
	heuristic \cite{FiatShamir}.

\subsection{Cross-Ciphertext Decentralized Verifiable Functional Encryption}
Verifiable functional encryption (VFE) \cite{VFE} was introduced to 
capture ciphertexts and decryption keys' correctness under fine-grained control. 
Informally, VFE requires that for any ciphertext $C$ that passes a public verification procedure, 
there exists a unique message $m$ such that, for any valid function description $f$ and 
any corresponding function key $dk_f$ that also pass public verification, 
the decryption algorithm outputs $f(m)$ when given $C$, $dk_f$, and $f$ as inputs.

Nguyen et al.\cite{VMCFE} proposed a decentralized multi-client verifiable functional encryption scheme.
The scheme focuses on verifying properties of individual ciphertexts. 
In our work, we propose a verifiable functional encryption scheme that enables cross-ciphertext verification, 
allowing the validation of specific relations among multiple ciphertexts in various application scenarios.
Next, we give the formal definition of our Cross-Ciphertext Decentralized Verifiable multi-client Functional Encryption for inner product (CC-DVFE).

\begin{definition}
	A Cross-Ciphertext Decentralized Verifiable Multi-Client Functional Encryption for Inner Product among 
	a set of $n$ clients
	$\left( \mathcal{C}_i\right)_{i \in [n]}$ with a $m$-dimensional vector on $\mathcal{M}^m$
	contains eight algorithms.
	\begin{itemize}
		\item \textbf{\texttt{SetUp}$(\lambda) \to pp$: }%\to \left(\mathcal{PG},\mathcal{DG},h_p,\mathcal{H}_1, \mathcal{H}_2,\mathcal{H}'_1 ,\mathcal{H}, \ell_D \right)$:
			Takes as input the security parameter $\lambda$ and outputs the public parameters $pp$. 
		\item \textbf{\texttt{KeyGen}($pp$) $\to \left((ek_i,sk_i)_{i \in [n]}, vk_{CT}, vk_{DK}, pk\right)$: }% \to \left((sk_i, ek_i)_{i \in [n]}, pk, vk_{CT}, vk_{DK}\right):$
			Each client $\mathcal{C}_i$ generates own encryption key $ek_i$ and secret key $sk_i$, and
			eventually outputs the verifiation key for ciphertexts $vk_{CT}$ and for functional key
			share $vk_{Dk}$, and the public key $pk$. 
		\item \textbf{\texttt{Encrypt}$\left(ek_i, \vec{x}_i, aux, \ell, (\ell_{Enc,j})_{j \in [m]}, pp\right)\to (C_{\ell, i},$\allowbreak$ \pi_{CT,i})$ :}
		Takes as input an encryption key $ek_i$, a vector $\vec{x}_i$ to  encrypt, 
		the auxiliary information $aux$ used in the ciphertext proof protocol, a label $\ell$ 
		and a set encryption labels $\left(\ell_{Enc,j}\right)_{j \in [m]}$ that ensures the 
		ciphertexts of different dimensions in 
		each round cannot be illegally aggregated.
		Outputs the ciphertext $C_{\ell,i}$ and a zero knowledge proof $\pi_{CT,i}$.
		\item \textbf{\texttt{VerifyCT}$\left((C_{\ell,i})_{i \in [n]},\pi_{CT}, vk_{CT}\right) \to 1 \text{ or } 0 $: }% \to 1/0 :$
		Takes as input the ciphertexts $(C_{\ell,i})_{i \in [n]}$ and proofs $\pi_{CT}$ from $n$ client  
		and the verification keys $vk_{CT}$.
		It outputs $1$ for accepting or $0$ with 
		a set of malicious clients $\mathcal{CC}_{CT} \ne \emptyset$ for rejecting.
		\item \textbf{\texttt{DKeyGenShare}$\left(sk_i, pk, y_i, \ell_{y}, pp \right) 
			\to \left(dk_{y,i}, \ \pi_{DK,i}\right)$: } 
		Takes as input the secret key $sk_i$ of a client, public key $pk$, 
		a value $y_i$ and a function label $\ell_{y}$.
		Outputs a functional decryption key share ${dk}_{y,i}$ and a zero knowledge proof $\pi_{DK,i}$. 
		\item \textbf{\texttt{VerifyDK}$\left(({dk}_{y,i})_{i\in [n]},\pi_{DK}, vk_{DK}, pp\right) \to  1 \text{ or } 0 $: }%\to 1/0 :$
		Takes as input functional decryption key  shares $({dk}_{y,i})_{i\in [n]}$ and proofs $\pi_{DK}$ from $n$ clients 
		and a verification key $vk_{DK}$. 
		It outputs $1$ for accepting or $0$ with 
		a set of malicious clients $\mathcal{CC}_{DK} \ne \emptyset$ for rejecting.
		\item \textbf{\texttt{DKeyComb}$\left(({dk}_{y,i})_{i\in [n]}, \ell_{y},pk \right) \to dk $:}
		Takes as input all of the functional decryption key shares $({dk}_{y,i})_{i \in [n]}$, 
		a function label $\ell_{y}$ and public key $pk$. 
		Outputs the functional decryption key $dk$.
		\item \textbf{\texttt{Decrypt}$(\left(C_{\ell,i}\right)_{i \in [n]} , {dk}, \vec{y}) \to( \langle \vec{x}^{\top}, \vec{y}\rangle)_{j \in [m]} \text{ or } \bot$: }
		Takes as input the m-dimensions vector ciphertext 
		$\left(C_{\ell,i}\right)_{i \in [n]}$ from $n$ client ,
		a functional decryption key ${dk}$ and the function vector $\vec{y}$.
		Outputs $(\langle \vec{x}^{\top}, \vec{y}\rangle)_{j \in [m]} $ or $\bot$,
		where $\vec{x}^{\top} = (x_{1,j}, x_{2,j}, \dots, x_{n,j})$ denotes the vector formed by the $j$-th component from each of the $n$ users.
	\end{itemize}
\end{definition}

\textbf{Correctness.} Given any set of message $\left(\vec{x}_{1}, \cdots, \vec{x}_{n}\right)\in \mathcal{M}^{n \times m}$
and the auxiliary information $aux$ %$\vec{x}_0 \in \mathcal{M}^m$ ,
and any function vector $\vec{y} \in \mathbb{Z}_p^n$,
We say that an CC-DVFE scheme is correct if and only if following equation holds.

\begin{equation*}
\begin{split}
\Pr\left[\begin{array}{l|l}
                                                                        & \texttt{SetUp}(\lambda)\rightarrow \mathrm{pp}\\
                                                                        & \texttt{KeyGen}(pp) \rightarrow((ek_i,sk_i)_{i \in [n]}, \\
																		& \quad vk_{CT}, vk_{DK}, pk),\\
\texttt{DKeyComb}(                                                      & \forall i \in [n],\ \texttt{Encrypt}(ek_i, \vec{x}_i,   \\
\quad ({dk}_{y,i})_{i\in [n]}, \ell_{y},  						            & \quad aux, \ell, (\ell_{Enc,j})_{j \in [m]}, pp) \\
\quad pk) \to dk 					                                    & \quad \rightarrow (C_{\ell, i}, \pi_{CT,i})\\
\texttt{Decrypt}((C_{\ell,i})_{i \in [n]},                              & \forall i \in [n],\ \texttt{DKeyGenShare}(   \\
\quad {dk}, \vec{y}) \to  	                                            & \quad sk_i, pk, y_i, \ell_{y}, pp)\\
\quad \left( \langle \vec{x}^{\top}, \vec{y}\rangle \right)_{j \in [m]} &\quad \rightarrow (dk_{y,i}, \pi_{\mathrm{DK},i})\\
                                                                        & \texttt{VerifyCT}((C_{\ell,i})_{i \in [n]},\pi_{CT}, \\
																		& \quad vk_{CT})  \to 1 \\
                                                                        & \texttt{VerifyDK}(({dk}_{y,i})_{i\in [n]},\pi_{DK},  \\ 
																		& \quad vk_{DK}, pp) \to 1\\
\end{array}
\right]=1.
\end{split}
\end{equation*}
% $$
% \left\{\begin{array}{l}
% \mathrm{pp} \leftarrow \texttt{SetUp}(\lambda) \\
% \left(\left(\mathrm{sk}_i, \mathrm{ek}_i\right)_{i \in[n]}, \mathrm{vk}_{\mathrm{CT}}, \mathrm{vk}_{\mathrm{DK}}, \mathrm{pk}\right) \leftarrow \texttt{KeyGen}(pp) \\
% \left(C_{\ell, i}, \pi_{CT, i}\right) \leftarrow \texttt{Encrypt}(ek_i, \vec{x}_i, \vec{x}_0, \ell, \left(\ell_{Enc,j}\right)_{j \in [m]}, pp), \forall i \in[n] \\
% \left(dk_{y,i}, \pi_{DK,i}\right) \leftarrow \texttt{DKeyGenShare}(sk_i, pk, y_i, \ell_{y}, pp), \forall i \in[n]
% \end{array}\right.
% $$
% then
% $$
% \left\{\begin{array}{l}
% \texttt{VerifyCT}((C_{\ell,i})_{i \in [n]}, vk_{CT})=1 \\
% \texttt{VerifyDK}((dk_{y,i})_{i\in [n]}, vk_{DK}, pp)=1 \\
% \texttt{Decrypt}(\left\{C_{\ell,i}\right\}_{i \in [n]} , h^{dk_{y}}, y)=\left\{ \sum_i x_{i,j} \cdot y_i \right\}_{j \in [m]}
% \end{array}\right.
% $$
% with probability 1.

\subsection{Our New Aggregation Rule}
Inspired by the aggregation rule of FLTrust\cite{FLtrust}, 
we similarly introduce root dataset $D_0$
to bootstrap trust to achieve Byzantine robustness against malicious clients
and consider both the direction and 
magnitudes of the local models of clients.
However, FLTrust contains cosine similarity operations that 
are difficult to complete the verification using a simple zero-knowledge proof,
which is avoided in our new aggregation rule.

A malicious client can arbitrarily manipulate the direction and magnitudes of 
the local model to degrade the accuracy of the global model, 
so our robust aggregation rule restricts both ways. 
First, we assume that the server has a small clean dataset $D_0$ and train on it to 
get a baseline model $W^t_0$.
Without considering the magnitudes of the model, we use the inner product of 
the local client model $W^t_i$ and the baseline model $W^t_0$ 
to measure the directional similarity between them. 
In order to avoid malicious users to enhance 
the attack effect of malicious vectors by amplifying the gradient, we incorporate 
the inner product value of the local client model $W^t_i$ itself into the 
robustness aggregation rule. We further consider that when the inner product value 
is negative, the model still has some negative impact on the global 
model of the aggregation, so we use the Relu function to further process the 
inner product value. The Relu function is defined as follows:
$$
	\text{ReLu}\left( x \right)=
	\begin{cases}
	x, & \text{if } x>0 \\
	0, & \text{others}
	\end{cases}.
$$
We define the robust aggregation rule as follows:
$$
	W^*=\sum_{i \in [n]} \text{Relu}\left(\frac{\langle W^t_i, W^t_0\rangle}
	{\langle W^t_i, W^t_i\rangle}\right) W^t_i
$$

At the same time, in order to avoid the global model unable to statute to a good level 
of accuracy, we normalize each epoch of global model to a certain 
magnitude. We use the $W_0^t$ as a reference to normalize each 
epoch of global model. That is, theglobal gradient $W^t$ of for the $t$-th iteration
is defined as follows:
$$
	W^t=\frac{||W^t_0||}{||W^*||}W^*
$$ 
Denote $y_i = \text{Relu}\left(\frac{\langle W^t_i, W^t_0\rangle}{\langle W^t_i, W^t_i\rangle}\right)$ and 
$W^*$ can be obtained from the functional encryption for inner product.
The robustness of this aggregation rule is proved in Section \ref{analysis}.

\section{OUR DECENTRALIZED VERIABLE FUNCTIONAL ENCRYPTION}
\subsection{Construction of CC-DVFE}
\begin{figure}[htbp]
	\centering
	\includegraphics[width=0.45\textwidth]{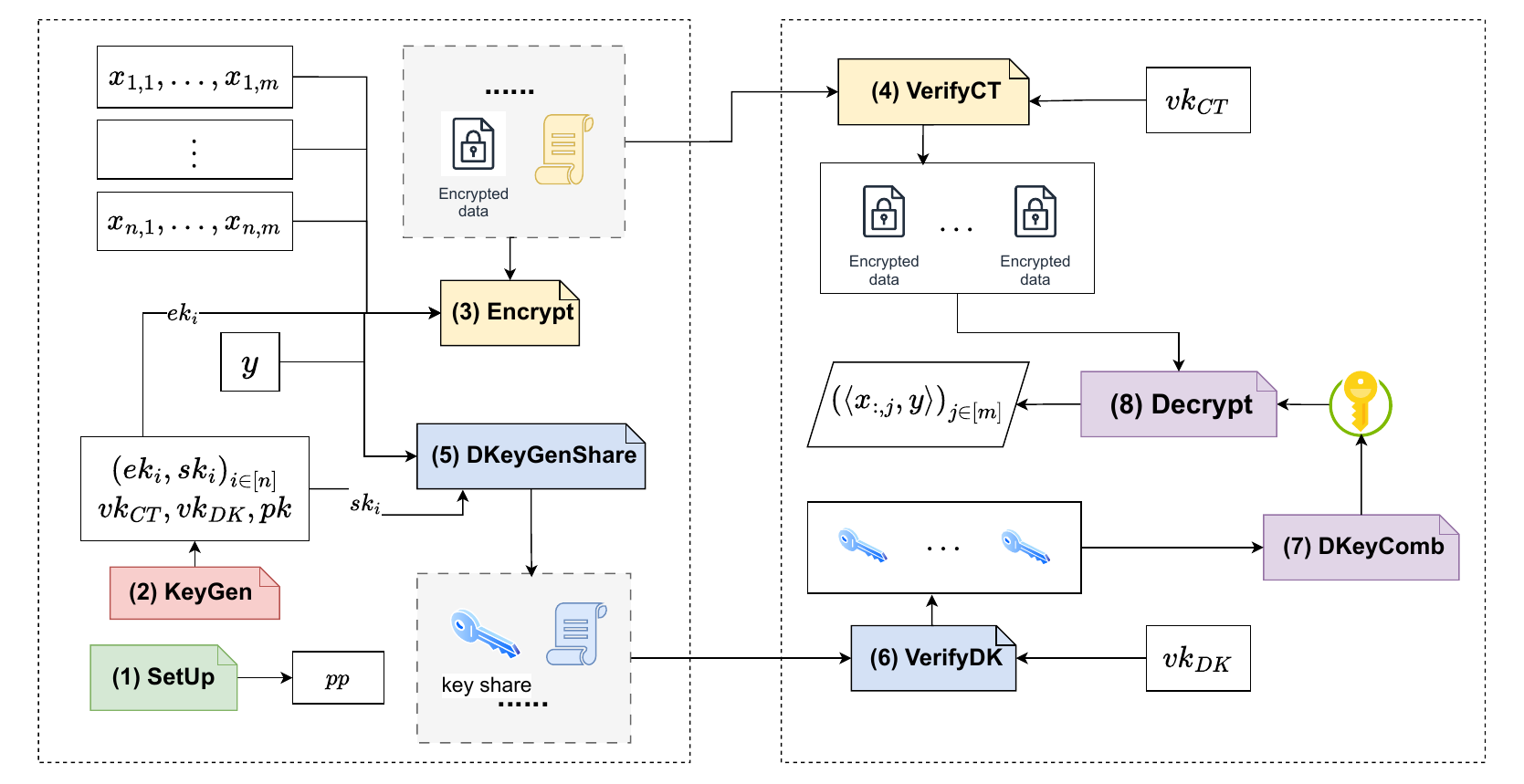}
	\caption{The workflow and main components of CC-DVFE}
	\label{VFEwork}
\end{figure}
Although verifiable functional encryption schemes supporting weight aggregation for 
encrypting local models are now readily available, their 
application to privacy-preserving and robust federated learning faces the 
following two challenges. 
The first one is that the current ciphertext verification in 
verifiable functional encryption supports range proofs for individual elements 
but cannot verify Byzantine robustness rules. Secondly, the construction of zero-knowledge proofs
for the Byzantine robust aggregation rule faces compatibility issues: 
existing inner-product arguments \cite{Bulletproofs} and functional encryption schemes cannot be directly 
combined for verification. Specifically, the current inner-product argument fails to 
achieve the zero-knowledge property, and the ciphertext structure in 
functional encryption does not support direct verification of 
Byzantine robust aggregation rules. 

To address these limitations, 
we propose a novel verifiable functional encryption scheme for 
the new robust aggregation rule.
Next we describe the complete DVFE scheme construction, with eight algorithms.
The flowchart of the algorithm DVFE is shown in Figure \ref{VFEwork}.

\begin{itemize}
	\item \textbf{\texttt{SetUp}($\lambda$) $\to pp$: }
	Takes as input the security parameter $\lambda$ and outputs the public parameters 
	$pp =(\mathcal{PG}, \mathcal{DG} ,h_p,\mathcal{H}_1, 
	\mathcal{H}_2,$ \allowbreak $\mathcal{H}'_1 ,\mathcal{H}, \ell_D)$,
	where  $\mathrm{\texttt{PGGen}}(\lambda) \to \mathcal{PG} = (\mathbb{G}_1, \mathbb{G}_2, \mathbb{G}_T, p, g, h, e)$ is a pairing group, 
	$g$ the generator of $\mathbb{G}_1$, $h$ the generator of $\mathbb{G}_2$,
	$\texttt{GenClassGroup}$ \allowbreak $(\lambda, p) \to \mathcal{DG} = (\tilde{s}, f, \hat{h}_p,$ \allowbreak $ \hat{G}, F,\hat{G}^p)$ is a DDH group with an easy DL subgroup, 
	$\hat{t} \gets \mathcal{D}_p$ is sampled, and $h_p = \hat{h}_p^{\hat{t}}$ is set.
	Four full-domain hash functions are defined as follows:
	\begin{align*}
		\mathcal{H}_{1}  &: \{0,1\}^{*} \rightarrow \mathbb{G}_{1}^{2}, 
		& \mathcal{H}_{2}  &: \{0,1\}^{*} \rightarrow \mathbb{G}_{2}^{2}, \\
		\mathcal{H}'_{1} &: \{0,1\}^{*} \rightarrow \mathbb{G}_{1}, 
		& \mathcal{H}      &: \{0,1\}^{*} \rightarrow \mathbb{Z}_p .
	\end{align*}
	And $\ell_D = (\{0,1\}^*)$ is the initialization label.
	\item \textbf{\texttt{KeyGen}($pp$) $\to \left((ek_i,sk_i)_{i \in [n]}, vk_{CT}, vk_{DK}, pk \right)$: }
		It is executed by each client $\mathcal{C}_i$ to generate own encryption key $ek_i$ and 
		secret key $sk_i$, and eventually outputs a verifiation key for ciphertexts $vk_{CT}$ 
		and functional key share $vk_{Dk}$, and a public key $pk$. 
		A client selects 
		$\hat{k}_i = (\hat{k}_{i,1}, \hat{k}_{i,2}) \overset{\$}{\gets} \mathbb{Z}_p^2$, $ t_i=(t_{i,1},t_{i,2}) \overset{\$}{\gets} \mathcal{D}^2_p $ and 
		$s_{i} = (s_{i,1},s_{i,2}) \overset{\$}{\gets} \mathbb{Z}_p^2 $,
		then computes $T_i = \left(h_p^{t_{i,b}}\right)_{b \in [2]}$ and  
		$$
			d_{i} = \left( 
					f^{\hat{k}_{i,b}} \cdot \left( \prod_{i<i^*}^n T_{i^*,b} \cdot 
							\prod_{i>i^*}^n T_{i^*,b}^{-1} \right)^{t_{i,b}}
					\right)_{b \in [2]}
		$$
		and commits $s_{i}$ as 
		$$
			\mathrm{com}_{i} = (v)^{s_{i}}
		$$
		where $ v = \mathcal{H}_1(\ell_{D}) \in \mathbb{G}_{1}^{2} $.
		For each client $\mathcal{C}_i$, encryption key is $ ek_i = s_{i} $,  
		secret key  is $sk_i=(s_{i}, \hat{k}_i,t_{i})$ , and 
		public key is $pk=\left(T_i, \ d_i\right)_{i \in [n]}$. The verification key for 
		ciphertexts is $vk_{CT} = \left(\mathrm{com}_{i}\right)_{i\in[n]}$, 
		while for functional keys it is 
		$vk_{DK} = \left(T_i,d_i, \mathrm{com}_{i}\right)_{i \in [n]}$.
	\item \textbf{\texttt{Encrypt}$(ek_i, \vec{x}_i, \vec{x}_0, \ell, (\ell_{Enc,j})_{j \in [m]}, pp)\to  (C_{\ell, i},$\allowbreak$ \pi_{CT,i})$ :}
		Takes as input an encryption key $ek_i$, a vector $\vec{x}_i$ to  encrypt, 
		a vector $\vec{x}_0$ used in the ciphertext proof protocol, a label $\ell$ 
		and a set encryption labels $\left(\ell_{Enc,j}\right)_{j \in [m]}$ that ensures the 
		ciphertexts of different dimensions in 
		each round cannot be illegally aggregated.
		The ciphertext is
		$C_{\ell,i}= \left( \left(u_j\right)^{s_i} \cdot (w_j)^{x_{i,j}} \right)_{j \in [m]}$ 
		where $u_j=\mathcal{H}_1(\ell||j) \in \mathbb{G}_{1}^{2}$, 
		$w_j = \mathcal{H}_1^{'}(\ell_{Enc,j}||j) \in \mathbb{G}_{1}$.
		It re-computes $\mathrm{com}_i$ and $V=\prod_{j \in [m]} C_{\ell,i,j}$, 
		then generates a proof $\pi_{CT,i}$ 
		for the relation $\mathcal{R}_{Encrypt}$
		on input $\left(V, \mathrm{com}_i, \ell, \ell_{Enc}, y_i, \vec{x}_0\right)$.
		The $\pi_{CT,i}$ is described as 
        \begin{equation*}
            \text{PoK}
			\left\{(s_i, \vec{x}_i): 
                \begin{aligned}
				V_i = (\bar{u})^{s_{i}} \cdot \vec{w}^{\vec{x}_i} \wedge \\
				y_i = \frac{\langle \vec{x}_i, \vec{x}_0\rangle}{\langle \vec{x}_i, \vec{x}_i\rangle} \wedge \\
				\mathrm{com}_i = (v)^{s_{i}}.
                \end{aligned}
			\right\}
        \end{equation*}
		where $\bar{u}=\prod_{j}^m u_j$, $u_j = \mathcal{H}_1(\ell||j)$, $\vec{w}^{\vec{x}_i} = \prod_{j \in [m]} (w_j)^{x_{i,j}} $.
		Outputs ciphertext $C_{\ell,i}$ and proof $\pi_{CT,i}$.
	\item \textbf{\texttt{VerifyCT}($(C_{\ell,i})_{i \in [n]},\pi_{CT}, vk_{CT}$) $\to 1 \text{ or } 0 $: }
		Takes as input the ciphertexts $(C_{\ell,i})_{i \in [n]}$ and proofs $\pi_{CT}$ from $n$ client  
		and the verification keys $vk_{CT}$.
		For $i \in [n]$: it verifies $\pi_{CT,i}$ for the relation 
		$\mathcal{R}_{Encrypt}$ and outputs $1$ for accepting or $0$ with 
		a set of malicious clients $\mathcal{CC}_{CT} \ne \emptyset$ for rejecting.
	\item \textbf{\texttt{DKeyGenShare}$(sk_i, pk, y_i, \ell_{y}, pp) \to 
		({dk}_{y,i}, \pi_{DK,i})$: }
		Takes as input a client secret key $sk_i = (s_{i}, \ \hat{k}_i,\ t_{i})$, 
		public key $pk$, the function coefficient $y_i$
		and a function  label $\ell_{y}$. It computes $dk_i$ as 
		$$
		dk_i = \left( (\hat{v}_b)^{\hat{k}_i} \cdot h^{s_{i,b} \cdot y_i}\right)_{b \in [2]}
		$$
		where $\hat{v}_b=\mathcal{H}_2(\ell_{y,b}) \in \mathbb{G}_{2}^{2}, b \in [2]$.
		It re-computes $\mathrm{com}_{i}$ and generates a proof $\pi_{DK,i}$
		for the relation $\mathcal{R}_{DKeyGenShare,i}$
		on input $\left(pk, com, {dk}_{y,i}, \ell_{y}\right)$. 
		The $\pi_{DK,i}$ is described as
        \begin{equation*}
            \text{PoK}
			\left\{(s_i, t_i,\hat{k}):
            \begin{aligned}
                T_i = h_p^{t_i} \wedge \\
				d_i = (f^{\hat{k}_b} K_{\sum,i,b}^{t_{i,b}})_{b \in [2]} \wedge \\
				dk_i = \left( (\hat{v}_b)^{\hat{k}} \cdot h^{s_{i,b} \cdot y_i}\right)_{b \in [2]} \wedge \\
				\mathrm{com}_i = (v)^{s_{i}}
            \end{aligned} 
			\right\}
        \end{equation*}  
		where $K_{\sum,i,b} = \prod_{i<i^*}^{n} T_{i^*,b} \cdot \left( \prod_{i>i^*}^n T_{i^*,b} \right)^{-1} \in G$, 
		$\hat{v}_b=\mathcal{H}_2(\ell_{y,b}), b \in [2], v= \mathcal{H}_1(\ell_{D})$.
		% It is detailed in section \ref{zkp_dkg}.
		Outputs a functional decryption key share ${dk}_{y,i}$ and $\pi_{DK,i}$.
	\item \textbf{\texttt{VerifyDK}($({dk}_{y,i})_{i\in [n]}, \pi_{DK}, vk_{DK}, pp$) $\to  1 \text{ or } 0$: }
		Takes as input decryption key  shares $({dk}_{y,i})_{i\in [n]}$,
		proof $\pi_{DK}$ from $n$ clients and the verification key $vk_{DK}$. 
		For $i \in [n]$: it verifies $\pi_{DK,i}$ for the relation 
		$\mathcal{R}_{DKeyGenShare}$ and outputs $1$ for accepting or $0$ with 
		a set of malicious clients $\mathcal{CC}_{DK} \ne \emptyset$ for rejecting.
	\item \textbf{\texttt{DKeyComb}($({dk}_{y,i})_{i\in [n]}, \ell_{y},pk$) 
		$\to {dk} $:}
		Takes as input all of the functional decryption key shares $({dk}_{y,i})_{i \in [n]}$, 
		a function label $\ell_{y}$ and public key $pk$. It computes $ f^d =\left(\prod_{i \in [n]} d_{i,b}\right)_{b \in [2]} $ and solves the discrete logarithm 
		problem on $\mathcal{DG}$ to obtain $d$, 
		then computes $h^{dk_{y}}$  as
		 $$ 
			 h^{dk_{y}} = \left(\frac{\prod_{i \in [n]} dk_{i,b}}{(\hat{v}_b)^d}\right)_{b \in [2]} 
		 $$
		 where $\hat{v}_b=\left(\mathcal{H}_2(\ell_{y,b})\right), b \in [2]$. 
		Outputs the functional decryption key $dk=h^{dk_{y}}$.
	\item \textbf{\texttt{Decrypt}($\left(C_{\ell,i}\right)_{i \in [n]} , dk, y$) 
	$\to \left(\langle \vec{x}^{\top}, \vec{y}\rangle\right)_{j \in [m]}$ \text{or} $\bot$ : }
		Takes as input $n$ client with m-dimensions vector ciphertext 
		$\left(C_{\ell,i}\right)_{i \in [n]}$,
		a functional decryption key $h^{dk_{y}}$ and the function vector $y$.
		It computes $u_j=\mathcal{H}_1(\ell||j)$, and 
		$\left(E_j\right)_{j \in [m]} = \left(e(w_j, h)\right)_{j \in [m]}$, 
		where $w_j = \mathcal{H}_1^{'}(\ell_{Enc,j}||j)$.
		For $ j \in [m]$, it gets 
		$$
			E_j^{\langle \vec{x}^{\top}, \vec{y}\rangle} = \frac{\prod_{i \in [n]}e(c_{l,i,j}, h^{y_i})}
				                            {e(u_{j,1}, h^{dk_{y,1}})e(u_{j,2}, h^{dk_{y,2}})} 
		$$
		and solves the discrete logarithm in basis $E_j$ to return 
		$\langle \vec{x}^{\top}, \vec{y} \rangle$, where $ \vec{x}^{\top} = (x_{1,j}, x_{2,j}, \dots, x_{n,j}) \in \mathcal{M}^n$ denotes the vector consisting of the $j$-th component of each of $n$ clients.
		It outputs $(\langle \vec{x}^{\top}, \vec{y}\rangle)_{j \in [m]} $ or $\bot$.
\end{itemize}

\textbf{Correctness}.We have
\begin{align*}
	f^{\sum_{i \in [n]} \hat{k}_i} & =\prod_{i \in [n]} d_{i} \\
	d &= \sum_{i \in [n]} \hat{k}_i = \log_f f^{\sum_{i \in [n]} \hat{k}_i} \\
	h^{dk_{y}} &= \left(\frac{\prod_{i \in [n]} dk_{i,b}}{(\hat{v}_b)^d}\right)_{b \in [2]} \\
	& = \left( 
		\frac{\prod_{i \in [n]} (\hat{v}_b)^{\hat{k}_i} \cdot h^{s_{i,b} \cdot y_i}}{(\hat{v}_b)^{\sum_{i \in [n]} \hat{k}_i}}
	\right)_{b \in [2]}
	 \\
	& = \left( h^{\sum_{i \in [n]}s_{i,b} \cdot y_i}\right)_{b \in [2]}
\end{align*} 
\begin{align*}
	(E_j)^{\langle \vec{x}^{\top}, y\rangle} &= \frac{\prod_{i \in [n]}e(c_{l,i,j}, h^{y_i})}{e(u_{j,1}, h^{dk_{y,1}})e(u_{j,2}, h^{dk_{y,2}})} \\
	& = \frac{\prod_{i \in [n]}e( u^{s_i} \cdot (w_j)^{x_{i,j}}, h^{y_i})}{e(u_{j,1}, h^{dk_{y,1}})e(u_{j,2}, h^{dk_{y,2}})} \\
	% & = \frac{\prod_{i \in [n]} \left( e(u_1,h)^{s_{i,1} \cdot y_i} e(u_{j,2},h)^{s_{i,2} \cdot y_i} e(w_j,h)^{x_{i,j} y_i}\right)}
	% 	{e(u_1,h)^{\sum_{i \in [n]} s_{i,1} \cdot y_i}e(u_{j,2},h)^{\sum_{i \in [n]} s_{i,2} \cdot y_i}}  \\
	& = e(w_j,h)^{\langle \vec{x}^{\top}, y\rangle}
\end{align*}
where $\vec{x}^{\top} = (x_{1,j}, x_{2,j}, \dots, x_{n,j}) \in \mathcal{M}^n$ denote the vector composed of the $j$-th 
parameter of the local models from each of the $n$ clients,
and then solves the discrete logarithm in basis $E_j$ to return 
$\langle \vec{x}^{\top}, y\rangle$.
Performing $m$ decryption operations on 
$\{\left\{c_{\ell,i,j}\right\}_{i \in [n]}\}_{j \in [m]}$
returns the result $\{ \langle \vec{x}^{\top}, y\rangle \}_{j \in [m]}$.

\begin{theorem}\label{theorem:ind}
	The decentralized verifiable  multi-client functional encryption for inner product is static-IND-secure 
	under the DDH, multi-DDH and HSM assumptions, 
	as Definition \ref{definition:ind}.
	More precisely, we have
	\begin{align*}
		\mathrm{Adv_{CC-DVFE}^{sta-ind}}(t, q_E, q_K) \le \  
		&q_E \mathrm{Adv}^\mathrm{zk}_{\prod_{Encrypt}}(t) &\\
        &+ q_K \mathrm{Adv}^\mathrm{zk}_{\prod_{DKeyGenShare}}(t) &\\
		&+ \mathrm{Adv^{std-ind}_{LDSUM}}(t, q_K) &\\
		&  + 2q_E(2\mathrm{Adv}^{ddh}_{\mathbb{G}_1}(t) + \frac{1}{p}) &\\
		&+ 2\mathrm{Adv}_{\mathbb{G}_1}^{ddh}(t + 4 q_E \times t_{\mathbb{G}_1}) &
	\end{align*}
	where
	\begin{itemize}
		\item[-] $\mathrm{Adv_{CC-DVFE}^{sta-ind}}(t, q_E, q_K)$ is the best advantage of any PPT adversary running in time $t$
 			with $q_E$ encryption queries and $q_K$ key share queries against the static-IND-secure game of the CC-DVFE scheme;
		\item[-] $\mathrm{Adv}^\mathrm{zk}_{\prod_{Encrypt}}(t)$ is the best advantage of any PPT adversary running in time $t$
 			against the zero-knowledge property of the $\prod_{Encrypt}$ scheme;
		\item[-] $\mathrm{Adv}^\mathrm{zk}_{\prod_{DKeyGenShare}}(t)$ is the best advantage of any PPT adversary running in time $t$
 			against the zero-knowledge property of the $\prod_{DKeyGenShare}$ scheme;
		\item[-] $\mathrm{Adv^{std-ind}_{LDSUM}}(t, q_K)$ is the best advantage of any PPT adversary running in time $t$
 		  with $q_K$ encryption queries against the IND-security game of the LDSUM scheme\cite{VMCFE}.
		\item[-] $\mathrm{Adv}^{ddh}_{\mathbb{G}_1}(t)$ is the best advantage of 
			breaking the DDH assumption in $\mathbb{G}_1$ within time $t$.
	\end{itemize}
\end{theorem}
\noindent
	\begin{proof}
		It is worth noting that, to achieve a non-negligible advantage in winning the game, 
		the adversary $\mathcal{A}$ must leave at least two honest clients, each possessing 
		two distinct messages $(X_0, X_1)$ for encryption queries.
		We define games as follows: 
	
		\setlength{\parskip}{1em}
		\noindent
		\textbf{Game $G_0$:} The challenger plays as this game described in the Definition \ref{definition:ind}, with the set of 
		$\mathcal{CC}$ of the corrupted clients are established before initialization. 
		Note that the hash function $\mathcal{H}_1, \mathcal{H}_2, \mathcal{H}'_1,\mathcal{H}$ 
		is modeled as a random oracle RO onto 
		$\mathbb{G}_1^2, \mathbb{G}_2^2, \mathbb{G}_1, \mathbb{Z}_p$.
	
		\noindent
		\textbf{Game $G_1$:} This game is as $G_0$, except that for generating proofs: 
		$\pi_{CT}$ and $\pi_{DK}$ in the QEncrypt and QDKeyGenShare.
		\begin{itemize}
			\item[-]
			$\pi_\mathrm{CT}$ is produced by a simulator $\mathrm{Sim}_\mathrm{Encrypt}$ on the input
			$(\ell, \left\{\ell_\mathrm{Enc}\right\}_{j \in [m]}, C_{\ell,i}, \vec{x}_0, \mathrm{com}_i)$ if $\mathcal{C}_i$ is queried in 
			QEncrypt for the input 
			$(i, X^0, X^1 \in [0, 2^q-1]^{n \times m}, \ell,\left\{\ell_\mathrm{Enc}\right\}_{j \in [m]})$.
			\item[-]
			$\pi_{\mathrm{DK}}$ is produced by a simulator $\mathrm{Sim}_\mathrm{DKeyGenShare}$ on the input
			$\left(pk, \ell_{y}, dk_i, \mathrm{com}_i\right)$ if $\mathcal{C}_i$ is queried in 
			QDKeyGenShare for the input$\left(i, y \in [0, 2^q-1]^n, \ell_{y}\right)$.
		\end{itemize}
		Given $q_E$ encryption queries and $q_K$ functional key share queries, 
		with the zero-knowledge property of NIZK proofs, we have
        \begin{align*}
            & \left|\mathrm{Adv}_{G_0} - \mathrm{Adv}_{G_1}  \right| \\ 
            \le & \
			q_E \mathrm{Adv^{zk}_{\prod_{Encrypt}}}(t) + q_K \mathrm{Adv^{zk}_{\prod_{DKeyGenShare}}}(t)
        \end{align*}
		%---密钥的
		\noindent
		\textbf{Game $G_2$:} This game is as $G_1$, except for functional key share queries 
		$\mathrm{QDKeyGen}(i,y)$: 
		\begin{itemize}
			\item[-] if $i$ is a corrupted index, then the challenger computes  
				$$
					dk_{i,j} = \left( 
						(\hat{v}_b)^{\hat{k}_i} \cdot h^{s_{i,b}\cdot y_i} 
						\right)_{b \in [2]}
				$$
			\item[-] if $i$ is the last non-corrupted index, the challenger computes
				$$
					dk_i = \left( 
						(\hat{v}_b)^{\hat{k}_i} \cdot h^{\sum_{t \in \mathcal{HC}}s_{t,1} \cdot y_t}
					\right)_{b \in [2]}
				$$	
			\item[-] for other non-corrupted index $i$, the challenger computes
			$$
				dk_i = \left( (\hat{v}_b)^{\hat{k}_i} \cdot h^{0},
				\right)_{b \in [2]}
			$$
		\end{itemize}
		Given $q_K$ key share queries, by the (static) IND-security of the LDSUM scheme in 
		\cite{VMCFE}, we have 
		$\left|\mathrm{Adv}_{G_1} - \mathrm{Adv}_{G_2} \right| 
		\le \mathrm{Adv^{std-ind}_{LDSUM}}(t, q_K)$ under the HSM assumptions.
	
		\noindent
		\textbf{Game $G_3$:} We simulate the answers to $\mathcal{H}_1$ 
		to any new RO-query of $\mathbb{G}_1^2$
		by a truly random pair in $\mathbb{G}_1^2$ on the fly. 
		Namely, $(g^{a}, g^{b})$, $a, b, \overset{\$}{\gets} \mathbb{Z}_p$.
		The simulation remains perfect, 
		and so $\mathrm{Adv}_2 = \mathrm{Adv}_3$.
	
		\noindent
		\textbf{Game $G_4$:}
		We simulate the answers to $\mathcal{H}'_1$ to any new RO-query of $\mathbb{G}_1$ by a 
		truly random element of $w_j=g^{\hat{w}_j}$, with $\hat{w}_j \overset{\$}{\gets} \mathbb{Z}_p$. 
		$\mathrm{Adv}_3 = \mathrm{Adv}_4$.
	
		\noindent
		\textbf{Game $G_5$:}  We simulate the answers to $\mathcal{H}_1$ to any new RO-query of $\mathbb{G}^2_1$ 
		by a truly random pair in the span of $g^\mathbf{a}$ 
		for $\mathbf{a}:=\left(1,\ a\right)^T$, with $a \overset{\$}{\gets} \mathbb{Z}_p$.
		In other words, 
		$u = g^{\hat{u}} = g^{\mathbf{a} r_\ell } = \left(g^{r_\ell}, g^{a r_\ell}\right),
		r_{\ell} \overset{\$}{\gets} \mathrm{RF}_{\mathbb{Z}_p}(\ell)$.
		From the following Lemma \ref{lemma4_5}, 
		$\left|\mathrm{Adv}_4 - \mathrm{Adv}_5 \right|\le 
		2\mathrm{Adv}_{\mathbb{G}_1}^{ddh}(t + 4q_E \times t_{\mathbb{G}_1})$.
	
		\noindent
		\textbf{Game $G_6$:} We simulate any QEncrypt query as the encryption of $X_i^0$ instead of 
		$X_i^b$ and go back for the answers to $\mathcal{H}_1$ to any new RO-query by a truly random pair in 
		$\mathbb{G}^2_1$ and to $\mathcal{H}'_1$ to a truly random element in $\mathbb{G}_1$.
		From the following Lemma \ref{lemma5_6}, we have
		$\mathrm{Adv}_{5}-\mathrm{Adv}_{6} \le 2q_E(2\mathrm{Adv}^{ddh}_{\mathbb{G}}(t) + \frac{1}{p})$. 
	
		\noindent
		In addition, it is clear that in the Game $G_6$ the advantage of 
		any adversary is exactly $0$ since $b$ does not appear anywhere. In other words,
		$\mathrm{Adv}_6 = 0$, which concludes the proof.
	\end{proof}
	
	\begin{lemma}\label{lemma4_5}
		If the Multi-DDH assumption holds, then no polynomial-time adversary 
		can distinguish between Game $G_4$ and Game $G_5$ with a non-negligible advantage.
	\end{lemma}
	
	\noindent
	\begin{proof}
		Given a multi-DDH instance 
		$$\left((g^x,(g^{y_q}, g^{z_q})_{q \in [q_E]}) \mid x, y_q \overset{\$}{\gets} \mathbb{Z}_p\right),$$ 
		the challenger flips 
		an unbiased coin with $\{0,1\}$ and obtains a bit $\mu\in\{0,1\}$.
		if $\mu = 0, z_q = x y_q$; if $\mu = 1, z_q \overset{\$}{\gets} \mathbb{Z}_p$.
		Suppose there exists an PPT adversary $\mathcal{A}$ who can distinguishes 
		$$A = \left\{ g^a, (g^{r_q}, g^{a r_q})_{q \in [q_E]} \mid r_q \overset{\$}{\gets} \mathrm{RF}_{\mathbb{Z}_p}(\ell_q),a \overset{\$}{\gets} \mathbb{Z}_p, q \in [q_E] \right\}$$
            and
        $$B = \left\{ g^a, (g^b, g^c)_{q \in [q_E]} \mid a,b,c \overset{\$}{\gets} \mathbb{Z}_p, q \in [q_E] \right\}$$
		by a advantage $\mathrm{Adv}(t)$, 
		where $\mathrm{Adv}(t)$ denotes the advantage of an adversary in 
		distinguishing $G_4$ from $G_5$ within time $t$.
		There exists a simulator $\mathcal{B}$ who can use 
		$\mathcal{A}$ to break  the multi-DDH assumption as follows.
		$\mathcal{B}$ sets $C=(g^x,(g^{y_q}, g^{z_q})_{q \in [q_E]})$.
		$\mathcal{B}$ sends $C$ to $\mathcal{A}$ and make the adversary guess that C is equal to A or B. 
		If $\mathcal{A}$ the believes that $C=A$ outputs $\omega=0$, $C=B$ outputs $\omega=1$.
		$\mathcal{B}$ outputs the its guess $\mu'$ on $\mu$: $\mu'= \omega$.
		It follows from Definition \ref{def_multiDDH} that 
		the total advantage with which $\mathcal{B}$ can break the 
		multi-DDH assumption in $\mathbb{G}_1$ is
		$\mathrm{Adv}_{\mathbb{G}_1}^{ddh}(t + 4q_E \times t_{\mathbb{G}_1})$.
		Therefore,
		\begin{align*}
			    & \mathrm{Adv}_{\mathbb{G}_1}^{ddh}(t + 4q_E \times t_{\mathbb{G}_1}) \\
			=   & |\operatorname{Pr}\left[\mu^{\prime}=\mu \mid \mu=0\right] \times \operatorname{Pr}[\mu=0] \\
				& - \operatorname{Pr}\left[\mu^{\prime}=\mu \mid \mu=1\right] \times \operatorname{Pr}[\mu=1]|\\
			\ge &  (\frac{1}{2}+\mathrm{Adv}(t)) \times \frac{1}{2} - \frac{1}{2} \times \frac{1}{2} \\
			=   & \frac{1}{2} \mathrm{Adv}(t)
		\end{align*}
		Therefore, $\left|\mathrm{Adv}_4 - \mathrm{Adv}_5 \right|\le 
		2\mathrm{Adv}_{\mathbb{G}_1}^{ddh}(t + 4q_E \times t_{\mathbb{G}_1})$.
	\end{proof}
	
	\begin{lemma}\label{lemma5_6}
		If the DDH assumption holds, then no polynomial-time adversary 
		can distinguish between Game $G_5$ and Game $G_6$ with a non-negligible advantage.
	\end{lemma}
	
	\begin{proof}
		The gap between $G_5$ and $G_6$ will be proven 
		using a hybrid technique on the RO-queries. Accordingly, we index the following games 
		by $q$, where $q = 1, \cdots, q_E$.
		\begin{itemize}
			\item $G_{5.1.1}$: This is exactly game $G_5$. Thus, $\mathrm{Adv}_5 = \mathrm{Adv}_{5.1.1}$
			\item $G_{5.q.2}$: We start by modifying the distribution of the output of 
				the \( q \)-th RO-query in \( \mathbb{G}^2_1 \), changing it from uniformly 
				random in the span of \( g^{\mathbf{a}} \) to uniformly random 
				over \( \mathbb{G}^2_1 \). Then, we switch to 
				$\mathbf{u} = u_1 \cdot (\hat{a})+u_2 \cdot (\hat{a})^{\bot}$ where 
				$u_1 \overset{\$}{\gets} \mathbb{Z}_p, \ u_2 \overset{\$}{\gets} \mathbb{Z}^*_p$
				and $\mathbf{a}:=\left(1,\ a\right)^T, \
				(\hat{a})^{\bot}:=\left(-a,\ 1\right)^T $, 
				which only changes the adversary view by a statistical distance of $\frac{1}{p}$.
				From the following Lemmas \ref{lemma1_2}, 
				we can get $\left|\mathrm{Adv}_{5.q.1} - \mathrm{Adv}_{5.q.2} \right| 
					\le 2\mathrm{Adv}^{ddh}_{\mathbb{G}_1}(t) + \frac{1}{p}$.
			\item $G_{5.q.3}:$ We now change the ciphertext response to encryption queries, 
			shifting $\{c_{\ell,i,j}\}_{j \in [m]}= 
			\{(u_j)^{s_i} \cdot (w_j)^{X^b_{i,j}}\}_{j \in [m]}$
			to $\{c_{\ell,i,j}\}_{j \in [m]}= \{(u_j)^{s_i} \cdot (w_j)^{X^0_{i,j}}\}_{j \in [m]}$ 
			where $u_j$ corresponds to the $q$-th RO-query of $\mathbb{G}^2_1$ and 
			$w_j$ corresponds to the $q$-th RO-query of $\mathbb{G}_1$.
			From the following Lemmas \ref{lemma2_3}, we can get 
			$\mathrm{Adv}_{5.q.2} = \mathrm{Adv}_{5.q.3}$.
			\item $G_{5.q+1.1}:$ This transition is the reverse of that 
				$G_{5.q.1} $ to $ G_{5.q.2}$, namely, we switch back 
				the distribution of $u$ computed on the $q$-th RO-query from uniformly 
				random over $\mathbb{G}^2$ (conditioned on the fact that 
				$\hat{u}(\hat{a})^{\bot} \ne 0$) 
				to uniformly random in the span of $g^\mathbf{a}$. Thus, we can get 
				$\left|\mathrm{Adv}_{5.q.3}-\mathrm{Adv}_{5.q+1.1} \right|\le 
				2\mathrm{Adv}^{ddh}_{\mathbb{G}}(t) + \frac{1}{p}$.
		\end{itemize}
		Therefore, $\left| \mathrm{Adv}_{5}-\mathrm{Adv}_{6} \right| \le 
		2q_E(2\mathrm{Adv}^{ddh}_{\mathbb{G}}(t) + \frac{1}{p})$. 
	\end{proof}
	
	\begin{lemma}\label{lemma1_2} 
		For any PPT adversary, let
			$\mathrm{Adv}(t)$ denote its advantage in distinguishing the distributions
			$$\mathcal{D}_0 = \{ g, g^r, g^{a r} \mid a \overset{\$}{\gets} \mathbb{Z}_p,\ r = RF_{\mathbb{Z}_p}(\ell) \}$$
			and
			$$\mathcal{D}_1 = \{ g, g^{u_1 - u_2 a}, g^{u_1 a + u_2} \mid a, u_1, u_2 \overset{\$}{\gets} \mathbb{Z}_p \}.$$
			Let $\mathrm{Adv}_{\mathbb{G}_1}^{ddh}(t)$ denote the advantage in solving the DDH problem in $\mathbb{G}_1$.
			Then it holds that $\mathrm{Adv}(t) \le 2\mathrm{Adv}_{\mathbb{G}_1}^{ddh}(t) $.
	\end{lemma}
	\begin{proof}
		Given a DDH instance  
		$(g^x,g^y,g^z \mid x, y\overset{\$}{\gets} \mathbb{Z}_p)$, 
		the challenger flips an unbiased coin with $\{0,1\}$ and obtains a bit $\mu \overset{\$}{\gets} \{0,1\}$.
		if $\mu = 0, z=xy$; if $\mu = 1, z \overset{\$}{\gets} \mathbb{Z}_p$.
		Suppose there exists an PPT adversary $\mathcal{A}$ who can distinguishes 
		$A = \{g, g^r,g^{ar} \mid a \overset{\$}{\gets} \mathbb{Z}_p, r = RF_{\mathbb{Z}_p}(\ell)\}$
		and 
		$B = \{g, g^{u_1-u_2 a},g^{u_1a+u_2} \mid a, u_1,u_2 \overset{\$}{\gets} \mathbb{Z}_p\}$
		within time $t$ by a advantage $\mathrm{Adv}(t)$, there exists a simulator $\mathcal{B}$ who can use 
		$\mathcal{A}$ to break  the DDH assumption as follows. \\
		 $\mathcal{B}$ sets $C=(X, Y, Z)$.
		$\mathcal{B}$ sends $C$ to $\mathcal{A}$ and make the adversary guess that $C$
		is equal to $A$ or $B$. 
		If $\mathcal{A}$ believes that $C=A$ outputs $\omega=0$; $C=B$ outputs $\omega=1$.
		$\mathcal{B}$ outputs $\mu'= \omega$.
		Let $\mathrm{Adv}_{\mathbb{G}_1}^{ddh}(t)$ denote the best advantage of 
		breaking the DDH assumption in $\mathbb{G}_1$ within time $t$.
		Therefore, 
		\begin{align*}
			& \mathrm{Adv}_{\mathbb{G}_1}^{ddh}(t) \\
            =   &\left|\operatorname{Pr}\left[\mu^{\prime}=\mu \middle| \mu=0\right] \operatorname{Pr}[\mu=0]-\operatorname{Pr}\left[\mu^{\prime}=\mu \middle| \mu=1\right] \operatorname{Pr}[\mu=1]\right|\\
			\ge &  (\frac{1}{2}+\mathrm{Adv}(t)) \times \frac{1}{2} - \frac{1}{2} \times \frac{1}{2} \\
			 =  &\frac{1}{2} \mathrm{Adv}(t)
		\end{align*}
		Therefore, $\mathrm{Adv}(t) \le 2 \mathrm{Adv}_{\mathbb{G}_1}^{ddh}(t)$.
	\end{proof}

	\begin{lemma}\label{lemma2_3}
		In the adversary's point of view, game $G_{5.q.2}$ is the same as game $G_{5.q.3}$, 
		thus, $\mathrm{Adv}_{5.q.2} = \mathrm{Adv}_{5.q.3}$, for $q=1,\cdots, q_E$.
	\end{lemma}
	
	\begin{proof}	
		We then prove this does not change the 
		adversarys view from $G_{5.q.2}$ to $G_{5.q.3}$. 
		Obviously, if the output of the $q$-th RO-query is not used by QEncrypt-queries, 
		then the games $G_{5.q.2}$ and $G_{5.q.3}$ are identical. However, 
		we can also demonstrate that this holds when the 
		$q$-th RO-queries are genuinely involved in QEncrypt-queries.
		The prove can be divided into two step.
		\begin{itemize}
			\item[-] \textbf{Step 1:} We show that there exists a PPT adversary $\mathcal{B^*}$ 
			such that, for $t = 2,3$,
			$\mathrm{Adv}_{5.q.t} = \left(p^{2m} + 1\right)^{-n} \mathrm{Adv}^*_{5.q.t}(\mathcal{B^*})$, 
			holds. Here, the games $G^*_{5.q.2}$ and $G^*_{5.q.3}$ are 
			selective variants of games $G_{5.q.2}$ and $G_{6.q.3}$, respectively.
			% and the set of indices submitted to QCorrupt queries are known for the challenger
			% before the initialization phase. 
			Fisrtly, the
			adversary $\mathcal{B^*}$ guesses for all
			$i\in [n], Z_i \overset{\$}{\gets} \mathbb{Z}^{2 \times m}_p \cup \{\bot\}$, 
			which it sends to its selective game $G^*_{5.q.t}$.
			Specifically, each guess $Z_i$ is either a pair of vector $(X_i^0, X_i^1)$ 
			queried to QEncrypt, or $\bot$, indicating no query was made to QEncrypt. 
			Then, it simulates $\mathcal{A}$'s view using its own oracles.
			When $\mathcal{B}^*$ guesses correctly, denoted by event $E$, 
			it simulates $\mathcal{A}$'s view exactly as in $G_{5.q.t}$.
			If the guess is wrong, then $\mathcal{B}$ stops the simulation and 
			outputs a random bit $\beta$. Since event $E$ happens with probability
			$(p^{2m} + 1)^{-n}$ and 
			is independent of the view of adversary $\mathcal{A}$, 
			we can get that $\mathrm{Adv}^*_{5.q.t}(\mathcal{B^*})$ is equal to 
			\begin{align*}
				& | \operatorname{Pr}\left[G_{5.q.t}^{\star} \middle| b=0, E\right] 
					\cdot \operatorname{Pr}[E]+
					\frac{\operatorname{Pr}[\neg E]}{2} \\
				&	- \operatorname{Pr}\left[G_{5.q.t}^{\star} \middle| b=1, E\right]
					\cdot \operatorname{Pr}[E] +
					\frac{\operatorname{Pr}[\neg E]}{2}
				 | \\
				= &\operatorname{Pr}[E]\cdot
				| 
					\operatorname{Pr}\left[G_{5.q \cdot t}^{\star} \middle| b=0, E\right] - \operatorname{Pr}\left[G_{5.q \cdot t}^{\star} \middle| b=1, E\right]
				| \\
				 = &\left(p^{2m} + 1\right)^{-n} \cdot \operatorname{Adv}_{5.q.t}.
			\end{align*}

		%  \begin{align*}
		% 		& |
		% 			\operatorname{Pr}[b=0]\left(\operatorname{Pr}\left[G_{5.q.t}^{\star} \middle| b=0, E\right] 
		% 			\cdot \operatorname{Pr}[E]+
		% 			\frac{\operatorname{Pr}[\neg E]}{2}\right) \\
		% 		&	-\operatorname{Pr}[b=1]\left(\operatorname{Pr}\left[G_{5.q.t}^{\star} \middle| b=1, E\right]
		% 			\cdot \operatorname{Pr}[E] +
		% 			\frac{\operatorname{Pr}[\neg E]}{2}\right)
		% 		 | \\
		% 		= &\operatorname{Pr}[E]\cdot
		% 		|
		% 			\operatorname{Pr}[b=0] \cdot 
		% 			\operatorname{Pr}\left[G_{5.q \cdot t}^{\star} \middle| b=0, E\right] \\
		% 		&	- \operatorname{Pr}[b=1] \cdot
		% 			\operatorname{Pr}\left[G_{5.q \cdot t}^{\star} \middle| b=1, E\right]
		% 		| \\
		% 		 = &\left(p^{2m} + 1\right)^{-n} \cdot \operatorname{Adv}_{5.q.t}.
		% 	\end{align*}
			\item[-] \textbf{Step 2:} We assume the matrix $(Z_i)_{i\in[n]}$ sent by 
			$\mathcal{B}^*$ are consistent. That is, Specifically, it does not cause the game to 
			terminate prematurely and 
			Finalize on $\beta$ does not return a random bit independent of $b'$, 
			conditioned on $E'$. Now, we demonstrate that the games $G^*_{5.q.2}$ and $G^*_{5.q.3}$ 
			are identically distributed. To prove this, we leverage the fact that the 
			following two distributions are identical, for any choice of $\gamma_j$:
			$$
				\left\{\left( s_i\right)_{i \in [n], Z_{i,j}=(X_{i,j}^0,X_{i,j}^1)}\right\}_{j \in [m]}
			$$
            and 
            $$
                \left\{\left(s_i + (\hat{a})^\bot \cdot \gamma_j (X_{i,j}^b-X_{i,j}^0)\right)_{i \in [n], Z_{i,j}=(X_{i,j}^0,X_{i,j}^1)}\right\}_{j \in [m]}
			$$
			where $(\hat{a})^{\bot}:= (-a,1)^T \in \mathbb{Z}^2_p$ and $s_i\gets \mathbb{Z}^2_p$, 
			for all $i = 1,\cdots, n$. They are identical because the \( s_i \) values are 
			independent of the \( Z_{i,j} \), a property that holds due to the selective nature 
			of our setting. This independence may not necessarily apply in the case of adaptive 
			QEncrypt queries.
			Thus, we can re-write $s_i$ into 
			$s_i + (\hat{a})^\bot \cdot \gamma_j (X_{i,j}^b-X_{i,j}^0)$ without 
			changing the distribution of the game. \\
			We now show the additional terms 
			$\Delta_{i,j} = (\hat{a})^\bot \cdot \gamma_j (X_{i,j}^b-X_{i,j}^0)$ 
			manifest in the adversary's view.
			\begin{enumerate}
				\item They do not appear in the output of QCorrupt.
				\item They might appear in $\text{QDKeyGen}(i,y)$. 
				In practice, if different clients have the same function for each 
				dimension aggregation, only one key can be generated, and if each 
				dimension aggregation function is different, $m$ keys can be generated.
				For generality, we consider the case of m keys.
				\begin{itemize}
				\item[-] if $i$ is a corrupted index, then the challenger computes  
					$
						\{dk_{i,j}\}_{j \in [m]} = \{( 
							(\hat{v}_b)^{\hat{k}_i} \cdot h^{(s_{i,b}+\Delta_{i,j,b})\cdot y_i}
						)_{b \in [2]}\}_{j \in [m]}.
					$
					As $i\in \mathcal{CC}, Z_i\ne\bot, X_i^1 = X_i^0$,  
					we can get that $\left\{\Delta_{i,j} = 0^2\right\}_{j \in [m]}$. 
					Therefore in this case, the decryption key remains the same upon change.
				\item[-] if $i$ is the last honest index, for $j \in [m]$, the challenger computes
					\small
					$$
						dk_{i,j} = ( 
							(\hat{v}_b)^{\hat{k}_i} \cdot h^{\sum_{t \in \mathcal{HC}}(s_{t,b}+\Delta_{t,j,b}) \cdot y_t}
						)
					$$
					\normalsize
					By the constraints for $E'$,  as $i\in \mathcal{HC}$, $Z_i\ne\bot$
					and $\{\langle X^{\top 0}_j, \vec{y} \rangle = \langle X^{\top 1}_j, \vec{y} \rangle\}_{j \in [m]}$,
					we can get that,  for $j \in [m]$, 
					$$
						\sum_{i \in [n]:Z_{i,j}=(X_{i,j}^0,X_{i,j}^1)} y_i(X_{i,j}^b - X_{i,j}^0)  = 0.
					$$
					Consider the constraints on queries to corrupt clients, as $i\in \mathcal{CC}$ and $X^0_i = X^1_i$,
					for $j \in [m]$,
					$$\sum_{i \in \mathcal{CC}:Z_{i,j}=(X_{i,j}^0, X_{i,j}^1)} y_i(X_{i,j}^b - X_{i,j}^0)  = 0 $$ 
						is hold,
					and we can get that, for $j \in [m]$, 
					$$\sum_{i \in \mathcal{HC}}\Delta_{i,j} \cdot y_i = \textbf{0}.$$ 
					Thus, for $j \in [m]$,
					$$\sum_{t \in \mathcal{HC}}(s_{t}+\Delta_{t,j}) \cdot y_t = 
					\sum_{t \in \mathcal{HC}} s_{t} \cdot y_t.$$
					Therefore in this case, the decryption key remains the same upon change.
					\item[-] for other honest index $i$, the challenger computes
					$$
						dk_i = \left( (\hat{v}_b)^{\hat{k}_i} \cdot h^{0}\right)_{b \in [2]}
					$$
					This case does not involve the encryption keys.
				\end{itemize}
				So in $\text{QDKeyGen}(i,y)$, it's indistinguishable before and after the transformation.
				\item They appear in the output of the QEncrypt-queries 
				which use $u_j$ computed on the q-th RO-query. Since for all others, 
				the vector $u_j$ lies in the span of $g^{\hat{a}}$, i.e., $u_j=g^{\hat{u}_j}, w_j = g^{\hat{w}_j}$
				and $(\hat{a})^{\top}(\hat{a})^{\bot} = 0$. We thus have, for $j \in [m]$,
				$$
					C_{i,j} = (u_j)^{s_i +(\hat{a})^\bot \cdot \gamma_j (X_{i,j}^b-X_{i,j}^0)}
					(w_j)^{X_{i,j}^b}.
				$$
				Since $\hat{u}_j (\hat{a})^{\bot} \ne 0$, we choose 
				$\gamma_j = \frac{-\hat{w}_j}{\hat{u}_j (\hat{a})^{\bot}} \mod{p}$
				and then $C_{i,j}= (u_j)^{s_i}(w_j)^{X_{i,j}^0}$. 
				We stress that $\gamma_j$ is independent of the index $i$, and so this 
				simultaneously converts all the encryptions of $X_i^b$ into encryptions 
				of $X_i^0$. Finally, reverting these statistically  perfect changes, 
				we obtain that $C_i$ is identically distributed to 
				$\{(u_j)^{s_i}(w_j)^{X_{i,j}^0}\}_{j \in [m]}$, 
				as in game $G^*_{5.q.3}$.
			\end{enumerate}
			Thus, when event $E'$ happens, the games are identically distributed. 
			When $E$ happens, the games both return 
			$\beta \overset{\$}{\gets} \{0, 1\}$: 
			$\mathrm{Adv}^*_{5.q.2}(\mathcal{B^*}) = \mathrm{Adv}^*_{5.q.3}(\mathcal{B^*})$.
		\end{itemize}
		As a conclusion, the change from $G_{5.q.2}$ to $G_{5.q.3}$ is imperceptible 
		from the point of view of the adversary $\mathcal{A}$. 
		Therefore, $\mathrm{Adv}_{5.q.2} = \mathrm{Adv}_{5.q.3}$, for $q=1,\cdots, q_E$.
	\end{proof}

\begin{theorem}\label{theorem:ver}
	The decentralized verifiable multi-client functional encryption for inner product supports 
	verifiability in the random oracle, as Definition \ref{definition:ver}.
	More precisely, we have
	\begin{align*}
		& \mathrm{Adv}_{\text{CC-DVFE}}^{\text{verif}}\left(t, q_C, q_F\right) \le 
		\epsilon_1 + \\
		&  q_C \cdot \max \left\{
		\mathrm{Adv}_{\prod_{Encrypt}}^{\text{snd}}(t),
		q_F \cdot \mathrm{Adv}_{\prod_{\text{DKeyGenShare}}}^{\text{snd}}(t)
		\right\}
	\end{align*}
	where
	\begin{itemize}
		\item[-] $\mathrm{Adv}_{\text{CC-DVFE}}^{\text{verif}}\left(t, q_C, q_F\right)$ 
				is the best advantage of any PPT adversary running in time t against 
				the verifiability game in Definition \ref{definition:ver} with $q_C$ corruption queries and 
				$q_F$ functions for the finalization phase
		\item[-] $\mathrm{Adv}_{\prod_{Encrypt}}^{\text{snd}}(t)$ is the best advantage of any PPT adversary running in time $t$
			against the soundness of $\prod_{\text{Encrypt}}$.
		\item[-] $\mathrm{Adv}_{\prod_{\text{DKeyGenShare}}}^{\text{snd}}(t)$ is the best advantage of any PPT adversary running in time $t$
			against the soundness of $\prod_{\text{DKeyGenShare}}$.
	\end{itemize}
	% for any adversary $\mathcal{A}$, running within time $t$, where $q_C$ is the number of 
	% corruption queries , $q_K$ is the number of function vectors for the KeyGenShare queries, 
	% $\epsilon_1$ is negligible, $\mathrm{Adv}_{\prod_{Encrypt}}^{\text{snd}}(t)$ 
	% is the best advantage of any PPT adversary running in time $t$ aqainst the soundness
	% of $\prod_{Encrypt}$,
	% and $\mathrm{Adv}_{\prod_{\text{DKeyGenShare}}}^{\text{snd}}(t)$ is that 
	% against the soundness of $\prod_{\text{DKeyGenShare}}$.
\end{theorem}
\begin{proof}
		If we suppose that there exists a PPT adversary $\mathcal{A}$ 
		that can win the verifiability game in 
		Definition \ref{definition:ver} 
		with a non-negligible probability. 
		As \textbf{Finalize}, it can be divided into the following two conditions
		\begin{itemize}
			\item If $\mathcal{A}$ wins the game by winning the first condition 
			with an advantage of $\epsilon_1$. In order to accuse an honest sender $\mathcal{C}_i$, 
			$\mathcal{A}$ has to broadcast some malicious share 
			that makes the proof of correct generation for ciphertext or functional key share
			of $\mathcal{C}_i$ invalid. 
			By the design of the scheme, the only used elements 
			that belongs to another client  are $\left\{T_k\right\}_{k \in [n]}$ 
			in the public keys, which are not used in the 
			relation $\mathcal{R}_{Encrypt}$. 
			As the only requirement is that $\{T_k\}_{k \in [n]}$ lie 
			in the class group $\hat{G}$, submitting malformed group encodings 
			can compromise the generation and verification processes of honest clients.
			% The only  requires $\left\{T_k\right\}_{k \in [n]}$ to be group 
			% elements in class group $\hat{G}$.
			% Therefore, sending an incorrect group-encoding $\left\{T_k\right\}_{k \in [n], k \ne i}$ 
			% can make the generation and then the proof of a honest client fail. 
			However, this is a trivial attack and 
			can be excluded, as each $\left\{T_k\right\}_{k \in [n], k \ne i}$ can be 
			efficiently verified to be in the class group $\hat{G}$  in the 
			$\prod_{\text{DKeyGenShare}}$ and the public will already 
			know it is the corrupted client who broadcast a malicious share. 
			Thus, the $\epsilon_1$ is negligible.
			\item If $\mathcal{A}$ wins the game by winning the second condition. 
			In this case, we have
			\[\texttt{VerifyCT}\left((C_{\ell,i})_{i \in [n]},\pi_{CT, i}, vk_{CT}\right) = 1\]
			and 
			\small
			\[\texttt{VerifyDK}\left((dk_{y,i})_{i\in [n]}, \pi_{DK,i}, vk_{DK}, pp\right) = 1.\]
			\normalsize
			This means the transcript output by $\mathcal{A}$ satisfies 
			the relations  $\mathcal{R}_{Encrypt}$ and $\mathcal{R}_{DKeyGenShare}$.
			From $\mathcal{R}_{\text{DKeyGenShare}}$, the tuple 
			$( \mathrm{vk}_{\mathrm{CT}}, \mathrm{vk}_{\mathrm{DK}}, \mathrm{pk})$ 
			is generated by the \texttt{KeyGen} algorithm, where the secret keys are given by 
			$sk_i = (s_{i}, \hat{k}_i, t_{i})$, and the encryption keys are defined 
			as $ek_i = s_i$. For each $i \in [n]$ and for each inner product function $\boldsymbol{y}$, 
			the decryption key $dk_{y,i}$ is generated by the \texttt{DKeyGenShare} algorithm, 
			taking $\left(sk_i, \ell_{y}, pk\right)$ as input.
			From $\mathcal{R}_{\text{Encrypt}}$, the ciphertext $C_{\ell, i}$ and the encryption 
			proof $\pi_{\text{Encrypt}, i}$ are generated by the \texttt{Encrypt} algorithm, using as 
			input the message $\vec{x}_i \in \mathcal{M}$ and the encryption key $s'_i$ for each $i \in [n]$.
			We model the hash function $\mathcal{H}_1$ as a random oracle onto $\mathbb{G}^2_1$. 
			Then $\mathrm{com}_i$ is perfectly binding. From above, we have $s_i = s'_i$. By the proved 
			correctness of the scheme, the decryption process with input 
			$C_\ell = \left\{C_{\ell,i}\right\}_{i \in [n]}$ and 
			$h^{dk_{y}} = \text{DKeyComb}\left(\{dk_{y,i}\}_{i \in [n]}, \ell_{y},pk\right)$ 
			will output the inner product $\left\{ \langle \vec{x}^{\top}, y\rangle \right\}_{j \in [m]}$ 
			for all vectors $y$, then $\mathcal{A}$ fails in the game.
			So $\mathcal{A}$ must either constructs a malicious ciphertext or 
			a malicious functional key share.
			If $\mathcal{A}$ wins the game with a 
			non-negligible probability, there exists a simulator 
			$\mathcal{B}$ who can use $\mathcal{A}$ to 
			breaking the soundness $\prod_{Encrypt}$ or $\prod_{DKeyGenShare}$.
			It can be constructed as follows:
			\begin{itemize}
				\item[-] Assuming that $\mathcal{A}$ wins by 
				constructing a malicious ciphertext with the proof.
					Given $q_C$ corrupted clients, 
					$\mathcal{B}$ plays as a challenger in the game with $\mathcal{A}$, 
					after $\mathcal{A}$ finalized the game, $\mathcal{B}$ guesses an index $i$ 
					from the corrupted set $\mathcal{MC}_\mathcal{A}$
					and outputs the instance 
					$\left(\ell, C_{\ell,i}, \mathrm{com}_i, \pi_{i,\text{Encrypt}}\right)$ from $\mathcal{A}'$s 
					transcript. In the case $\mathcal{A}$ wins the game, 
					the probability that $\mathcal{B}$ breaks the soundness of 
					$\prod_{Encrypt}$ is  $\frac{1}{q_C}$.
				\item[-] Assuming that $\mathcal{A}$ wins by 
				constructing a malicious functional key share with the proof.
					Given $q_C$ corrupted clients and $q_F$ inner-product functions $y$ to be finalized, 
					$\mathcal{B}$ outputs one in $q_C \cdot q_F$ instances of 
					$\left((T_i)_{i\in[n]}, d_{i}, \mathrm{com}_i, dk_{i, y}\right)$ from $\mathcal{A}$, 
					which incurs a security loss of $q_C \cdot q_F$.
			\end{itemize}
			To finalize, we have
			% $$
			% \operatorname{Adv}_{\mathrm{CC-DVFE}}^{\text{verif}}\left(\mathcal{A}, t, q_C, q_F\right) 
			% \leq \epsilon_1 +
			% q_C \cdot \max \left\{\operatorname{Adv}_{\prod_{\text{Encrypt}}}^\text{snd}(t), \
			% q_F \cdot \operatorname{Adv}_{\prod_{\text{DKeyGenShare}}}^\text{snd}(t)\right\}
			% $$
            \begin{equation*}
                \begin{aligned}
                &\mathrm{Adv}_{\text{CC-DVFE}}^{\text{verif}}\left(t, q_C, q_F\right) \le 
                \epsilon_1 + \\
                &q_C \cdot \max \left\{
                \mathrm{Adv}_{\prod_{Encrypt}}^{\text{snd}}(t),
                q_F \cdot \mathrm{Adv}_{\prod_{\text{DKeyGenShare}}}^{\text{snd}}(t)
                \right\}.
                \end{aligned}
            \end{equation*}
			As $\epsilon_1$, $\operatorname{Adv}_{\prod_{\text{Encrypt}}}^\text{snd}(t)$ and 
			$\operatorname{Adv}_{\prod_{\text{DKeyGenShare}}}^\text{snd}(t)$ are negligible, 
			and $q_C$ and $q_F$ are polynomially bounded, the proof is complete.
		\end{itemize}
	\end{proof}
	
The zero-knowledge proof algorithms related to \texttt{Encrypt} and \texttt{VerifyCT} 
are instantiated in Appendix \ref{zkp_enc}, 
while those related to \texttt{DKeyGenShare} and \texttt{VerifyDK} 
are instantiated in Appendix \ref{zkp_dkg}.
The above two zero-knowledge proofs are transformed into non-interacting 
zero-knowledge proofs using the Fiat-shamir transform \cite{Fiat-shamir}, 
which can be applied to the CC-DVFE scheme. 

\section{VFEFL DESIGN} 
\subsection{High-Level Description of VFEFL}
\begin{figure*}[t]
	\centering
	\includegraphics[width=0.95\textwidth]{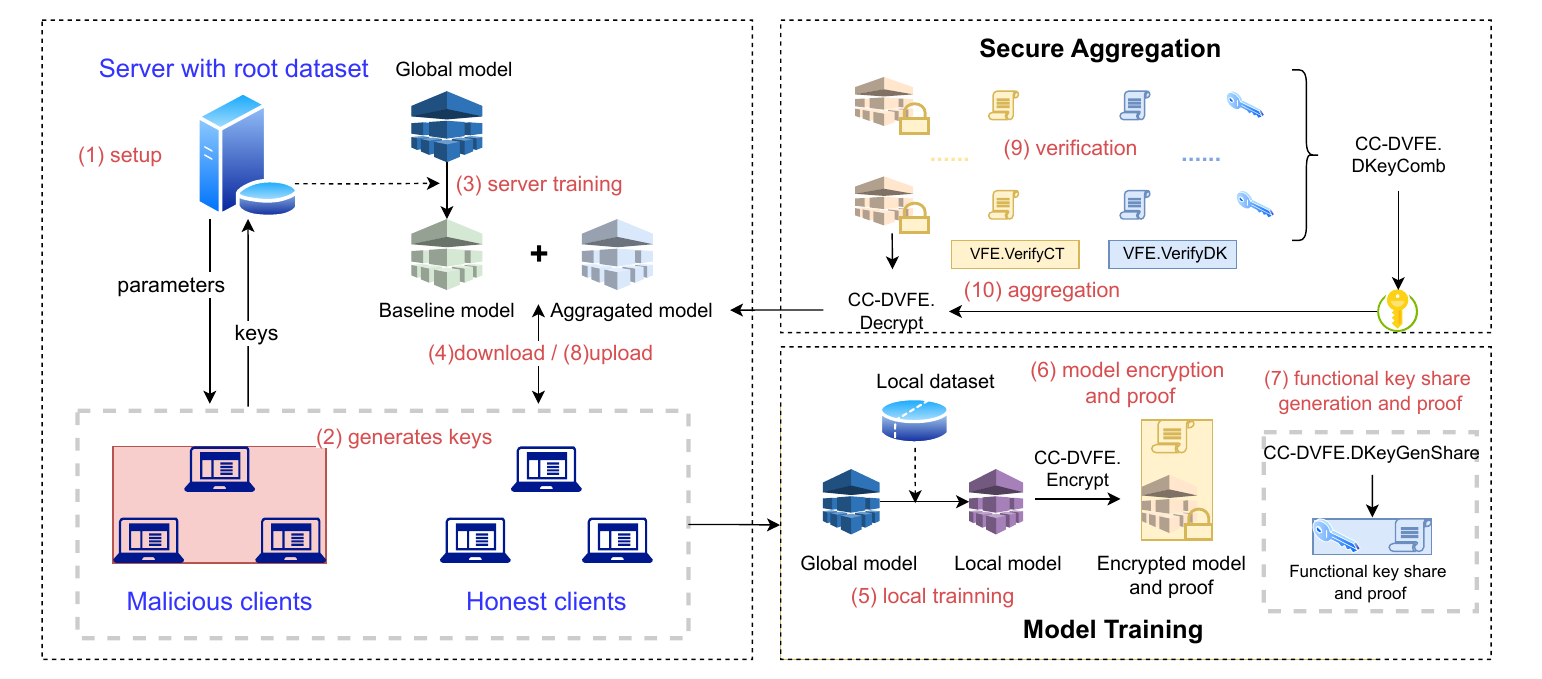}
	\caption{The workflow and main components of VFEFL}
	\label{Framework}
\end{figure*}
The overview of VFEFL is depicted in the Figure \ref{Framework}. VFEFL utilizes the 
CC-DVFE to achieve a privacy and Byzantine FL among Server $\mathcal{S}$ and 
multiple clients $\left\{\mathcal{C}_i\right\}_{i \in [n]}$. 
Specifically, VFEFL consists of three phases.
\textbf{1) Setup  Phase}: $\mathcal{S}$ initializes the the public parameters $pp$, 
the set of clients $\mathcal{CS}$, and the global initial model $W_0$, which
are sent to $\mathcal{CS}$ (Step (1)). 
Each client independently generates its keys, broadcasts the corresponding 
the public keys and verification keys (Step (2)). \textbf{2) Model Training Phase}: 
In the $t$-th training epoch, denoting the last global model as $W^{t-1}$,
$\mathcal{S}$ updates the baseline model $W^t_0$ based on $D_0$ (Step (3)) and 
broadcasts it. 
Each participant downloads the $W^t_0$ and $W^{t-1}$ (Step (4)),
trains a local model $W^t_i$ from $W^{t-1}$
based on the local dataset $D_i$ (Step (5)), encrypts it (Step (6)), and 
generates the functional key share (Step (7)) along with corresponding proofs, which are 
then sent to $\mathcal{S}$ (Step (8)).
\textbf{3) Secure Aggregation Phase}: $\mathcal{S}$ collects and verifies all 
encrypted models $\left(E_i\right)_{i \in \mathcal{CS}}$ and functional key shares 
with the proofs (Step (9)).
If all verifications pass, the process moves to the next step;
otherwise, the server $\mathcal{S}$ requires some clients to resend the ciphertext 
or key share based on the verification results and re-verifies them. 
$\mathcal{S}$ aggregates the authenticated functional key share, decrypts the ciphertexts, and 
normalizes the decryption results to get the final global update $W^t$ (Step (10)) and 
send it to the $\mathcal{CS}$.

\subsection{Setup  Phase}
The server $\mathcal{S}$ collects all participating clients to initialize the client 
set $\mathcal{CS}$ and randomly sets the initial model parameters $W_0$. Then, 
the server runs the \texttt{CC-DVFE.SetUp} algorithm to generate the public parameters $pp$ 
and broadcasts $\left(\mathcal{CS}, W_0, pp\right)$. Each client $\mathcal{C}_i$ independently 
executes \texttt{CC-DVFE.KeyGen}($pp$) to generate its keys pair
$\left(ek_i = s_i, sk_i = (s_i, \hat{k}_i, t_i), T_i \right)$,
broadcasts $T_i$,
and eventually generate the public key $pk = (T_i, d_i)_{i \in [n]}$ verification keys  $\left(vk_{CT,i}=\mathrm{com}_i, \ 
vk_{CT,i}=(T_i, d_i, \mathrm{com}_i)\right)$ according to $\mathcal{CS}$.
Each client sends the verification keys $vk_{CT,i}$ and $vk_{DK,i}$ to $\mathcal{S}$.

\subsection{Model Training Phase}
Model training consists of three parts: local training, model encryption and proof 
and functional key share generation and proof.
\begin{enumerate}
	\item \textbf{Local Training:} In the $t$-th training iteration, after retrieving the global 
		model $W^{t-1}$ from the server, each client $\mathcal{C}_i$ optimizes the local 
		model using its private dataset $D_i$ by minimizing the corresponding loss 
		function $\text{Loss}$ to obtain the updated model $W_i^t$
		$$
			W_i^t = \operatorname*{arg \min}_{W^t_i} \text{Loss}(W^t_i, D_i)
		$$ 
		For benign clients, the local model is updated using stochastic gradient descent 
		(SGD) for \( R_l \) iterations with a local learning rate \( lr \) as:
		$$
			w^r \gets w^{r-1} - lr \, \nabla \text{Loss}(D'_b; w)
		$$
		where $r=1,..,R_l$, $D'_b$ is a randomly sampled mini-batch from the 
		local dataset \( D_i \), $W_i^t=\vec{w}^{R_l}$.
		For malicious clients, they perhaps chooses a random vector as a model, 
		or perhaps crafts it so that the global model can deviate maximally from the 
		right direction.
		Similarly, $\mathcal{S}$ trains the baseline model $W^t_0$ for this epoch $t$ 
		based on the root dataset $D_0$ and broadcast $W^t_0$ to the $\mathcal{C}$.
	\item \textbf{Model Encryption and Proof:} To prevent leakage of local models \( W_i^t \)
	 	to the server $\mathcal{S}$, $\mathcal{C}_i$  encrypts and generates a proof of 
		it before uploading. 
		Specifically, for an honest client, $\mathcal{C}_i$ applies the 
		\texttt{CC-DVFE.Encrypt}($ek_i, W^t_i, W^t_0, \ell, \left(\ell_{Enc,j}\right)_{j \in [m]}, pp$) 
		algorithm for encryption and outputs $E_{\ell,i} $ and $\pi_{CT, i}$. Then 
		$\mathcal{C}_i$ sends the encrypted model $E_{\ell,i}$ and the ciphertext proof
		$\pi_{CT, i}$ to the $\mathcal{S}$.
		For malicious clients, they may not use the correct encryption key or the correct 
		encryption algorithm for encryption in order to disrupt the aggregation process.
	\item \textbf{Functional Key Share Generation and Proof:} To ensure that malicious clients generate 
		functional key shares honestly, we require them to provide corresponding 
		zero-knowledge proofs during this process. For an honest client, 
		$\mathcal{C}_i$ takes $y_i$ from the previous step and processes it 
		using the ReLU function, resulting in $y'_i = \text{ReLU}(y_i)$.
		$\mathcal{C}_i$ executes \texttt{CC-DVFE.DKeyGenShare}($sk_i, pk, y'_i, \ell_y , pp$) to 
		generate the corresponding decryption key share $dk_i$ and 
		the proof $\pi_{DK,i}$ for it.
		For a malicious client, $\mathcal{C}_i$ may try to use an invalid secret key or 
		provide a malformed decryption key share, aiming to compromise the 
		decryption result.
\end{enumerate}
\subsection{Secure Aggregation Phase}
	The secure aggregation phase consists of three stages:
	ciphertext verification , functional key share verification and aggregation.
	\begin{enumerate}
	\item \textbf{Ciphertext Verification:} %To counter Byzantine attacks and the use of incorrect ciphertext, 
		In this phase, the server $\mathcal{S}$  verifies 
		all encrypted models $\left\{E_{\ell,i}\right\}_{i \in \mathcal{C}}$ from clients  using 
		\texttt{CC-DVFE.VerifyCT}($\{C_{\ell,i}\}_{i \in \mathcal{C}},\pi_{CT, i}, vk_{CT}$). If all 
		ciphertexts are valid, the process continues. Otherwise, the server 
		identifies malicious clients $\mathcal{CC}_{CT}$, and 
		requires all the clients in $\mathcal{CC}_{CT}$ to regenerate the ciphertext.		
	\item \textbf{Functional Key Share Verification:} 
		% To prevent a malicious client from 
		% destroying the decryption by generating an incorrect functional key share. 
		In this phase, the server $\mathcal{S}$ collects and verifies all functional key shares from clients
		by executing \texttt{CC-DVFE.VerifyDK}
		($\{dk_{y,i}\}_{i\in \mathcal{C}}, \pi_{DK,i},vk_{DK}, pp$).
		If all ciphertexts pass the verification, the process proceeds to the next 
		phase. Otherwise, the set of malicious clients $\mathcal{CC}_{DK}$
		whose key shares failed verification is identified and $\mathcal{S}$ 
		requires all clients in $\mathcal{CC}_{DK}$ to regenerate the functional key share.
	\item \textbf{Aggregation:} 
		In this phase, the server $\mathcal{S}$ performs $\texttt{CC-DVFE.DKeyComb}
		\left(\{dk_{y,i}\}_{i \in \mathcal{C}}, \ell_{y}, pk\right)$
		to aggregate the verified functional keys shares and obtain the decryption key $h^{dk_{y}}$, 
		and \texttt{CC-DVFE.Decrypt}($\left\{E_{\ell,i}\right\}_{i \in \mathcal{C}} , h^{dk_{y}}, y$) 
		is executed to decrypt the verified models to gain $W^*$.
		Then $\mathcal{S}$ further normalizes 
		$W^*$ to obtain $W^t$ according to the baseline model $W^t_0$. 
		Finally, $\mathcal{S}$  obtains the global 
		model $W^t$ for the epoch $t$ and send it to the $\mathcal{CS}$ for the next
		training epoch utill the global model converges.
	\end{enumerate}

	VFEFL repeats the model training and secure aggregation phases until the global model 
	converges, resulting in the final global model.

\section{ANALYSIS} \label{analysis}
\subsection{Privacy}\label{analysis_privacy}
\begin{theorem}
	According to the security model of VFEFL, we say this VFEFL is secure if any view 
	$\textbf{REAL}_{\mathcal{A}}^{VFEFL}$  of $\mathcal{A}$ in the real world is computationally 
	indistinguishable from \(\textbf{IDEAL}_{\mathcal{A*}}^{VFEFL}\) of  $\mathcal{A}^*$
	in the ideal world, namely
	$$
		\textbf{REAL}_{\mathcal{A}}^{\mathrm{VFEFL}}
				(\lambda, \hat{x}_\mathcal{S},\hat{x}_\mathcal{C}, \hat{y}) 
		\stackrel{c}{\approx} 
		\textbf{IDEAL}_{\mathcal{A}^*}^{f_{\mathrm{VFEFL}}}
				(\lambda, \hat{x}_\mathcal{S},\hat{x}_\mathcal{C}, \hat{y}),
	$$
	where $\hat{y}$ denotes intermediate results, 
	$\hat{x}_\mathcal{S},\hat{x}_\mathcal{C}$ denote the 
	inputs of $\mathcal{S}$ and $\mathcal{C}$.
\end{theorem}
\begin{proof}
	To prove the security of VFEFL, it suffices to demonstrate the confidentiality 
	of the privacy-preserving defense strategy, as it is the only component involving 
	the computation of private data of clients, for the unauthorized entities $\mathcal{S}^*$
	(the server or eavesdrops).	For the curious entities $\mathcal{S}^*$, 
	\( \textbf{REAL}^{\mathrm{VFEFL}}_\mathcal{A} \) contains 
	intermediate parameters and encrypted local models \( \{ E^t_i \}_{i=[n]} \)
	with corresponding proofs $\{\pi_{CT, i}\}_{i=[n]}$, collected during the execution of the VFEFL protocols.
	Additionally, we construct a PPT simulator $\mathcal{A}^*$ to execute 
	\( f_{\mathrm{VFEFL}} \), 
	simulating each step of the privacy-preserving defensive aggregation strategy.
	The detailed proof is provided below.

	\textbf{Hyb1.} We initialize a set of system parameters based on distributions 
	indistinguishable from those in \( \textbf{REAL}^{\mathrm{VFEFL}}_\mathcal{A} \) 
	during the actual protocol execution, including public parameters $pp$, 
	initial models $W_0$, and the keys.

	\textbf{Hyb2.} In this hybrid, we alter the behavior of the simulated client 
	\( \mathcal{C}^*_i\), % \in \mathcal{C}^* \), 
	having it use the selected random vector
	 \( \Theta^t_i \) as its local model \( W^t_i \). As only the 
	original contents of ciphertexts have changed, the static-IND security property of VFEFL 
	guarantees that $\mathcal{S}^*$ cannot distinguish the view of $\Theta^t_i$ 
	from the view of original \( W^t_i \). And in $\pi_{Encrypt, i}$, 
	the $y_i$ also be sent to the server. 
	Existing reconstruction attacks \cite{xu2022agic} cannot recover the model vector with limited 
	information of $\|W^t_i\|$. $y_i$ can be viewed as a variant of $\|W^t_i\|$ and
	$cosine(W^t_i,W^t_0)$ (which is also revealed in \cite{BSRFL})
	when baseline models $W^t_0$ are known. Thus, It is a negligible security threat 
	to $W^t_i$.

	\textbf{Hyb3.} In this hybrid, the server uses zero-knowledge $\pi_{Encrypt}$ of 
	encrypted random vector $\Theta^t_i$ instead of blinded models \( W^t_i \) to 
	verify the ciphertexts to obtain the result of $1$ or $0$, 
	without knowing the further information of inputs. Hence, this hybrid 
	is indistinguishable from the previous one.

	\textbf{Hyb4.} In this hybrid, the decryption key is computed by \texttt{CC-DVFE.DKeyComb}
		and the aggregated model $W^{*}$(before normalization) 
		is computed by the \texttt{CC-DVFE.Decrypt} algorithm. $\mathcal{S}$ can 
		obtain $\Theta^*= \sum_{i \in \mathcal{C}}\left(y_i \Theta^{t}_i\right)$.
		In addition, the local model $\Theta^{t}_i$ is encrypted while the randomness 
		of these functional key shares is still preserved. This randomness ensures that 
		the data is not compromised during the aggregation and decryption process. 
		The probability of $\mathcal{S}$ guessing the specific $\Theta^t_i$ correctly is 
		negligible. Hence, this hybrid is indistinguishable from the previous one.
	
	In \textbf{Hyb4}, ${\mathcal{A}^*}$ holds the view $\textbf{IDEAL}_{\mathcal{A}^*}^{f_{\mathrm{VFEFL}}}$,
	which is converted from \( \textbf{REAL}^{\mathrm{VFEFL}}_\mathcal{A} \) and 
	is indistinguishable from it.
	Thus, the above analysis demonstrates that the output of 
	\( \textbf{IDEAL}_{\mathcal{A}^*}^{f_{\mathrm{VFEFL}}} \) is computationally indistinguishable 
	from the output of \( \textbf{REAL}_{\mathcal{A}}^{\mathrm{VFEFL}} \). 
	Therefore, we conclude the proof of the privacy of VFEFL.
\end{proof}

\subsection{Robustness}
We provide provable guarantees on the Byzantine robustness of the proposed defensive 
aggregation. We show that under the following assumptions, the difference between the global 
model $W^t$ learned under attack by the aggregation condition and the optimal global model 
$W^{opt}$ is bounded. Next, we first describe the needed assumptions and then describe 
our theoretical results.

\textbf{Assumption 1. }
	The expected loss function \( F(w) \) is \( \mu \)-strongly convex and differentiable 
	over the space \( \Theta \) with an \( L \)-Lipschitz continuous gradient. Formally, 
	for any \( w, \hat{w} \in \Theta \), the following conditions hold:
	\begin{align*}
		&F(w) \geq F(\hat{w}) + \langle\nabla F(\hat{w}), (w - \hat{w})\rangle + \frac{\mu}{2} \|w - \hat{w}\|^2. \\
		&\|\nabla F(w) - \nabla F(\hat{w})\| \leq L \|w - \hat{w}\|.
	\end{align*}

\textbf{Assumption 2. }
	The local training dataset \( D_i \) (for \( i = 1, 2, \dots, n \)) of each client 
	and the root dataset \( D_0 \) are independently sampled from the distribution 
	\( \mathcal{X} \).

\begin{theorem}\label{theorem:robust}
	For any number of malicious clients, the difference between the 
	global model $W^t$ learned by VFEFL and the optimal global model \( W^{opt}\) 
	(in the absence of attacks) is bounded. 
	\begin{align*}
		& \| W^t-W^{opt} \| \le \\
		& 3\left(\sqrt{(1 - 2\eta \mu+\eta^2 L^2)}\right)^{lr} \|W_0 - W^{opt}\| +2\|W^{opt}\|
	\end{align*}
	where $\eta$ is the local learning rate
	and $W_0$ is the initial model.
	When $\|\sqrt{(1 - 2\eta \mu+\eta^2 L^2)}\| < 1$, we have $\operatorname{Lim}_{t \to \infty} | W^t-W^{opt} \| \le 2\|W^{opt}\|$.
\end{theorem}
\begin{proof}\label{proof:robust}
	We denote by $\mathcal{HC}$ the set of clients whose $y$ is positive in the $t$th global iteration. 
	$W^{t}_i$ is the local model of client $i$ in the interaction $t$ and $W^t_0 $
	is the baseline model in the interaction $t$.
	The aggregation rule can be written as:
	$$
		W^*=\sum_{i \in \mathcal{HC}} \frac{\langle W^{t}_i, W^t_0\rangle}{\langle W^t_i, W^t_i\rangle} W^t_i
		\text{ and }
		W^t= \frac{\left\| W^t_0 \right\|}{\left\| W^* \right\|} W^*
	$$
	We have the following equations for the tth global iteration:
	\begin{align*}
		& \|W^t-W^{opt}\| \\
		& =\left\|W^t-W^t_0+W^t_0-W^{opt}\right\| \\
		& \leq\left\|W^t-W^t_0\right\|+\left\|W^t_0-W^{opt}\right\| \\
		& \stackrel{(a)}{\leq}\left\|W^t+W^t_0\right\|+\left\|W^t_0-W^{opt}\right\| \\
		& \leq\left\|W^t\right\|+\left\|W^t_0\right\|+\left\|W^t_0-W^{opt}\right\| \\
		& \stackrel{(b)}{=} 2\left\|W^t_0\right\|+\left\|W^t_0-W^{opt}\right\| \\
		& =2\left\|W^t_0-W^{opt}+W^{opt}\right\|+\left\|W^t_0-W^{opt}\right\| \\
		& \leq 2\left\|W^t_0-W^{opt}\right\|+2\|W^{opt}\|+\left\|W^t_0-W^{opt}\right\| \\
		& =3\left\|W^t_0-W^{opt}\right\|+2\|W^{opt}\| \\
		& \stackrel{(c)}{\leq} 3\left(\sqrt{(1 - 2\eta \mu+\eta^2 L^2)}\right)^{lr} \|W_0 - W^{opt}\| +2\|W^{opt}\|
	\end{align*}
	where $W_0^t$ is a base model trained by the server on a clean root dataset,
	\( lr \) is the total number of server-side training iterations over $t$ rounds of iteration,
	(a) is because $\langle W^t_i, W^t_0\rangle > 0$ for $i \in \mathcal{HC}$ then 
	$\langle W^t, W^t_0\rangle > 0$; (b) is because $\|W^t\|=\|W^t_0\|$;
	(c) is based on lemma \ref{lemma_w0}.
	Thus, we conclude the proof.
\end{proof}

\subsection{Verifiability, self-contained and fidelity}
	The verifiability of VFEFL is inherently dependent on the verifiability 
	of CC-DVFE Scheme. This dependency ensures two critical properties:
it prevents malicious clients from framing honest ones,
and ensures that authenticated ciphertexts and keys can be correctly decrypted.
	Thus, the clients successfully proves to the server that both the inputs and behavior 
	follow the protocol specification without compromising the privacy of the model.
	
	The proposed scheme is self-contained, as it can be deployed in the most basic 
	federated learning framework without two non-colluding servers and
	any additional trusted third parties. Furthermore, 
	it achieves high fidelity by ensuring VFEFL introduces no noise during aggregation, 
	thereby preserving the accuracy of the global model in the absence of attacks.
	This fidelity is validated through experiments in Section \ref{experiments}, 
	demonstrating the scheme's reliability in maintaining model performance.

\section{EXPERIMENTS}\label{experiments}
\subsection{Experimental Settings}
In this section, we evaluate the model accuracy, Byzantine robustness and efficiency 
of the federated learning architecture VFEFL.

\textbf{1)Implementation.}
We conduct our experiments on a local machine equipped with a 13th Gen Intel(R) Core(TM) i5-1340P CPU @ 1.90GHz and 16GB RAM.
And we approximate floating-point numbers to two 
decimal places and convert them into integers for cryptographic computations, a common 
practice in cryptographic schemes. 
The final decryption step involves computing a discrete logarithm, which introduces a 
substantial computational overhead. Specifically, the problem reduces to solving for 
\( x \) in the equation \( a = g^x \), where \( g \) is an element of a bilinear group 
and \( x \) is a relatively small integer. To efficiently tackle this, we utilize the 
baby-step giant-step algorithm \cite{BSGS}, which is known for its effectiveness in 
computing discrete logarithms in such contexts.
\textbf{2)Dataset and Models.}
On data set allocation, we set aside a portion of the dataset as a trusted root dataset 
$D_0$ for the server and then randomly divide the remaining data into clients. 

We use multiple datasets from different domains, %in the evaluation,
and evaluate the accuracy of the final global model on the entire test set.
The publicly available datasets are the following.
	\begin{itemize}
		\item[-] \textbf{MNIST:} 
			The dataset consists of handwritten digits, organized into 10 classes, with 
			60,000 samples for training and 10,000 samples for testing. Each sample is 
			a grayscale image with a resolution of \(28 \times 28\) pixels. 
		\item[-] \textbf{Fashion-MNIST:} The dataset consists of grayscale images of fashion items, 
			divided into 10 categories. It provides 60,000 training samples and 
			10,000 test samples, with each image having dimensions of $28 \times 28$ pixels. 

		\item[-] \textbf{CIFAR-10:} This dataset consists of color images belonging to 
			10 distinct classes, with 50,000 training samples and 10,000 test samples. 
			Each image has a resolution of \(32 \times 32\) pixels in RGB format.  
	\end{itemize}

\textbf{3)Attack Scenario.}
We investigate four types of attacks: Gaussian attack, Scaling attacks,  
Adaptive attack and Label Flip attack. 
In our experiments, 20\% of clients are designated as malicious under 
untargeted poisoning attacks, while this proportion is reduced to 10\% 
for targeted poisoning attacks.
The specific details of each attack are outlined as follows:
	\begin{itemize}
		\item[-] \textbf{Gaussian Attack\cite{fang2020local}(GA):} 
			Malicious clients may make the server's global model 
			less accurate by uploading random models. 
			This attack randomly generates local models on compromised clients, 
			which is designed using the approach and parameters from \cite{fang2020local}.
		\item[-] \textbf{Scaling Attack(SA):} On the basis of the above Gaussian attack, 
			the malicious client enhances the attack by amplifying its own model 
			parameters by a scaling factor \(\lambda \gg 1\) before submitting them to the 
			server, thereby disproportionately influencing the global model. 
			Based on prior work, we set the scaling factor 
			\( \lambda = n \), where \( n \) represents the number of clients.

		\item[-] \textbf{Adaptive Attack\cite{FLtrust}(AA):} 
			For each aggregation rule, 
			we define a corresponding objective function that quantifies the gap between 
			the global model aggregated solely from honest clients and the global model 
			after the participation of malicious clients and the solution of the 
			optimization problem is used as a model for the attack, 
			enabling the malicious client models to maximize the attack's effectiveness. 
			We use the default parameter settings in 
			\cite{FLtrust} and the inverse direction of the average of all models for 
			each dimension as the initial malicious model.
		\item[-] \textbf{Label Flipping Attack(LFA):} Our work adopts the label-flipping attack 
			strategy outlined in \cite{FLtrust}. Specifically, in all experiments, 
			the malicious clients flip its label $l$ to $M - l - 1$, where 
			$M$ is the total number of labels and $l \in \{0, 1, \cdots , M-1\}$.
		\end{itemize}

\textbf{4)Evaluation Metrics.}
	In this experiment, we use the accuracy (AC) of the global model as the primary 
	performance metric. For the LF attack, in addition to model accuracy, 
	we also consider the attack success rate (ASR) as an evaluation criterion. 
	The attack success rate is defined as the probability that a sample 
	with label \( l \) is misclassified as \( M-l-1 \).	
\subsection{Experimental Results}
	VFEFL achieves the three defense goals:
	First, VFEFL achieves fidelity by maintaining testing error rates comparable to FedAvg
	under benign conditions. Second, it demonstrates robustness by effectively countering attacks from 
	malicious clients. Finally, its efficiency remains within acceptable bounds, ensuring feasibility 
	for real-world deployment.

	First, when there is no attack, our scheme's accuracy rates similar to FedAvg.
	As we can see in the Figure \ref{fig:no_attack}, the accuracy of our scheme is 
	largely in line with FedAvg across multiple datasets.
	For instance, on MNIST, the accuracy of both FedAvg and VFEFL is close to 99\%.
	Our aggregation rules take into account all client updates when there is no attack, 
	so the accuracy of the global model is guaranteed. In the comparison in the table, 
	we chose comparison schemes that are all capable of implementing fidelity protection.

	\begin{figure*}[!t]
  \centering
  \subfloat[MNIST]{
    \includegraphics[width=0.3\linewidth]{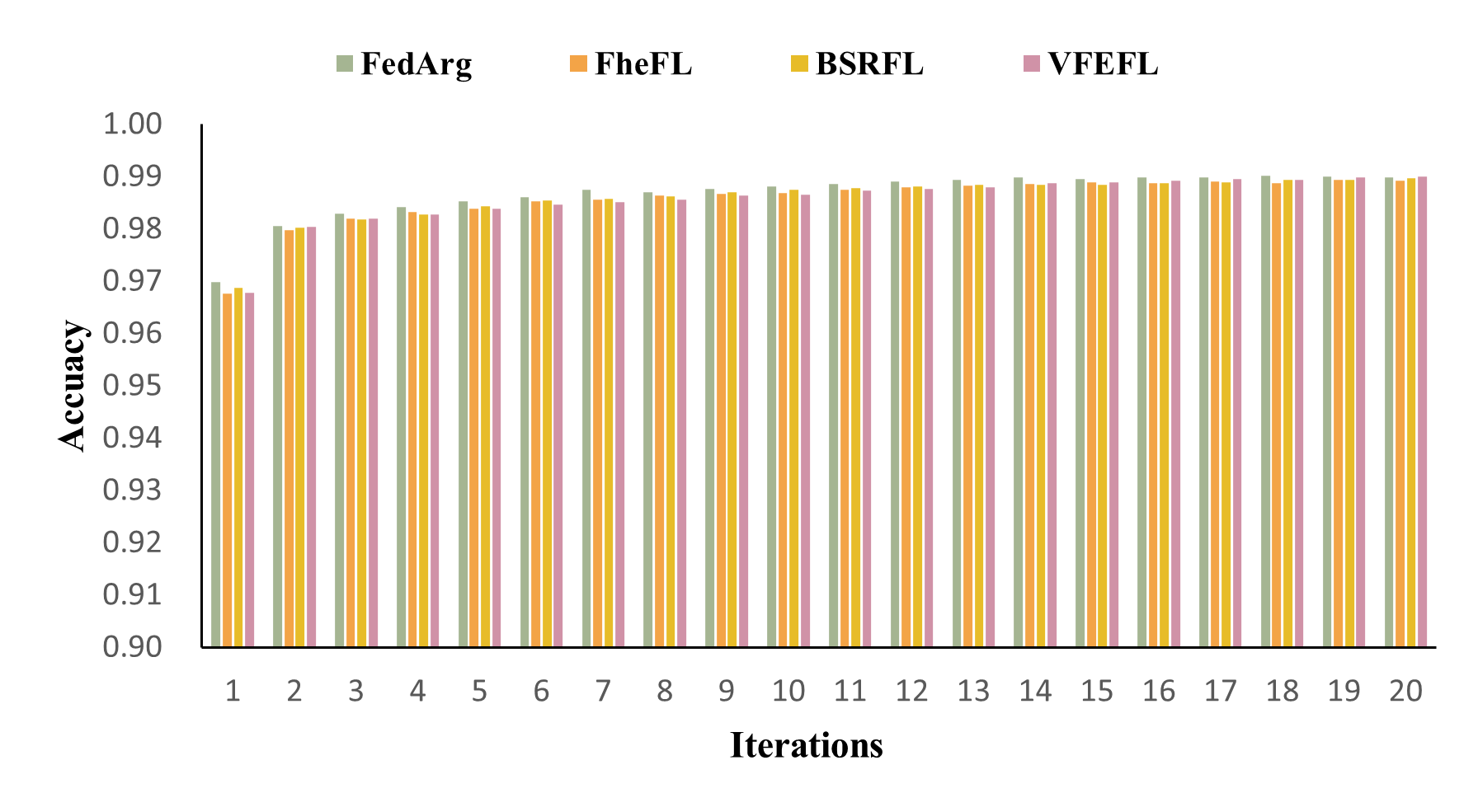}
    \label{fig:mnist}
  }
  \hfil
  \subfloat[Fashion-MNIST]{
    \includegraphics[width=0.3\linewidth]{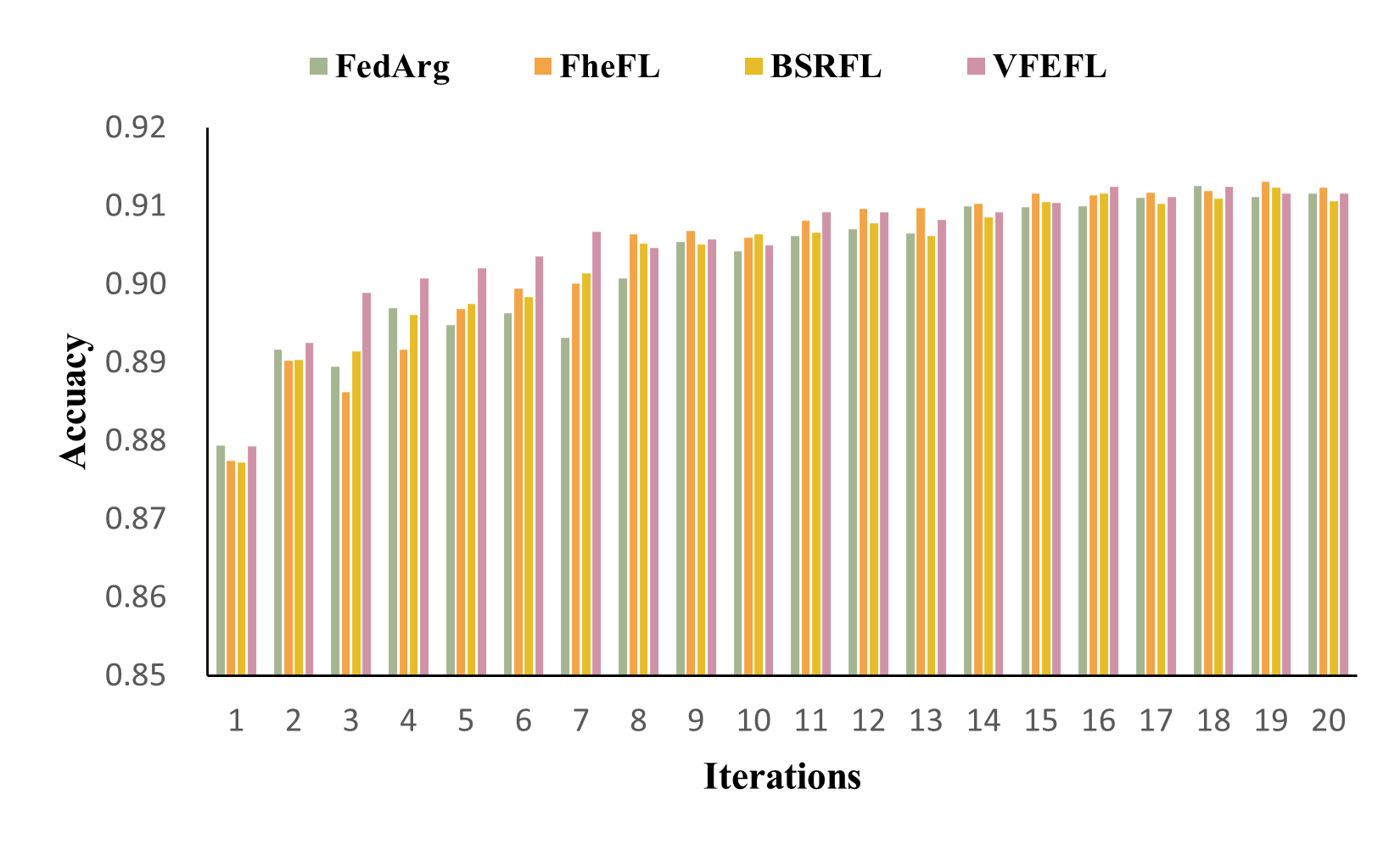}
    \label{fig:fashion}
  }
  \hfil
  \subfloat[CIFAR-10]{
    \includegraphics[width=0.3\linewidth]{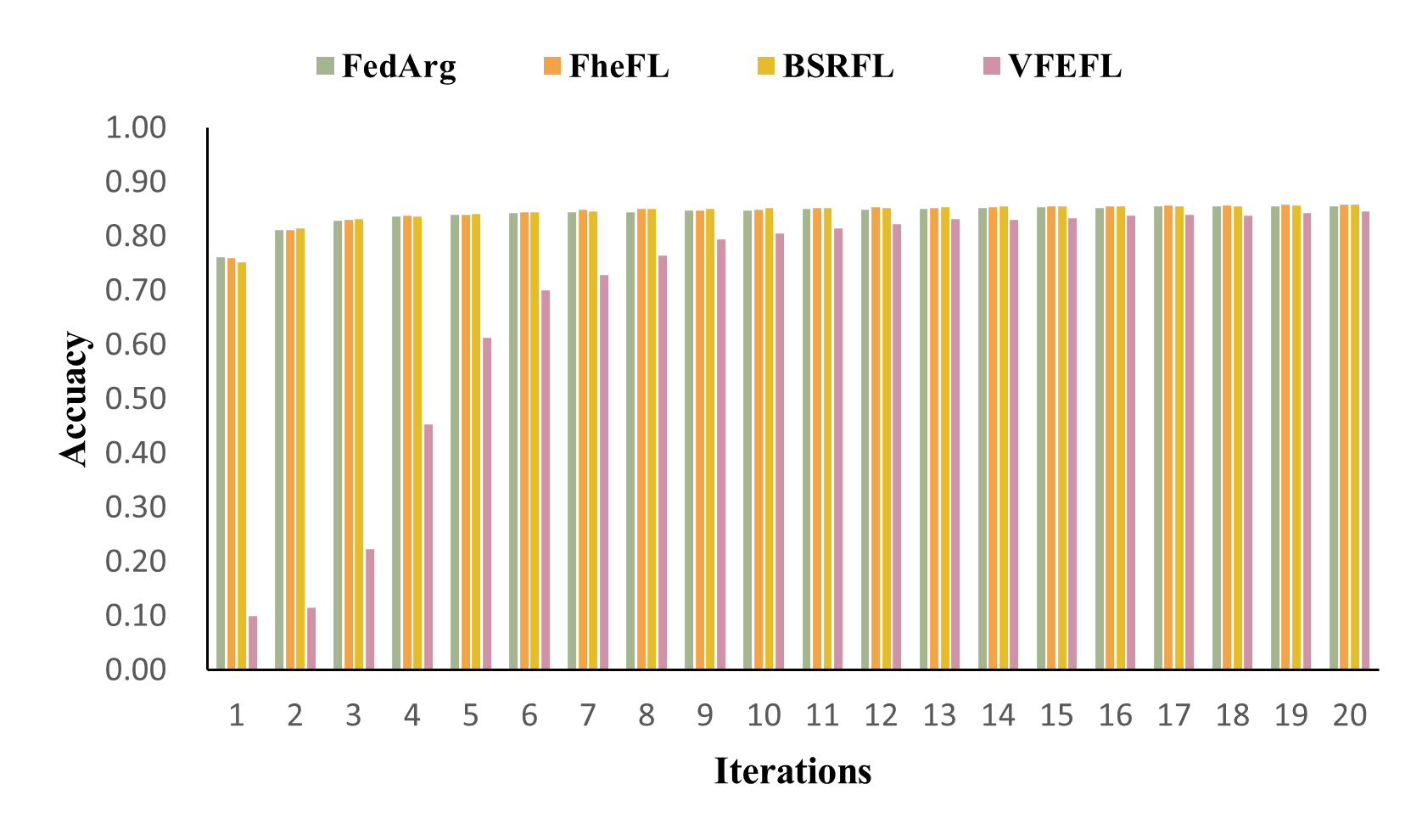}
    \label{fig:cifar}
  }
  \caption{Model accuracy on three datasets in the absence of attacks: (a) MNIST, (b) Fashion-MNIST, and (c) CIFAR-10.}
  \label{fig:no_attack}
\end{figure*}

	Secondly, the proposed scheme demonstrates strong robustness. 
	As illustrated in Figures \ref{fig:mnist_attacks}, \ref{fig:fashionmnist_attacks} and \ref{fig:cifar_attacks}, 
	across three different datasets, the global model trained under our scheme consistently maintains 
	high accuracy when subjected to Gaussian attack, Scaling attack, and Adaptive attack. 
	In the case of Gaussian attacks, the accuracy of the global model drops by less than 0.1\%. 
	The other three baseline methods also exhibit strong robustness against Gaussian attack,
	however, they fail to defend effectively against Scaling attacks. 
	This vulnerability arises from the fact that these methods neither constrain nor 
	validate the structural characteristics of local models submitted by malicious clients. 
	As a result, when adversaries employ models with sufficiently large magnitudes, 
	they can exert a high influence on the final aggregated model.
	In contrast, thanks to the robust 
	aggregation rule in our scheme, which effectively verifies and restricts the norm of 
	local models continues to 
	sustain high accuracy even in the presence of Scaling attacks. 
	In the presence of a Adaptive attack, both BSR-FL and VFEFL are successfully 
	aggregated into a final global model that achieves sufficiently high accuracy.
	Overall, these experimental results validate that our federated learning framework 
	possesses strong robustness against various types of untargeted poisoning attacks.

	\begin{figure*}[!t]
  \centering
  \subfloat[GA]{
    \includegraphics[width=0.3\linewidth]{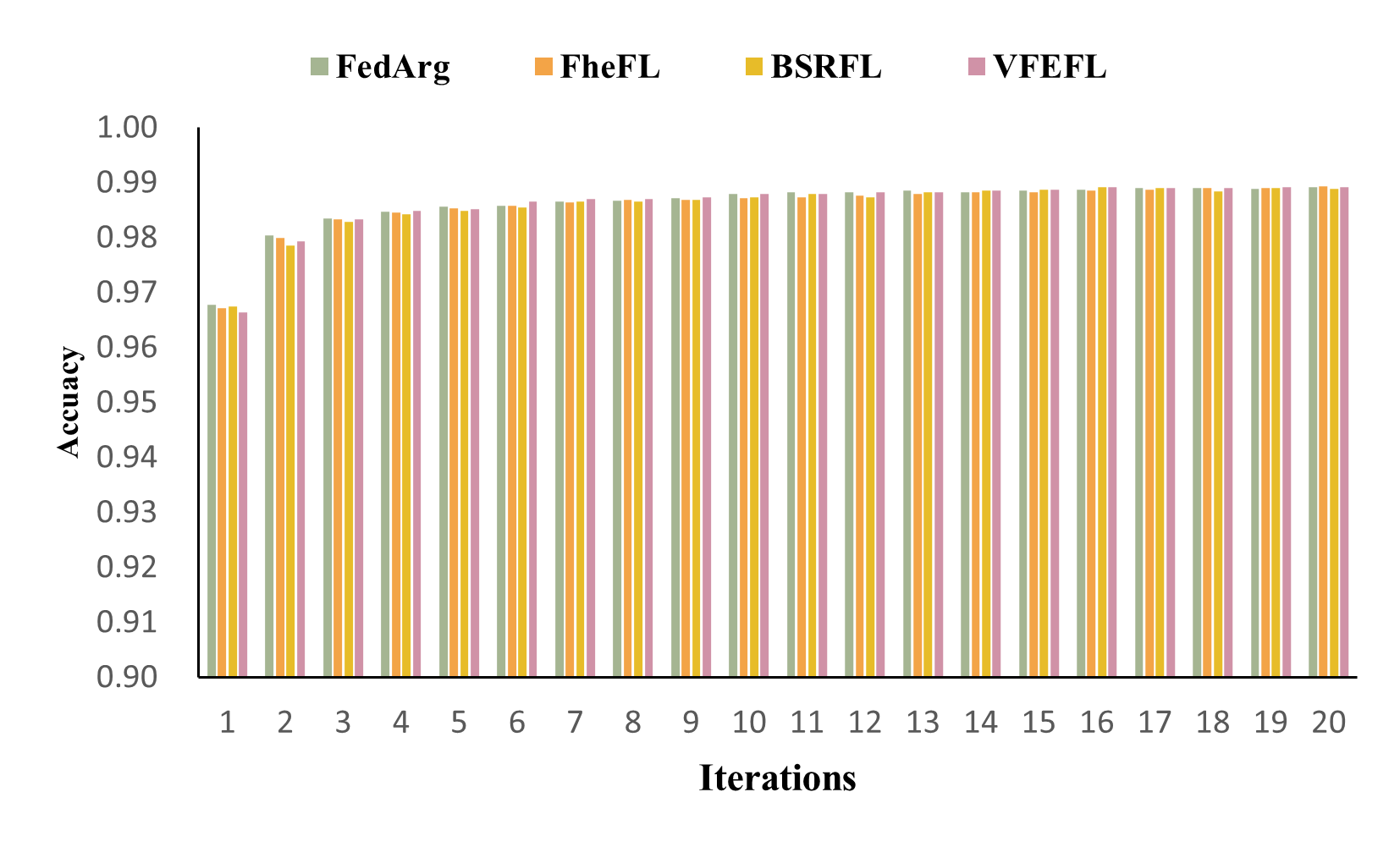}
    \label{fig:mnist_ga}
  }
  \hfil
  \subfloat[SA]{
    \includegraphics[width=0.3\linewidth]{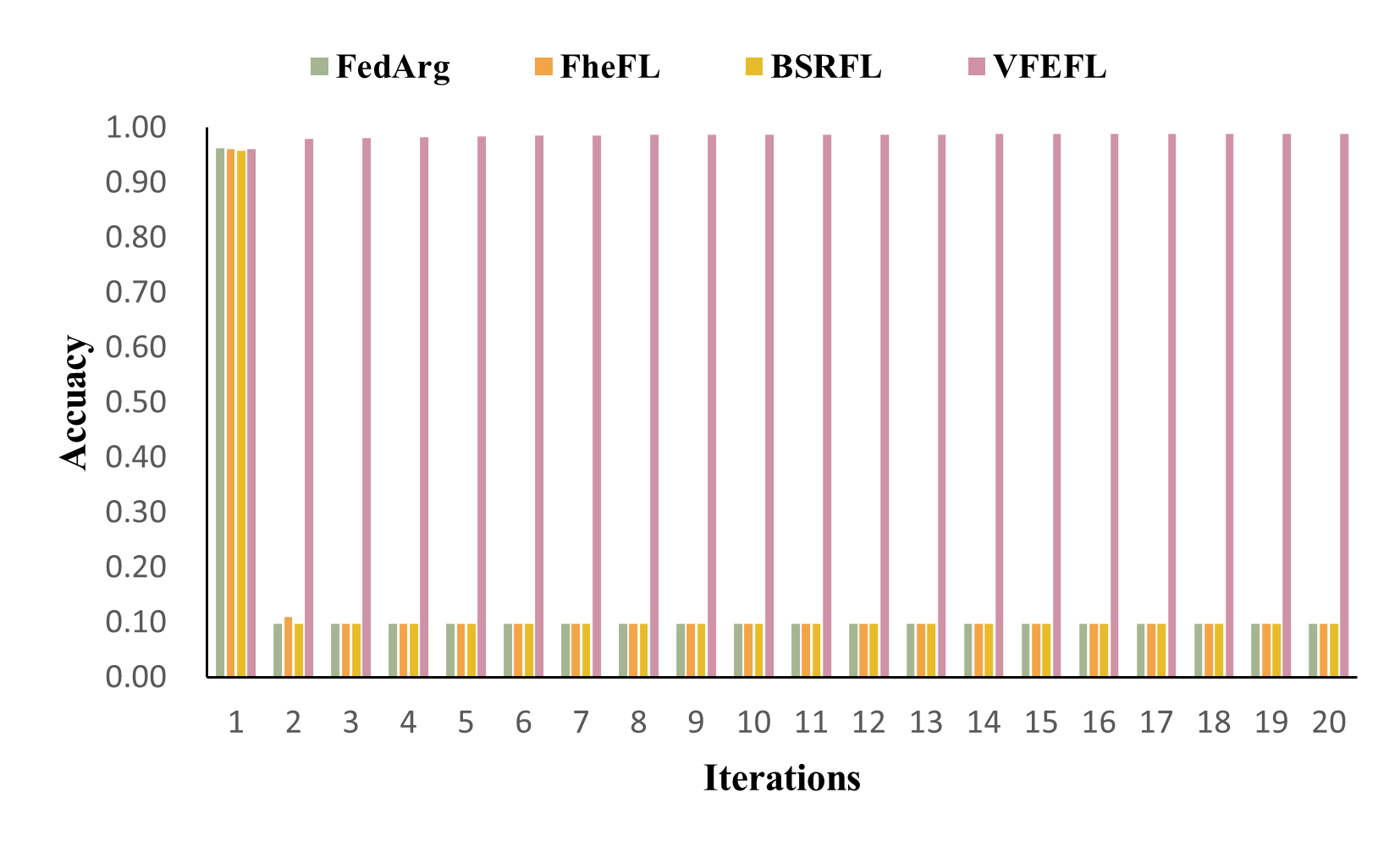}
    \label{fig:mnist_sa}
  }
  \hfil
  \subfloat[AA]{
    \includegraphics[width=0.3\linewidth]{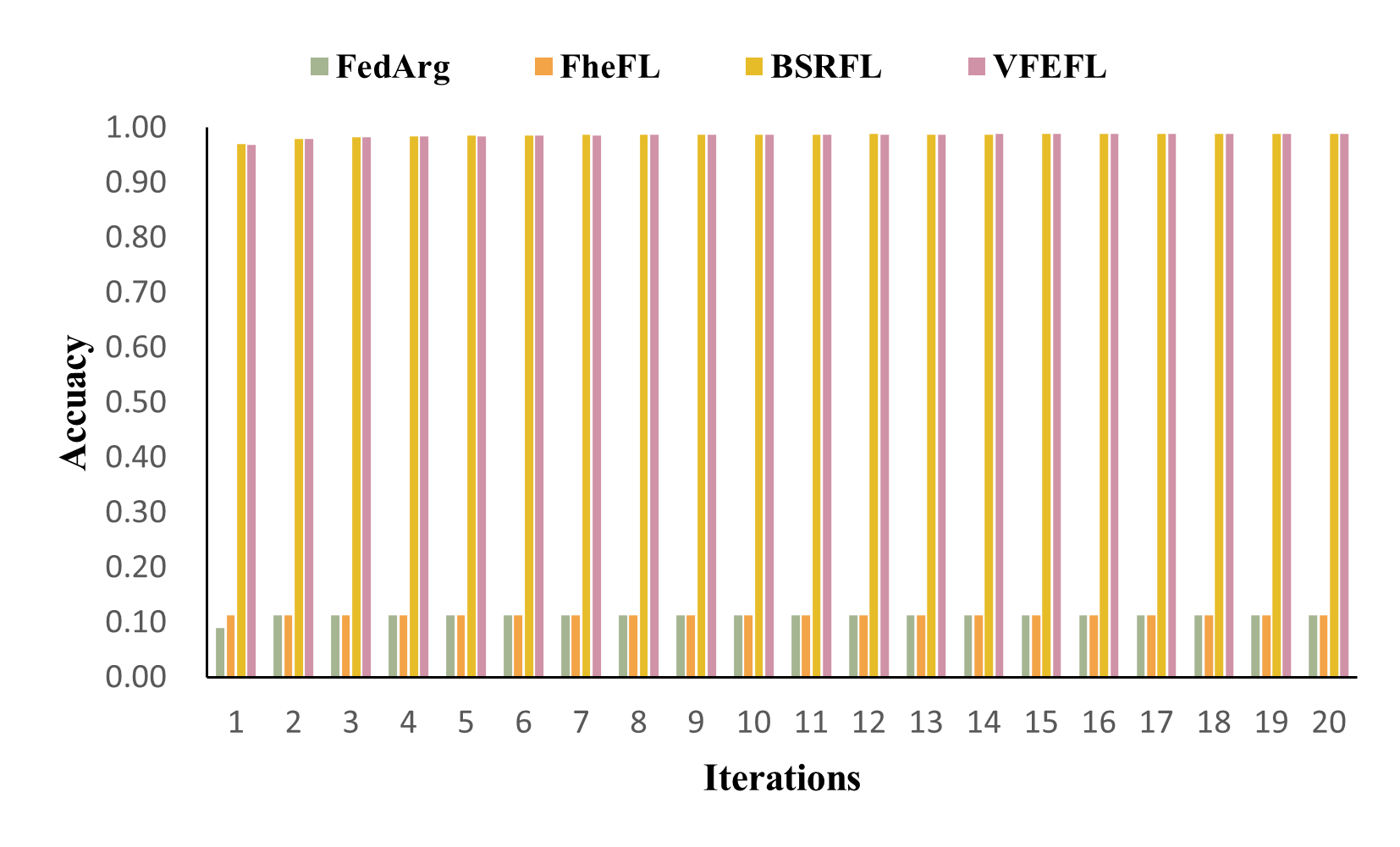}
    \label{fig:mnist_aa}
  }
  \caption{Model accuracy under different attacks on MNIST: (a) GA, (b) SA, and (c) AA.}
  \label{fig:mnist_attacks}
\end{figure*}
\begin{figure*}[!t]
  \centering
  \subfloat[GA]{
    \includegraphics[width=0.3\linewidth]{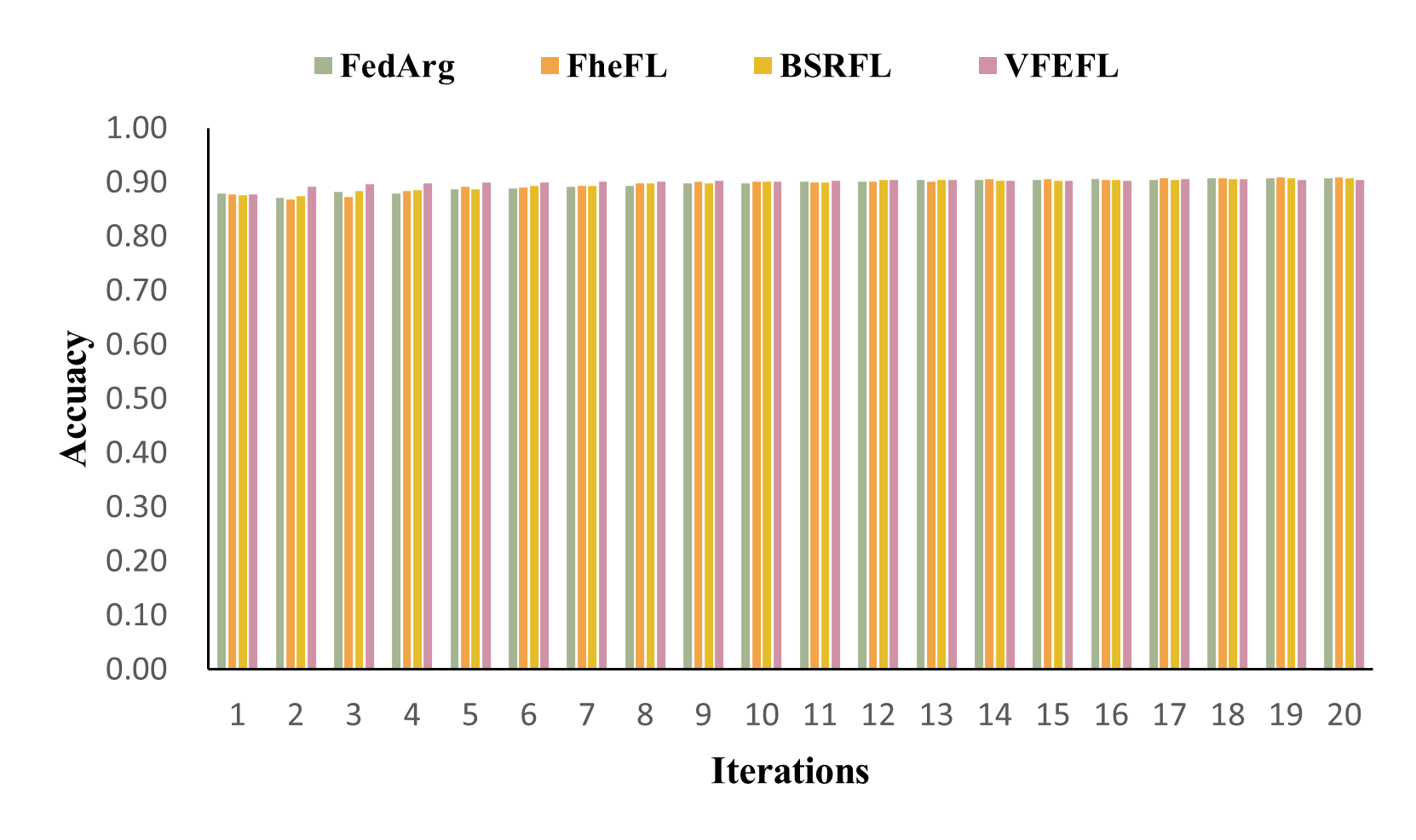}
    \label{fig:fashion_ga}
  }
  \hfil
  \subfloat[SA]{
    \includegraphics[width=0.3\linewidth]{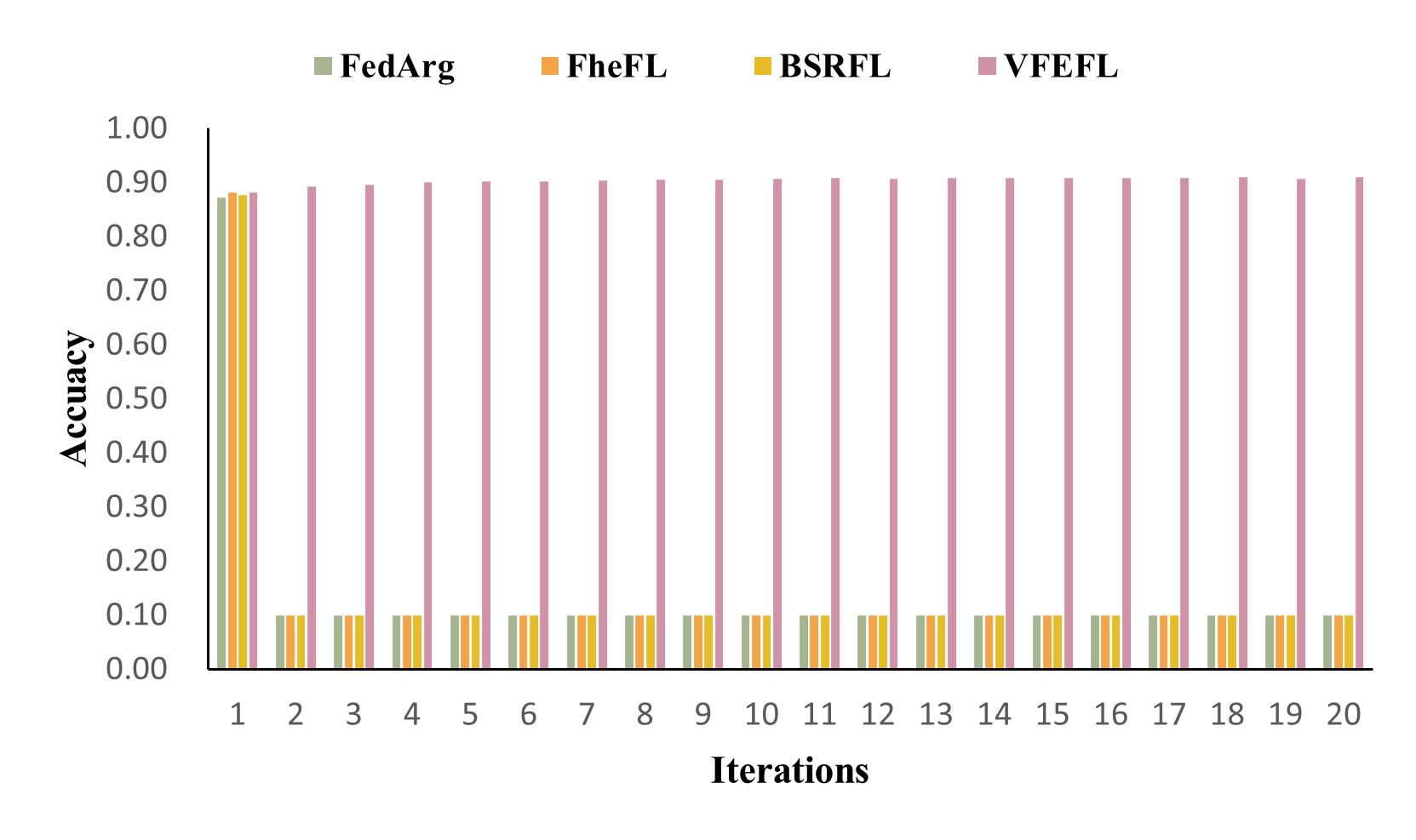}
    \label{fig:fashion_sa}
  }
  \hfil
  \subfloat[AA]{
    \includegraphics[width=0.3\linewidth]{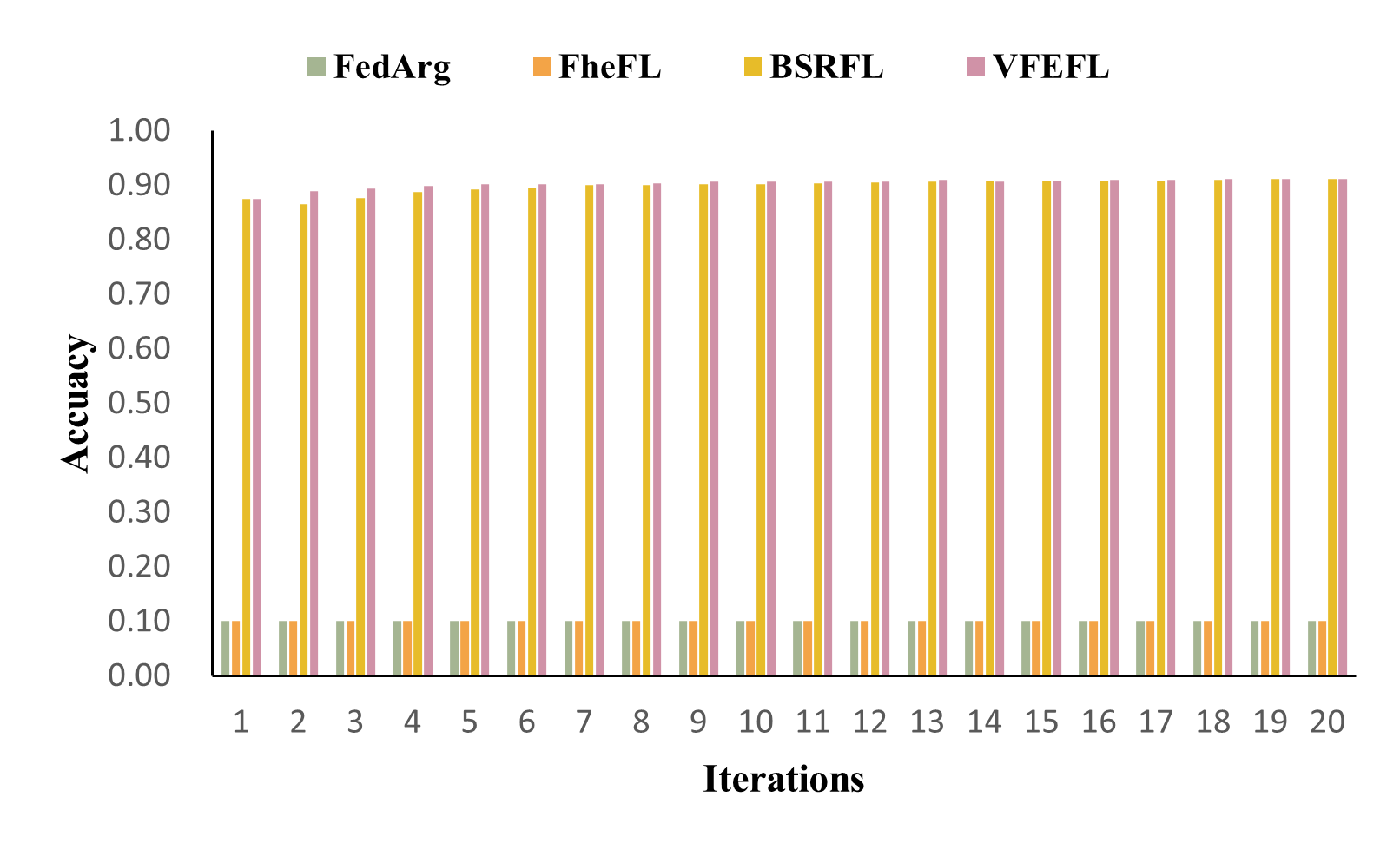}
    \label{fig:fashion_aa}
  }
  \caption{Model accuracy under different attacks on Fashion-MNIST: (a) GA, (b) SA, and (c) AA.}
  \label{fig:fashionmnist_attacks}
\end{figure*}
\begin{figure*}[!t]
  \centering
  \subfloat[GA]{
    \includegraphics[width=0.3\linewidth]{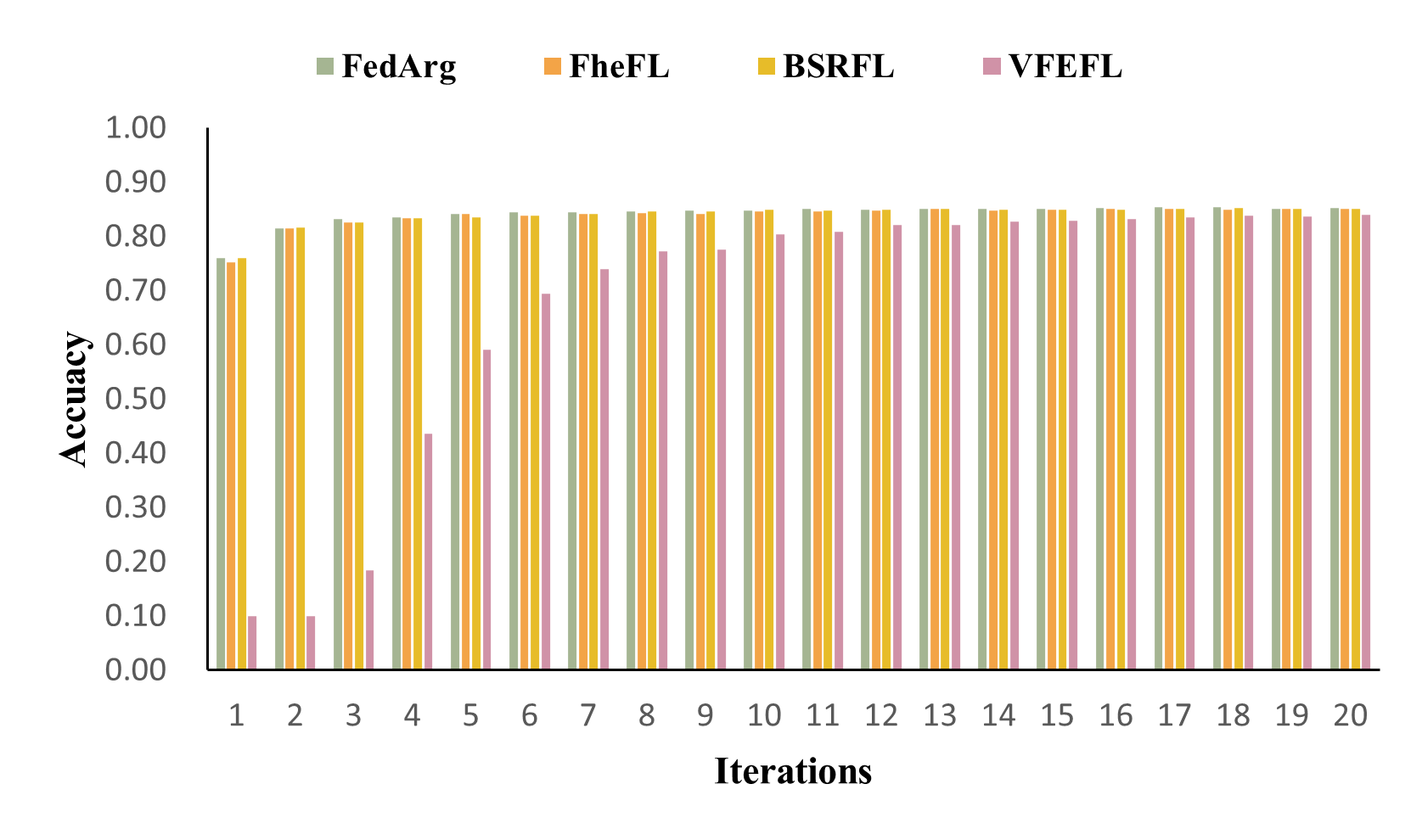}
    \label{fig:cifar_ga}
  }
  \hfil
  \subfloat[SA]{
    \includegraphics[width=0.3\linewidth]{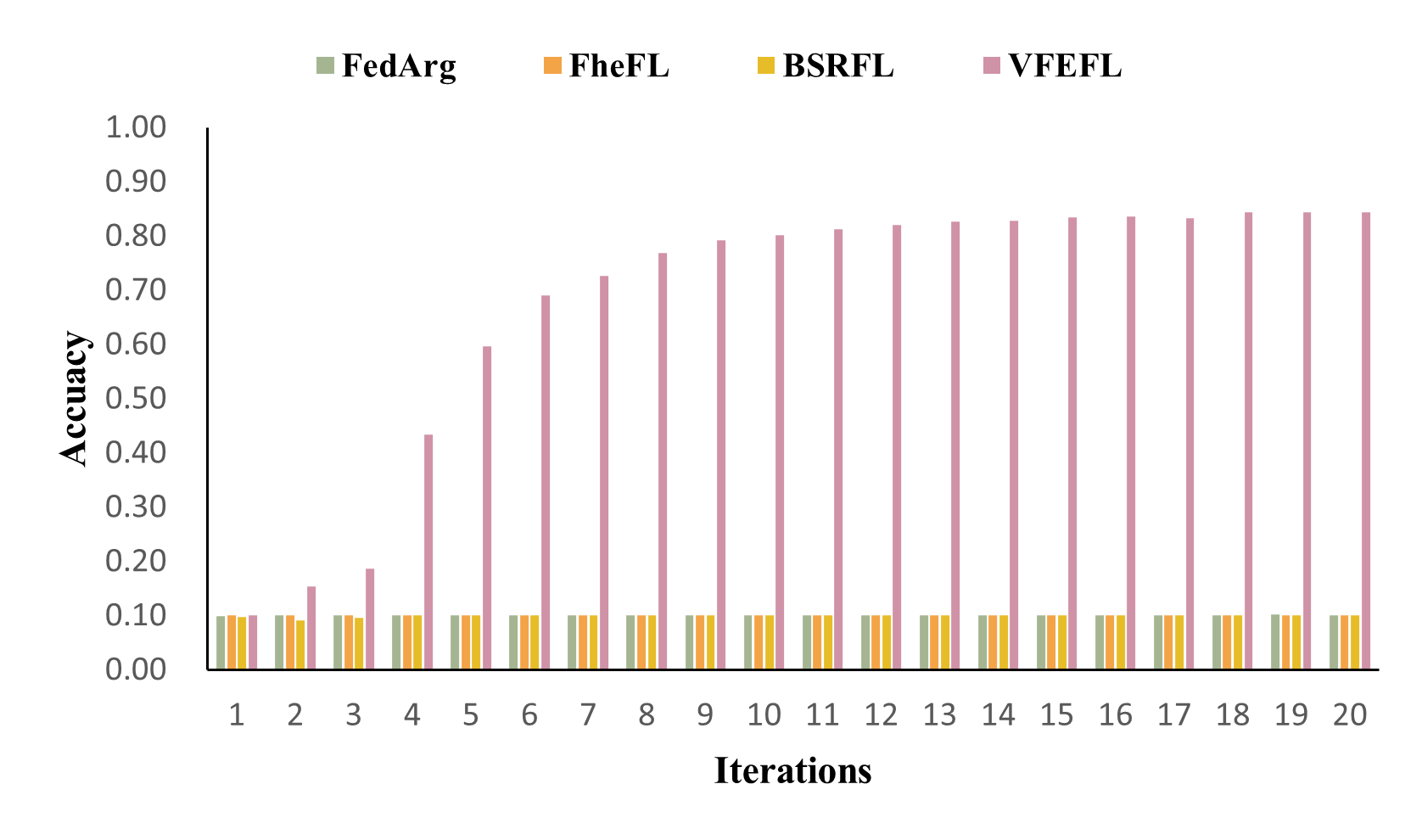}
    \label{fig:cifar_sa}
  }
  \hfil
  \subfloat[AA]{
    \includegraphics[width=0.3\linewidth]{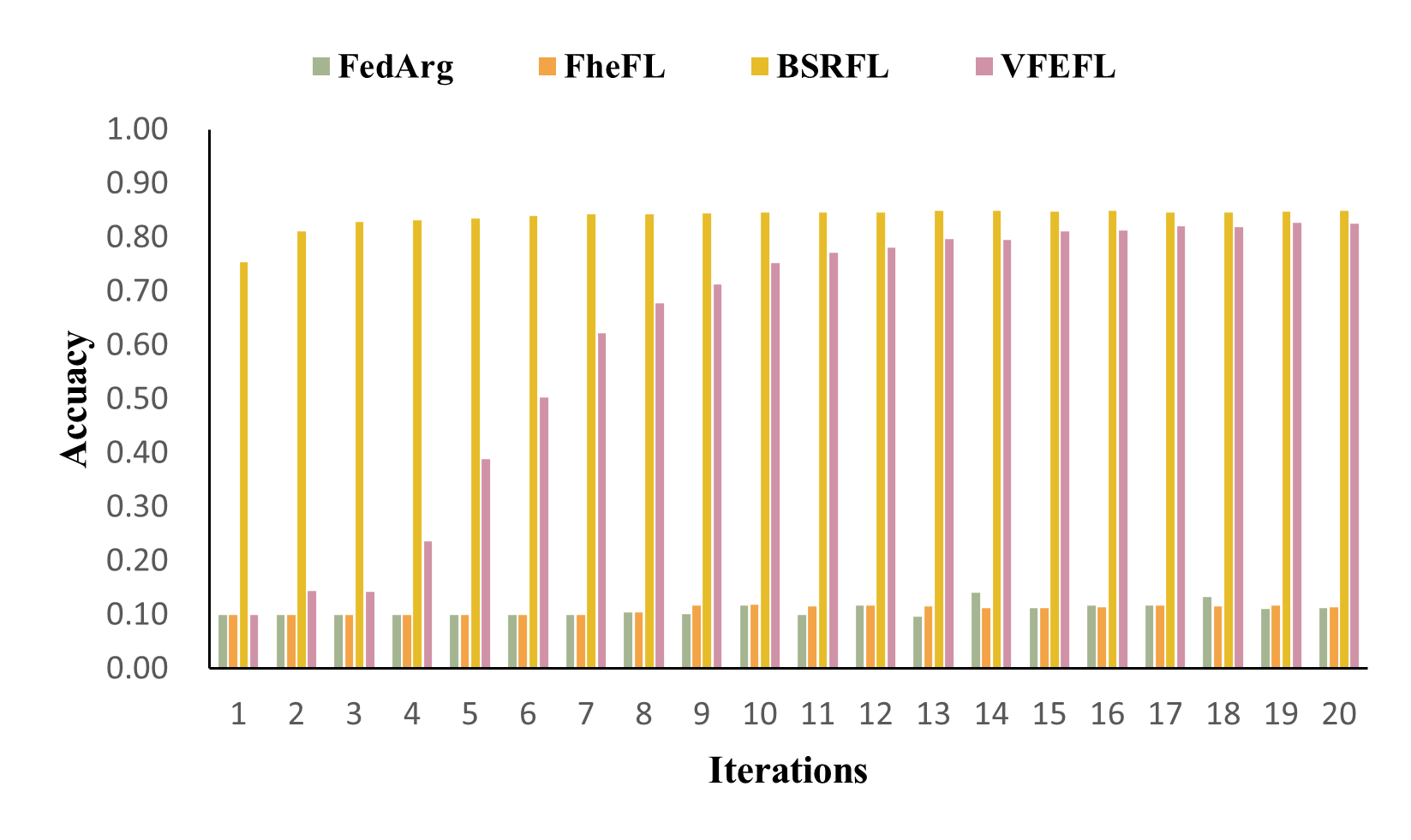}
    \label{fig:cifar_aa}
  }
  \caption{Model accuracy under different attacks on CIFAR-10: (a) GA, (b) SA, and (c) AA.}
  \label{fig:cifar_attacks}
\end{figure*}

	As illustrated in Figures \ref{fig:LF_mnist}, \ref{fig:LF_fashion} and \ref{fig:LF_cifar}, 
	the proposed scheme consistently achieves high accuracy under LF attacks across 
	all three datasets, while effectively suppressing the attack success rate. 
	This highlights the robustness of our approach in defending against LF attacks.
	But in the CIFAR10 dataset, the global models obtained from the training of 
	our scheme converge relatively slowly, as shown in Fig \ref{fig:LF_cifar}. 
	This phenomenon is 
	mainly due to the fact that the adopted aggregation rule normalises the 
	norm of each local model based on the root model. 
	This design effectively suppresses the malicious client's strategy to enhance 
	the attack effectiveness by amplifying the parameters of the local models, 
	as demonstrated by the scaling attack in Figs \ref{fig:mnist_sa}, 
	\ref{fig:fashion_sa}, and \ref{fig:cifar_sa}. 
	However, this 
	normalisation mechanism also increases the number of iterations required for 
	model convergence to some extent. Nevertheless, after about 20 iterations, 
	the global model trained by our scheme is still able to achieve 
	satisfactory accuracy. In addition, the global model trained by the proposed 
	scheme not only maintains high classification accuracy on MNIST and 
	Fashion-MNIST datasets, but also demonstrates better defence capability 
	and performance stability against malicious attacks  compared to other schemes
	in multiple iterations. 
	The above results further validate the robustness and effectiveness of 
	this paper's method in malicious environments.

	\begin{figure}[!t]
  \centering
  \subfloat[AC under LF attacks]{
    \includegraphics[width=0.45\linewidth]{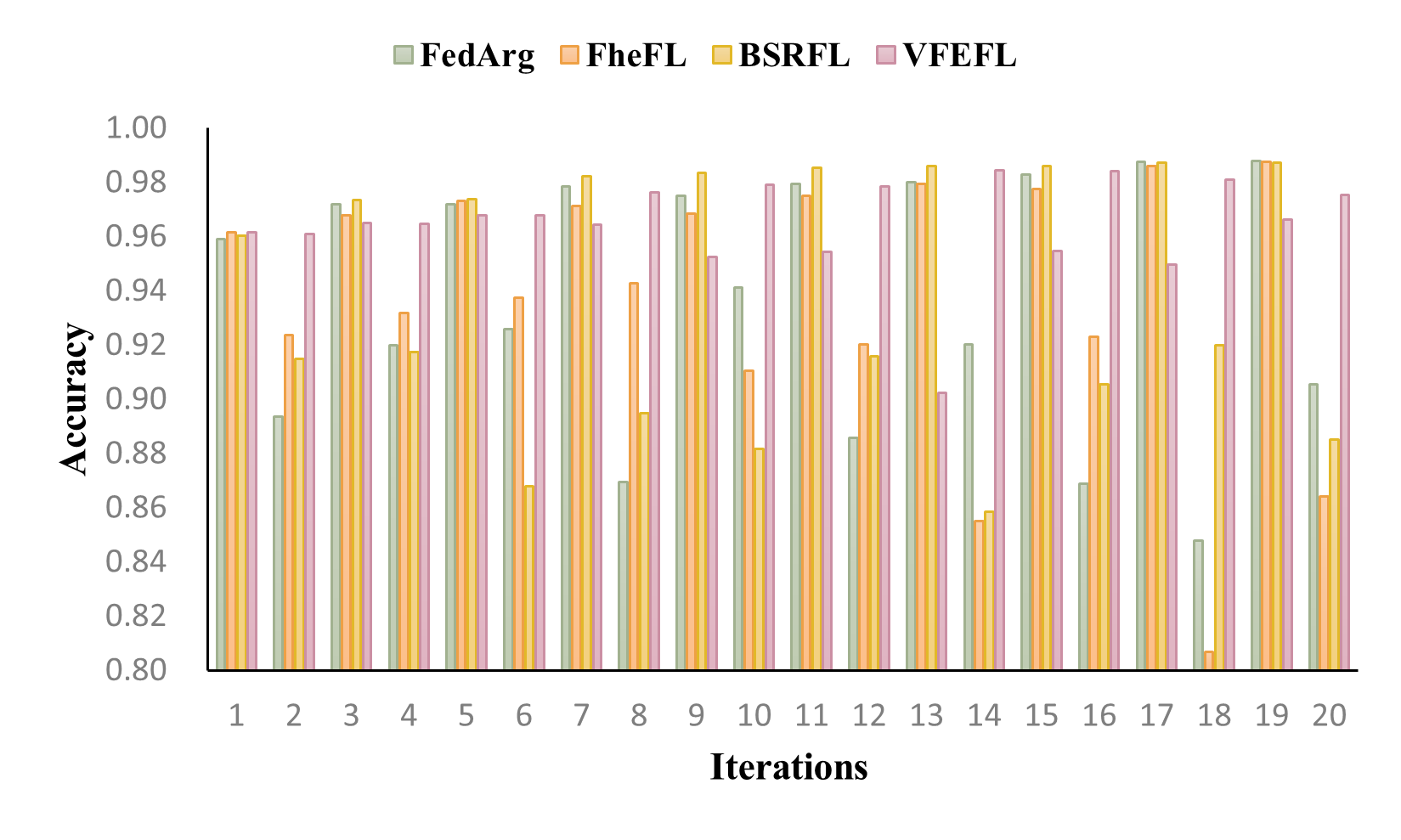}
    \label{fig:lf_mnist_ac}
  }
  \hfil
  \subfloat[ASR under LF attacks]{
    \includegraphics[width=0.45\linewidth]{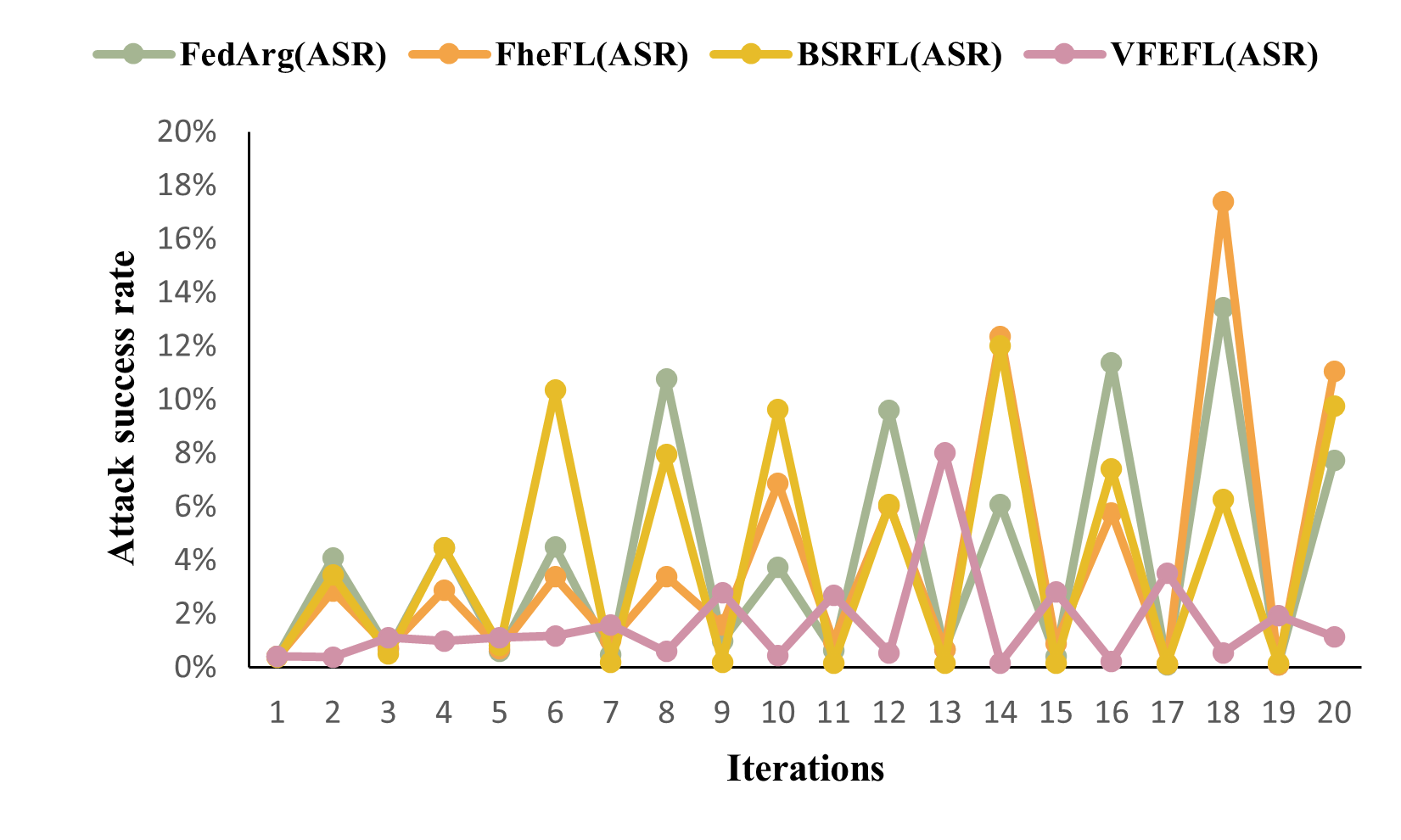}
    \label{fig:lf_mnist_asr}
  }
  \caption{AC and ASR under LF attacks on MNIST: (a) accuracy (AC) and (b) attack success rate (ASR).}
  \label{fig:LF_mnist}
\end{figure}
\begin{figure}[!t]
  \centering
  \subfloat[AC under LF attacks]{
    \includegraphics[width=0.45\linewidth]{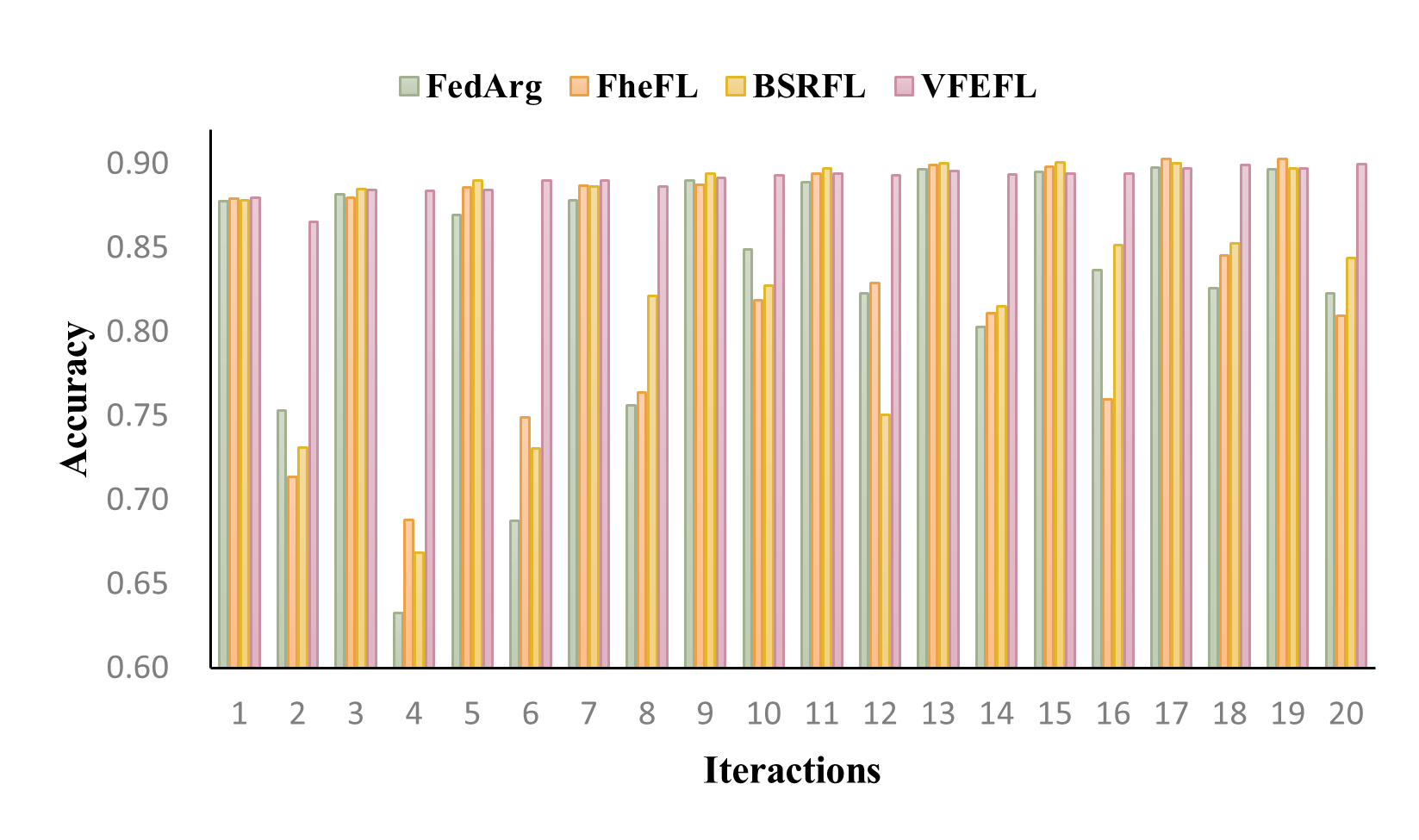}
    \label{fig:lf_fashion_ac}
  }
  \hfil
  \subfloat[ASR under LF attacks]{
    \includegraphics[width=0.45\linewidth]{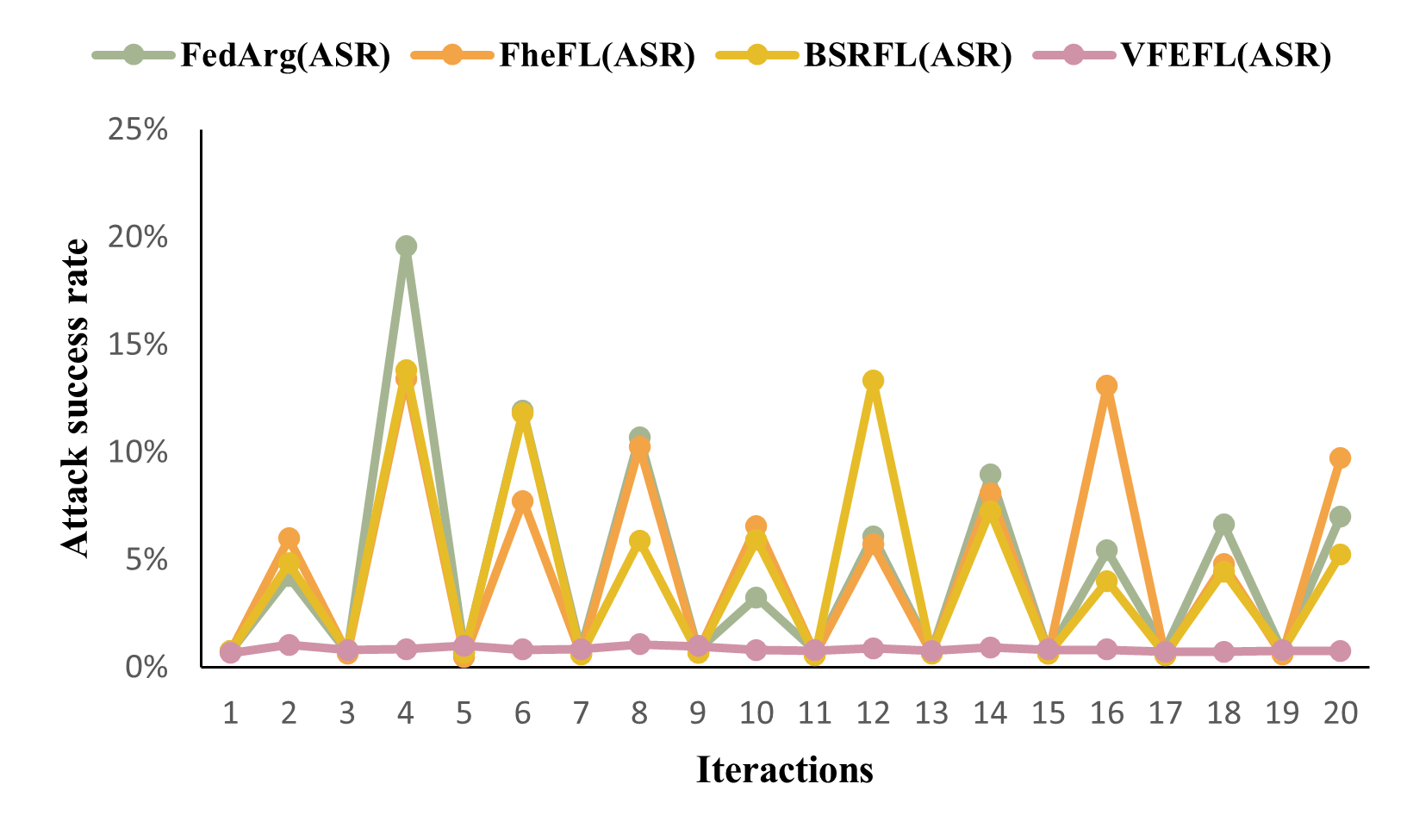}
    \label{fig:lf_fashion_asr}
  }
  \caption{AC and ASR under LF attacks on Fashion-MNIST: (a) accuracy (AC) and (b) attack success rate (ASR).}
  \label{fig:LF_fashion}
\end{figure}
\begin{figure}[!t]
  \centering
  \subfloat[AC under LF attacks]{
    \includegraphics[width=0.45\linewidth]{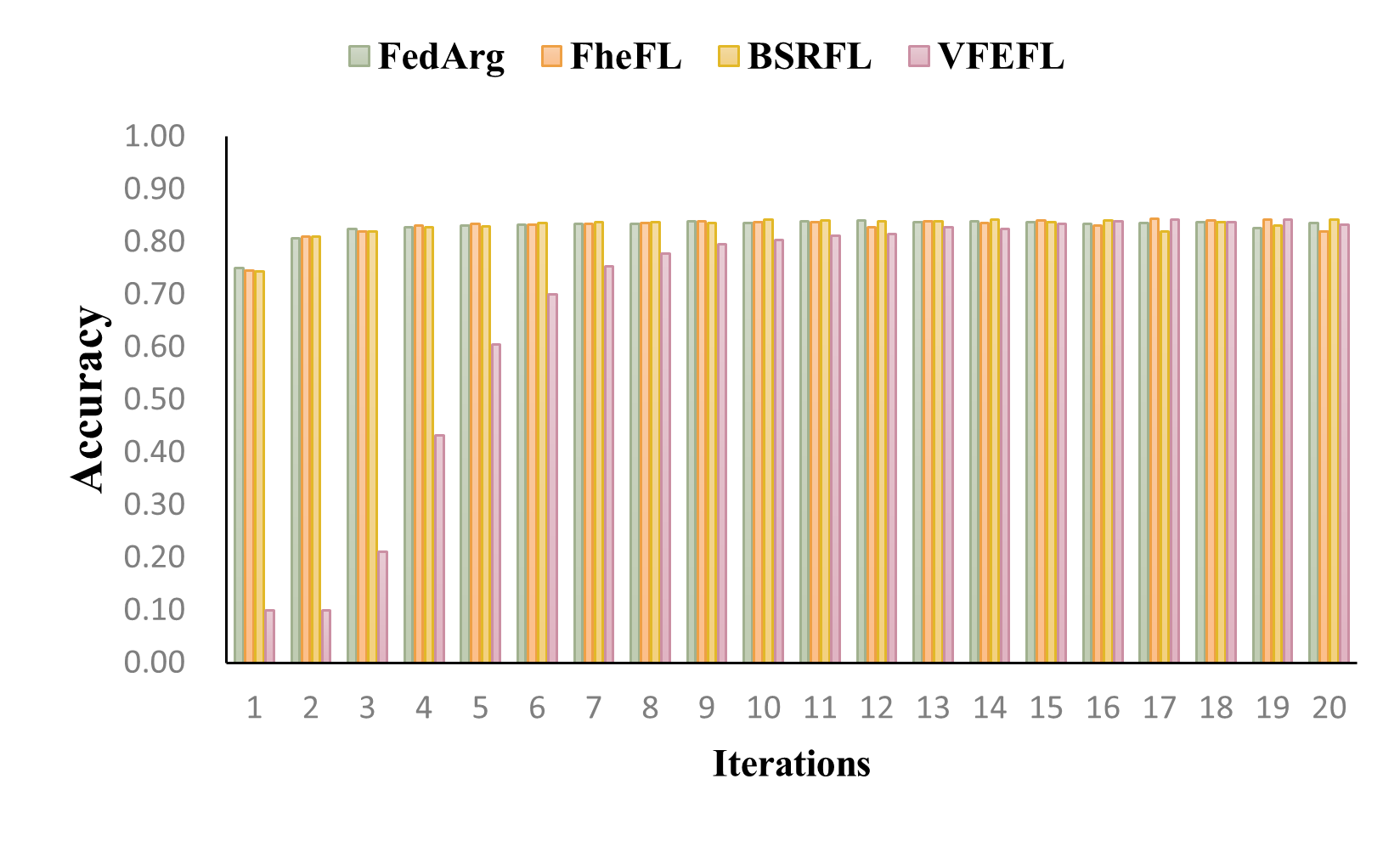}
    \label{fig:lf_cifar_ac}
  }
  \hfil
  \subfloat[ASR under LF attacks]{
    \includegraphics[width=0.45\linewidth]{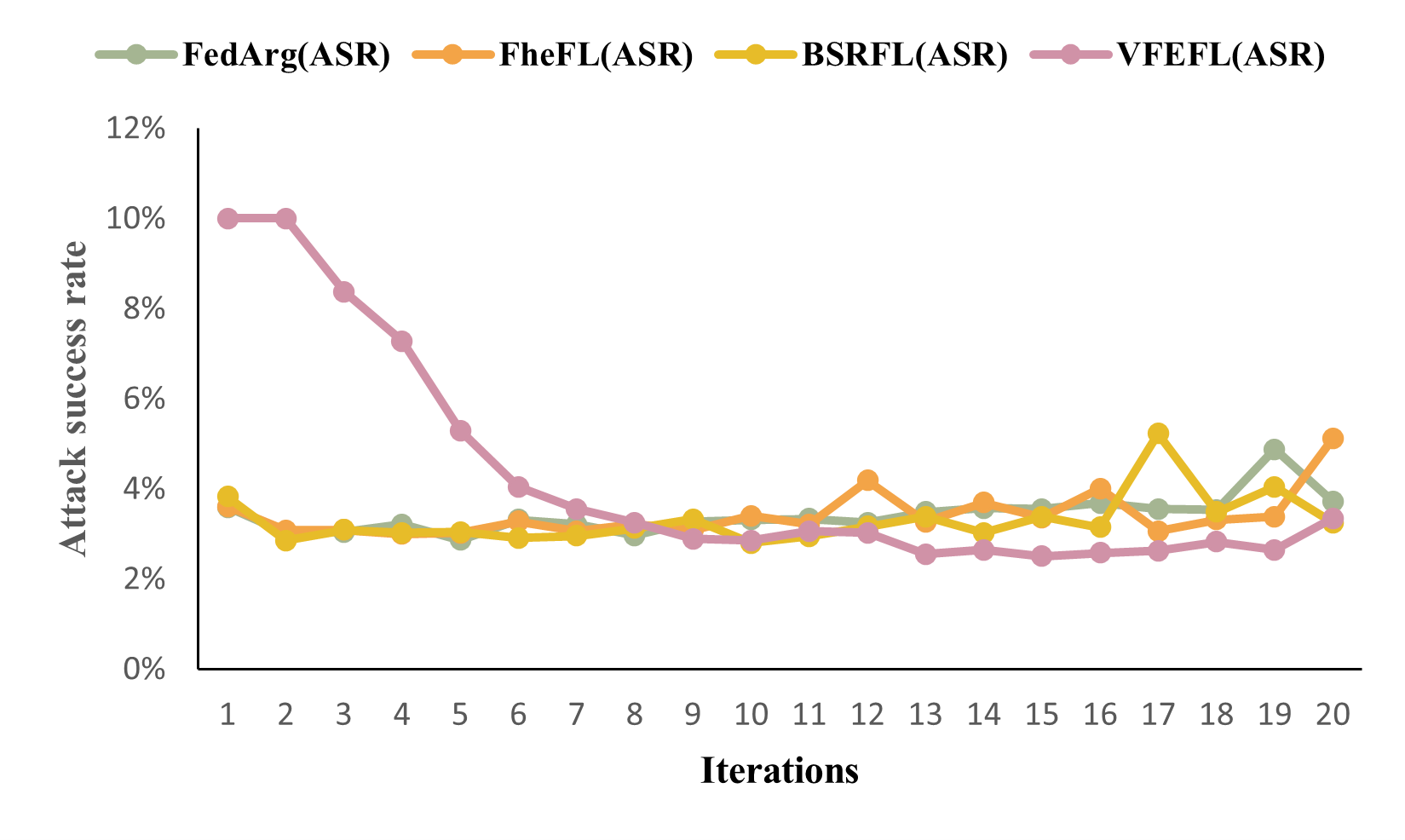}
    \label{fig:lf_cifar_asr}
  }
  \caption{AC and ASR under LF attacks on CIFAR-10: (a) accuracy (AC) and (b) attack success rate (ASR).}
  \label{fig:LF_cifar}
\end{figure}

	Finally, we conduct a comprehensive evaluation of the computational efficiency 
	of the federated learning framework. First, a performance test of the constructed 
	CC-DVFE scheme was conducted to measure the time consumed for each encryption and 
	decryption as well as the generation and verification of ciphertexts and keys, 
	and the results are shown in Table \ref{VFE_time}. 
	The time required for the Setup phase exhibits some uncertainty 
	due to the need for trial and error in generating appropriate prime numbers within the GenClassGroup.
	The runtime of KeyGen abd DKeyGenShare scales linearly with the number of clients, 
	while that of Encryption and VerifyCT increases with the ciphertext dimensionality.
	In the Encryption, hash to point operation is relatively time-consuming,
	however, it only needs to be executed once during the system construction.
	In this experiment, the encryption, key generation, and related verification operations 
	are recorded as the average time for each client.
	
 \begin{table*}[!t]
  \centering
  \caption{Time Consumption of CC-DVFE Under Different Dimensions}
  \label{VFE_time}
  \setlength{\tabcolsep}{8pt} 
  \begin{tabular}{llrrrrr}
    \toprule
    \multicolumn{2}{c}{\textbf{Operations}} & \textbf{m = 10} & \textbf{m = 50} & \textbf{m = 100} & \textbf{m = 500} & \textbf{m = 1000} \\
    \midrule
    \multicolumn{2}{l}{Setup}         & 9.5959 s &  4.8450 s & 9.3008 s & 14.8908 s &  10.5898 s \\
    \multicolumn{2}{l}{Key Generation}        & 1.3672 s &  1.3086 s & 1.3185 s & 1.3033 s  &  1.2764 s \\
    \multirow{3}{*}{Encryption} 
        & Hash to point             & 0.0242 s &  0.1144 s & 0.2174 s & 1.0573 s  & 2.0880 s \\
        & Encryption                  & 0.0077 s &  0.0189 s & 0.0427 s & 0.2123 s  & 0.4761 s \\
        & Ciphertext Proof            & 0.0235 s &  0.0628 s & 0.1080 s & 0.4817 s  & 1.0968 s \\
    \multicolumn{2}{l}{Verify Ciphertext}      & 0.0346 s &  0.0749 s & 0.1240 s & 0.4899 s  & 1.1786 s \\
    \multirow{2}{*}{DKeyGenShare}
        & Key Generation              & 0.0186 s &  0.0189 s & 0.0184 s & 0.0181 s  & 0.0197 s \\
        & Key Proof                   & 0.1401 s &  0.1523 s & 0.1378 s & 0.1351 s  & 0.1477 s \\
    \multicolumn{2}{l}{Verify Functional Key}      & 0.2237 s &  0.2276 s & 0.2174 s & 0.2163 s  & 0.2361 s \\
    \multicolumn{2}{l}{Key Combination}      & 0.0165 s &  0.0204 s & 0.0158 s & 0.0158 s  & 0.0157 s \\
    \multicolumn{2}{l}{Decryption}       & 0.1400 s &  0.9296 s & 1.4903 s & 7.0142 s  & 13.2858 s \\
    \bottomrule
  \end{tabular}
\end{table*}

\begin{table}[!t]
  \centering
  \caption{Time Consumption of VFEFL Per Iteration (LeNet5, MNIST, $n = 10$, $m = 61706$)}
  \label{VFEFL_time}
  \begin{tabular}{llr}
    \toprule
    \multicolumn{2}{c}{\textbf{Operations}} & \textbf{Time (s)} \\
    \midrule
    \multirow{2}{*}{Setup Phase}  
      & Setup & 87.54 \\
      & Communication & 0.43 \\
    \midrule
    \multirow{5}{*}{Model Training Phase} 
      & Local Training & 99.39 \\
      & Hash to Pairing & 128.61 \\
      & Model Encryption and Proof & 87.55 \\
      & Key Share Generation and Proof & 0.16 \\
      & Communication & 12.38 \\
    \midrule
    \multirow{3}{*}{Secure Aggregation Phase} 
      & Ciphertext Verification & 593.93 \\
      & Key Share Verification & 2.35 \\
      & Aggregation & 1014.62 \\
    \bottomrule
  \end{tabular}
\end{table}

	Based on the results of this test, 
	we further evaluate the time required to complete a full training iteration using 
	the MNIST dataset under the VFEFL framework, and the relevant results are 
	demonstrated in Table \ref{VFEFL_time}. 
	In practical experiments, the Setup phase also includes some preprocessing operations related to solving the discrete logarithm problem.
	The communication time of the Setup Phase refers to the time taken for the server to send the initial 
	model and the baseline model to the client. The communication time of the Model Training Phase 
	refers to the time taken for the client to upload the encrypted model, the decryption key share, and 
	corresponding proofs to the server. These communication times are averaged over multiple clients.
	In the aggregation phase, the server performs the verification of all ciphertexts and keys, and the total verification time is computed.
	On the whole, despite the introduction of cryptographic 
	operations, the overall training efficiency of the federated learning framework 
	is still within an acceptable range, which verifies that the proposed scheme can 
	effectively take into account the computational performance 
	while guaranteeing security.

\section{CONCLUSION}
In this paper, we presented VFEFL, based on decentralized verifiable functional encryption.
To implement the framework, we proposed a new aggregation rule and 
a new verifiable functional encryption scheme, and proved their security. 
Our VFEFL scheme achieves privacy protection and resistance to Byzantine 
clients in a single-server setup without introducing additional trusted third parties, 
and is more applicable to general scenarios. 
We analyzed both the privacy and robustness of the scheme, 
in addition to the verifiability, fidelity and self-containment achieved by VFEFL.
Finally, we further demonstrated the effectiveness of the VFEFL scheme 
through experiments.

\bibliographystyle{IEEEtran}
\bibliography{refs}

@article{1,
  title={Embedding the Internet: Wireless integrated network sensors},
  author={Pottie, G. J. and Kaiser, W. J.},
  journal={Communications of the Acm},
  year={2000},
}

@inproceedings{FL2017,
  title={Communication-efficient learning of deep networks from decentralized data},
  author={McMahan, Brendan and Moore, Eider and Ramage, Daniel and Hampson, Seth and y Arcas, Blaise Aguera},
  booktitle={AISTATS},
  pages={1273--1282},
  year={2017},
  organization={PMLR}
}

@inproceedings{Krum,
  author    = {Blanchard, Peva and El Mhamdi, El Mahdi and Guerraoui, Rachid and Stainer, Julien},
  booktitle = {NeurIPS 2017},
  publisher = {Curran Associates, Inc.},
  title     = {Machine Learning with Adversaries: Byzantine Tolerant Gradient Descent},
  year      = {2017}
}

@inproceedings{FLtrust,
  author    = {X. Cao and M. Fang and J. Liu and N. Z. Gong},
  title     = {FLTrust: Byzantine-robust Federated Learning via Trust Bootstrapping},
  booktitle = {NDSS 2021},
  year      = {2021},
  doi       = {10.14722/ndss.2021.24434}
}

@inproceedings{SecureFL,
  author    = {Hao, Meng and Li, Hongwei and Xu, Guowen and Chen, Hanxiao and Zhang, Tianwei},
  title     = {Efficient, Private and Robust Federated Learning},
  year      = {2021},
  isbn      = {9781450385794},
  publisher = {ACM},
  address   = {New York, NY, USA},
  url       = {https://doi.org/10.1145/3485832.3488014},
  doi       = {10.1145/3485832.3488014},
  booktitle = {ACSAC 2021},
  pages     = {45--60},
  numpages  = {16},
  location  = {Virtual Event, USA},
}

@inproceedings{AegisFL,
  title     = {{A}egis{FL}: Efficient and Flexible Privacy-Preserving {B}yzantine-Robust Cross-silo Federated Learning},
  author    = {Chen, Dong and Qu, Hongyuan and Xu, Guangwu},
  booktitle = {ICML 2024},
  pages     = {7207--7219},
  year      = {2024},
  month     = {21--27 Jul},
  publisher = {PMLR},
  url       = {https://proceedings.mlr.press/v235/chen24ag.html}
}

@misc{FheFL,
  title={FheFL: Fully Homomorphic Encryption Friendly Privacy-Preserving Federated Learning with Byzantine Users}, 
  author={Yogachandran Rahulamathavan and Charuka Herath and Xiaolan Liu and Sangarapillai Lambotharan and Carsten Maple},
  year={2024},
  eprint={2306.05112},
  archivePrefix={arXiv},
  primaryClass={cs.AI},
  url={https://arxiv.org/abs/2306.05112}, 
}

@article{BSRFL,
  author   = {Zeng, Honghong and Li, Jie and Lou, Jiong and Yuan, Shijing and Wu, Chentao and Zhao, Wei and Wu, Sijin and Wang, Zhiwen},
  journal  = {IEEE Transactions on Computers},
  title    = {BSR-FL: An Efficient Byzantine-Robust Privacy-Preserving Federated Learning Framework},
  year     = {2024},
  pages    = {2096-2110},
  doi      = {10.1109/TC.2024.3404102}
}

@inproceedings{VMCFE,
  title        = {Verifiable decentralized multi-client functional encryption for inner product},
  author       = {Nguyen, Dinh Duy and Phan, Duong Hieu and Pointcheval, David},
  booktitle    = {CT-RSA 2023},
  pages        = {33--65},
  year         = {2023},
  organization = {Springer}
}

@article{Biscotti,
  author   = {Shayan, Muhammad and Fung, Clement and Yoon, Chris J. M. and Beschastnikh, Ivan},
  journal  = {IEEE Transactions on Parallel and Distributed Systems},
  title    = {Biscotti: A Blockchain System for Private and Secure Federated Learning},
  year     = {2021},
  pages    = {1513-1525},
  doi      = {10.1109/TPDS.2020.3044223}
}

@article{FEFL,
  author   = {Chang, Yansong and Zhang, Kai and Gong, Junqing and Qian, Haifeng},
  journal  = {IEEE Transactions on Information Forensics and Security},
  title    = {Privacy-Preserving Federated Learning via Functional Encryption, Revisited},
  year     = {2023},
  number   = {},
  pages    = {1855-1869},
  doi      = {10.1109/TIFS.2023.3255171}
}

@article{DMCFEFL,
  author   = {Qian, Xinyuan and Li, Hongwei and Hao, Meng and Xu, Guowen and Wang, Haoyong and Fang, Yuguang},
  journal  = {IEEE Transactions on Dependable and Secure Computing},
  title    = {Decentralized Multi-Client Functional Encryption for Inner Product With Applications to Federated Learning},
  year     = {2024},
  pages    = {5781-5796},
  doi      = {10.1109/TDSC.2024.3386357}
}

@inproceedings{MIA,
  author    = {Fredrikson, Matt and Jha, Somesh and Ristenpart, Thomas},
  title     = {Model Inversion Attacks that Exploit Confidence Information and Basic Countermeasures},
  year      = {2015},
  isbn      = {9781450338325},
  publisher = {Association for Computing Machinery},
  address   = {New York, NY, USA},
  url       = {https://doi.org/10.1145/2810103.2813677},
  doi       = {10.1145/2810103.2813677},
  booktitle = {CCS 2015},
  pages     = {1322-1333},
  numpages  = {12},
  location  = {Denver, Colorado, USA},
}

@article{GIA,
  title={Deep leakage from gradients},
  author={Zhu, Ligeng and Liu, Zhijian and Han, Song},
  journal={Advances in neural information processing systems},
  volume={32},
  year={2019}
}

@inproceedings{GIA2,
  author    = {Li, Zhuohang and Zhang, Jiaxin and Liu, Luyang and Liu, Jian},
  booktitle = {CVPR},
  title     = {Auditing Privacy Defenses in Federated Learning via Generative Gradient Leakage},
  year      = {2022},
  pages     = {10122-10132},
  doi       = {10.1109/CVPR52688.2022.00989}
}

@article{OpenFL,
  title     = {Advances and open problems in federated learning},
  author    = {Kairouz, Peter and McMahan, H Brendan and Avent, Brendan and Bellet, Aur{\'e}lien and Bennis, Mehdi and Bhagoji, Arjun Nitin and Bonawitz, Kallista and Charles, Zachary and Cormode, Graham and Cummings, Rachel and others},
  journal   = {Foundations and trends{\textregistered} in machine learning},
  pages     = {1--210},
  year      = {2021},
  publisher = {Now Publishers, Inc.}
}

@article{ESFL,
  author   = {Miao, Yinbin and Xie, Rongpeng and Li, Xinghua and Liu, Zhiquan and Choo, Kim-Kwang Raymond and Deng, Robert H.},
  journal  = {IEEE Transactions on Dependable and Secure Computing},
  title    = {Efficient and Secure Federated Learning Against Backdoor Attacks},
  year     = {2024},
  pages    = {4619-4636},
  doi      = {10.1109/TDSC.2024.3354736}
}

@article{RVPFL,
  author   = {Lu, Zhi and Lu, Songfeng and Tang, Xueming and Wu, Junjun},
  journal  = {IEEE Transactions on Artificial Intelligence},
  title    = {Robust and Verifiable Privacy Federated Learning},
  year     = {2024},
  pages    = {1895-1908},
  doi      = {10.1109/TAI.2023.3309273}
}

@inproceedings{DMCFE,
  title        = {Decentralized multi-client functional encryption for inner product},
  author       = {Chotard, J{\'e}r{\'e}my and Dufour Sans, Edouard and Gay, Romain and Phan, Duong Hieu and Pointcheval, David},
  booktitle    = {ASIACRYPT 2018},
  pages        = {703--732},
  year         = {2018},
  organization = {Springer}
}

@article{SecProbe,
  author   = {Zhao, Lingchen and Wang, Qian and Zou, Qin and Zhang, Yan and Chen, Yanjiao},
  journal  = {IEEE Transactions on Information Forensics and Security},
  title    = {Privacy-Preserving Collaborative Deep Learning With Unreliable Participants},
  year     = {2020},
  pages    = {1486-1500},
  doi      = {10.1109/TIFS.2019.2939713}
}

@inproceedings{DLeasy,
  author    = {Castagnos, Guilhem and Laguillaumie, Fabien},
  title     = {Linearly Homomorphic Encryption from DDH},
  booktitle = {CT-RSA 2015},
  year      = {2015},
  publisher = {Springer International Publishing},
  address   = {Cham},
  pages     = {487--505},
  isbn      = {978-3-319-16715-2}
}

@inproceedings{TPECDSA,
  author    = {Castagnos, Guilhem and Catalano, Dario and Laguillaumie, Fabien and Savasta, Federico and Tucker, Ida},
  title     = {Two-Party ECDSA from Hash Proof Systems and Efficient Instantiations},
  year      = {2019},
  isbn      = {978-3-030-26953-1},
  publisher = {Springer-Verlag},
  address   = {Berlin, Heidelberg},
  doi       = {10.1007/978-3-030-26954-8_7},
  booktitle = {CRYPTO 2019},
  pages     = {191-221},
  numpages  = {31},
  location  = {Santa Barbara, CA, USA}
}

@inproceedings{Bandwidth,
  author    = {Castagnos, Guilhem and Catalano, Dario and Laguillaumie, Fabien and Savasta, Federico and Tucker, Ida},
  title     = {Bandwidth-Efficient Threshold EC-DSA},
  booktitle = {PKC 2020},
  year      = {2020},
  publisher = {Springer International Publishing},
  address   = {Cham},
  pages     = {266--296},
  isbn      = {978-3-030-45388-6}
}

@inproceedings{FiatShamir,
  author    = {Fiat, Amos and Shamir, Adi},
  title     = {How to prove yourself: practical solutions to identification and signature problems},
  year      = {1987},
  isbn      = {0387180478},
  publisher = {Springer-Verlag},
  address   = {Berlin, Heidelberg},
  booktitle = {CRYPTO 1986},
  pages     = {186-194},
  numpages  = {9},

}

@inproceedings{VFE,
  author    = {Badrinarayanan, Saikrishna
               and Goyal, Vipul
               and Jain, Aayush
               and Sahai, Amit},
  title     = {Verifiable Functional Encryption},
  booktitle = {ASIACRYPT 2016},
  year      = {2016},
  publisher = {Springer Berlin Heidelberg},
  address   = {Berlin, Heidelberg},
  pages     = {557--587},
  isbn      = {978-3-662-53890-6}
}

@inproceedings{xu2022agic,
  title={Agic: Approximate gradient inversion attack on federated learning},
  author={Xu, Jin and Hong, Chi and Huang, Jiyue and Chen, Lydia Y and Decouchant, Jeremie},
  booktitle={SRDS 2022},
  pages={12--22},
  year={2022},
  organization={IEEE}
}

@inproceedings{Bulletproofs,
  author    = {Bünz, Benedikt and Bootle, Jonathan and Boneh, Dan and Poelstra, Andrew and Wuille, Pieter and Maxwell, Greg},
  booktitle = {IEEE S\&P 2018},
  title     = {Bulletproofs: Short Proofs for Confidential Transactions and More},
  year      = {2018},
  pages     = {315-334},
  doi       = {10.1109/SP.2018.00020}
}

@article{BSGS,
  title={A modification of Shanks' baby-step giant-step algorithm},
  author={Terr, David},
  journal={Mathematics of Computation},
  pages={767--773},
  year={2000}
}

@inproceedings{FE_def,
  author    = {Boneh, Dan
               and Sahai, Amit
               and Waters, Brent},
  title     = {Functional Encryption: Definitions and Challenges},
  booktitle = {TCC 2011},
  year      = {2011},
  publisher = {Springer Berlin Heidelberg},
  address   = {Berlin, Heidelberg},
  pages     = {253--273},
  isbn      = {978-3-642-19571-6}
}

@article{FE,
  title={Functional encryption: a new vision for public-key cryptography},
  author={Boneh, Dan and Sahai, Amit and Waters, Brent},
  journal={Communications of the ACM},
  volume={55},
  number={11},
  pages={56--64},
  year={2012},
  publisher={ACM New York, NY, USA}
}

@inproceedings{MIFE2014,
  title        = {Multi-input functional encryption},
  author       = {Goldwasser, Shafi and Gordon, S Dov and Goyal, Vipul and Jain, Abhishek and Katz, Jonathan and Liu, Feng-Hao and Sahai, Amit and Shi, Elaine and Zhou, Hong-Sheng},
  booktitle    = {EUROCRYPT},
  pages        = {578--602},
  year         = {2014},
  organization = {Springer}
}

@inproceedings{MIFE2017,
  title        = {Multi-input inner-product functional encryption from pairings},
  author       = {Abdalla, Michel and Gay, Romain and Raykova, Mariana and Wee, Hoeteck},
  booktitle    = {EUROCRYPT},
  pages        = {601--626},
  year         = {2017},
  organization = {Springer}
}

@inproceedings{fang2020local,
  title     = {Local model poisoning attacks to Byzantine-Robust federated learning},
  author    = {Fang, Minghong and Cao, Xiaoyu and Jia, Jinyuan and Gong, Neil},
  booktitle = {USENIX Sec.},
  pages     = {1605--1622},
  year      = {2020}
}

@inproceedings{Fiat-shamir,
  author = {Bellare, Mihir and Rogaway, Phillip},
  title = {Random oracles are practical: a paradigm for designing efficient protocols},
  year = {1993},
  isbn = {0897916298},
  url = {https://doi.org/10.1145/168588.168596},
  doi = {10.1145/168588.168596},
  booktitle = {CCS 1993},
  pages = {62–73},
  numpages = {12},
  publisher = {Association for Computing Machinery},
}

{\appendices
	\section{assumptions and security models}
	\begin{definition}[The Symmetric eXternal Diffie-Hellman Assumption]
		SXDH Assumption states that, in a pairing  
		group $\mathcal{PG}$, the DDH assumption holds in both 
		$\mathbb{G}_1$ and $\mathbb{G}_2$.
	\end{definition}
	\begin{definition}[HSM assumption]
		Let \( \text{GenClassGroup} = (\texttt{Gen}, \texttt{Solve}) \) be a generator for 
		DDH groups with an easy DL subgroup. Let \( (\tilde{s}, f, \hat{h}_p, \hat{G}, F) \) 
		be the output of \( \text{Gen}(\lambda, p) \), and set \( h := h_p f \). 
		Denote by \( D \) (resp. \( \mathcal{D}_p \)) a distribution over integers such that 
		\( \{h^x, x \leftarrow D\} \) (resp. \( \{\hat{h}_p^x, x \leftarrow \mathcal{D}_p\} \)) 
		is within statistical distance \( 2^{-\lambda} \) from the uniform distribution 
		in \( \langle h \rangle \) (resp. \( \langle \hat{h}_p \rangle \)). 
		For any adversary \( \mathcal{A} \) against the HSM problem, its advantage is defined as:
		\begin{align*}
            &\text{Adv}_{\mathcal{A}}^{\text{HSM}}(\lambda) := \\
            &\left| \Pr \left[
            \begin{array}{l}
            \begin{matrix}
            (\tilde{s}, f, \hat{h}_p, h_p, F, \hat{G}^p) \\ 
            \leftarrow \texttt{Gen}(\lambda, p), \\
            t \leftarrow \mathcal{D}_p, \, h_p = \hat{h}_p^t, \, x, x' \\ 
            \leftarrow D, \, b \overset{\$}{\leftarrow} \{0, 1\}, \\
            Z_0 = h^x, \, Z_1 = h^{x'}, \, b' \leftarrow \\  
            \mathcal{A}(p, \tilde{s}, f, \hat{h}_p, h_p, F, 
            \hat{G}^p, Z_b, \texttt{Solve}(\cdot)) 
            \end{matrix}: b = b'
            \end{array}
            \right] - \frac{1}{2} \right|.
        \end{align*}	
		The HSM problem is hard in \( G \) if for all probabilistic polynomial-time 
		adversaries \( \mathcal{A} \), \( \text{Adv}_{\mathcal{A}}^{\text{HSM}}(\lambda) \) is negligible.
	\end{definition}
		From \cite{VMCFE}, one can set $S=2^{\lambda-2} \tilde{s}$, and instantiate $\mathcal{D}_p$ as 
		the uniform distribution on $\{0, \cdots, S\}$ and $\mathcal{D}$ as the uniform distribution 
		on $\{0, ..., pS\}$.
	\begin{definition}[Decisional Diffie-Hellman Assumption]
		The Decisional Diffie-Hellman (DDH) assumption states that, in a prime-order group 
		\( \mathcal{G} \leftarrow \text{GGen}(\lambda) \), no probabilistic polynomial-time 
		(PPT) adversary can distinguish between the following two distributions with 
		non-negligible advantage, namely
		\begin{align*}
            \mathrm{Adv}_{\mathbb{G}}^{ddh}(t)=
			|& 
				\Pr \left[ (g^a, g^b, g^{ab}) \mid a, b \leftarrow \mathbb{Z}_p \right] \\
				&-
				\Pr \left[ (g^a, g^b, g^c) \mid a,b,c \leftarrow \mathbb{Z}_p\right]
			|< \epsilon(\lambda)
        \end{align*}
		% \[
		% \{(g^a, g^b, g^{ab}) \mid a, b \leftarrow \mathbb{Z}_p\}
		% \quad \text{and} \quad
		% \{(g^a, g^b, g^c) \mid a,b,c \leftarrow \mathbb{Z}_p\}.
		% \]
		% Equivalently, the DDH assumption implies it is hard to distinguish, 
		% given \(g^{\hat{a}}\), for \(\hat{a} = (1, a)\), a random element from the 
		% span of \(g^{\hat{a}}\), from a random element in \(\mathbb{G}^2\)\cite{DMCFE}:  
		% \[
		% 	(g^{\hat{a}})^b = g^{{\hat{a}}b} = (g^b, g^{ab}) \approx (g^b, g^c).
		% \]
	\end{definition}
	
	\begin{definition}[Multi DDH Assumption\cite{DMCFE}]\label{def_multiDDH}
		For any distinguisher $\mathcal{A}$ with runtime at most $t$, its maximum advantage in distinguishing the following two distributions:
		$$
		\mathcal{D}_m = \left\{ \left(g^x, (g^{y_j}, g^{x y_j})_{j \in [m]} \right) \mid x, y_j \overset{\$}{\leftarrow} \mathbb{Z}_p \right\},
		$$
		$$
		\mathcal{D}'_m = \left\{ \left(g^x, (g^{y_j}, g^{z_j})_{j \in [m]} \right) \mid z_j \overset{\$}{\leftarrow} \mathbb{Z}_p \right\}
		$$
		is upper bounded by
		$$
		\operatorname{Adv}^{\mathrm{multi-ddh}}_{\mathbb{G}}(t) \le \operatorname{Adv}^{\mathrm{ddh}}_{\mathbb{G}}(t + 4 m \cdot t_{\mathbb{G}}),
		$$
		where $t_{\mathbb{G}}$ denotes the time for an exponentiation in $\mathbb{G}$.

		% Let $\mathcal{G} = (\mathbb{G}, p, g)$, $g$ be a generator of $\mathbb{G}$
		% and $\mathrm{Adv}_{\mathbb{G}}^{ddh}(t)$ denote the best advantage of 
		% breaking the DDH assumption in $\mathbb{G}$ within time $t$.
		% Given $X=g^x, Y_j=g^{y_j} \overset{\$}{\leftarrow} \mathbb{G}, j = 1, \dots, m $,
		% for any distinguisher \( \mathcal{A} \) running within time \( t \), the best advantage \( A \) can get in distinguishing
		% $\mathcal{D}_m = \{ X, (Y_j, g^{xy_j})_j\}$ 
		% from  
		% $\mathcal{D}'_m = \{ (X, (Y_j, Z_j)_j) \mid Z_j \overset{\$}{\leftarrow} \mathbb{G}, j = 1, \dots, m \}$ 
		% is bounded by 
		% \( \operatorname{Adv}^{\operatorname{ddh}}(t + 4m \times t_{\mathbb{G}}) \)
		% (for a detailed proof, see \cite{DMCFE}), 
		% where \( t_{\mathbb{G}} \) is the time for an exponentiation in \( \mathbb{G} \). 
		% Namely,
        % \begin{align*}
        %     & \left|
		% 		\Pr\left[\mathcal{A}\left(
		% 			X, (Y_j, g^{xy_j})_j
		% 		\right)=1 \right]
		% 		-
		% 		\Pr\left[
		% 			\mathcal{A}\left(
		% 			X, (Y_j, Z_j)_j
		% 			\right)=1 
		% 		\right]
		% 	\right|  \\
        %     \le & \operatorname{Adv}^{\operatorname{ddh}}_{\mathbb{G}}(t + 4m \times t_{\mathbb{G}}).
        % \end{align*}
	\end{definition}
	\begin{definition}[Static-IND-Security for CC-DVFE]\label{definition:ind}
		Consider a CC-DVFE scheme applied to a set of \(n\) clients, 
		each associated with an 
		\(m\)-dimensional vector. No PPT adversary \(\mathcal{A}\) should be able to 
		win the following security game with at least two non-corrupted clients,
		against a challenger \(\mathcal{C}\) with non-negligible probability.
		\begin{itemize}
			\item \textbf{Initialization:} The challenger \(\mathcal{C}\) runs 
				algorithm \texttt{SetUp}($\lambda$) and \texttt{KeyGen}($pp$) to get 
				public parameters $pp$ and the keys 
				$\left((sk_i, ek_i)_{i \in [n]}, vk_{CT}, vk_{DK}, pk\right)$ and chooses a random 
				bit $b \overset{\$}{\gets} \{0, 1\}$. It sends $\left(vk_{CT}, vk_{DK}, pk\right)$ to the $\mathcal{A}$. 
			\item \textbf{Encryption queries QEncrypt:} 
				$X^0, X^1$ are two different $n \times m$ dimensional plaintext matrices 
				corresponding to $m$-dimensional messages from $n$ clients. 
				$\mathcal{A}$ can possesses $q_E$ access to a Left-or-Right 
				encryption oracle with submitting 
				$\left(id, X^0, X^1, \ell,\left\{\ell_{Enc}\right\}_{j \in [m]}\right)$. 
				$\mathcal{C}$ generates the ciphertext  $C_{\ell,i}$ and proof $\pi_{CT, id}$
				by running \texttt{Encrypt}$\left(ek_{id}, X_{id}^b, aux, \ell, (\ell_{Enc,j})_{j \in [m]}, pp\right)\to (C_{\ell, id}, \pi_{CT,id})$ 
				and provides them to $\mathcal{A}$. Any subsequent queries for 
				the same pair $(\ell, id)$ will later be ignored.
				$\mathcal{A}$ obtaining the validation key $vk_{CT}$ in the 
				Initialization phase can implement the ciphertextvalidation itself 
				without querying \texttt{VerifyCT}($(C_{\ell,id})_{id \in [n]},\pi_{CT}, vk_{CT}$).
			\item \textbf{Functional decryption key queries QDKeyGen:} 
				$\mathcal{A}$ can access to QDKeyGen 
				with $\left(id, y_{id}, \ell_{y}\right)$ up to $q_K$ times.
				$\mathcal{C}$ runs \texttt{DKeyGenShare}$(sk_{id}, pk, y_{id}, \ell_{y}, pp ) \to (dk_{y,id},$\allowbreak$ \pi_{DK,id})$
				obtaining the decryption key validation key $vk_{DK}$ in the Initialization phase 
				can implement the validation itself without querying
			\texttt{VerifyDK}($({dk}_{y,id})_{id\in [n]},$\allowbreak$\pi_{DK}, vk_{DK}, pp$).
			\item \textbf{Corruption queries QCorrupt:} $\mathcal{A}$ 
				can  corruption queries by submitting the index $id \in \mathcal{CC}$
				up to $q_C$ times.
				$\mathcal{C}$ returns  
				secret and encryption keys $(sk_{id}, ek_{id})$ to $\mathcal{A}$. 
				QCorrupt are sent before the initialization in the sta-IND game.
			\item \textbf{Finalize:} $\mathcal{A}$ outputs its guess $b'$ on the bit $b$, 
			and this process produces the outcome $\beta$ of the security game, based on 
			the analysis provided below.
			Let $\mathcal{CC}$ represent the set of corrupted clients (the indices $id$ that are 
			input to QCorrupt throughout the entire game), while $\mathcal{HC}$ denotes the set 
			of honest (noncorrupted) clients. We set output $\beta = (b'== b)$ if none of 
			the following occur, otherwise, $\beta \overset{\$}{\gets} \{0,1\}$
			\begin{enumerate}
				\item A \texttt{QEncrypt}($id, X^0_{id}, X^1_{id}, \ell,\left\{\ell_{Enc}\right\}_{j \in [m]}$)-query has been 
						made for an index $(id \in \mathcal{CC})$ where $(X^0_{id} \neq X^1_{id})$.
				% \item A \texttt{QEncrypt}($i, X^0_i, X^1_i, \ell,\left\{\ell_{Enc}\right\}_{j \in [m]}$)-query has been 
				% 		made for an index $(i \in \mathcal{HC})$ where 
				% 		$\frac{\langle X^0_i, \vec{x}_0 \rangle}{\langle X^0_i,X^0_i \rangle} \ne 
				% 		\frac{\langle X^1_i, \vec{x}_0 \rangle}{\langle X^1_i,X^1_i \rangle}$
				\item For a given label $\ell$, an encryption query 
					\texttt{QEncrypt}$(id, X^0_{id}, X^1_{id}, \ell,\left\{\ell_{Enc}\right\}_{j \in [m]})$ has been submitted for some
					$id \in \mathcal{HC}$, while not all encryption queries 
					\texttt{QEncrypt}($id, X^0_{id}, X^1_{id}, \ell,\left\{\ell_{Enc}\right\}_{j \in [m]}$) have been made 
					for honest clients.
				\item For a given label \(\ell\) and a function vector $\vec{y}$ asked for \texttt{QDKeyGen},
					for all $id \in \mathcal{CC}$, $X_{id}^0 = X_{id}^1$ and for all $id' \in \mathcal{HC}$,
					\texttt{QEncrypt}($id', X^0_{id'}, X^1_{id'}, \ell,\left\{\ell_{Enc}\right\}_{j \in [m]}$) have been made,
					but there exists $ j \in [m]$ such that 
					$\langle X^{\top 0}_j, \vec{y} \rangle \neq \langle X^{\top 1}_j, \vec{y} \rangle$ holds,
					where $X^{\top 0}_j = \{X^0_{1,j}, \cdots, X^0_{n,j}\} \in \mathcal{M}^n$, $X^{\top 1}_j = \{X^1_{1,j}, \cdots, X^1_{n,j}\} \in \mathcal{M}^n$
				\item For a given label $\ell$, there exists a cross-ciphertext statement $\mathcal{R}$ such that an encryption query \texttt{QEncrypt}$(i, X_{id}^0, X_{id}^1, \ell, {\ell_{Enc}}_{j \in [m]})$ has been submitted for some $id \in \mathcal{HC}$, while $\mathcal{R}(X_{id}^0) \neq \mathcal{R}(X_{id}^1)$,  where $\mathcal{R}(\cdot)$ denotes the relation verified in the cross-ciphertext proof.
			\end{enumerate}
		\end{itemize}
		We say the CC-DVFE is $\epsilon(\lambda)$-static-IND-secure 
		if there exists no adversary $\mathcal{A}$ who can win the above game  
		with advantage at least $\epsilon(\lambda)$.
		Namely,
		$$
			\left| \mathrm{Adv_{CC-DVFE}^{sta-ind}} \right| = \left| \mathrm{Pr}[b'=b]- \frac{1}{2} \right| \le \epsilon(\lambda)
		$$
	\end{definition}

		\begin{definition}[Verifiability for CC-DVFE]\label{definition:ver}
			Consider a CC-DVFE scheme applied to a set of \(n\) clients, each associated with an 
		\(m\)-dimensional vector. No PPT adversary \(\mathcal{A}\) should be able to 
		win the following security game against a challenger $\mathcal{C}$ with non-negligible 
		probability.
		\begin{itemize}
			\item \textbf{Initialization:} Challenger $\mathcal{C}$ initializes by running 
			$\texttt{SetUp}(\lambda) \to pp$. It sends $pp$ to $\mathcal{A}$.
			\item \textbf{Key generation queries QKeyGen:} For only one time in the game, 
			$\mathcal{A}$ can play on behalf of corrupted clients $\mathcal{CC}$ to join a 
			\texttt{KeyGen($pp$)}. The challenger $\mathcal{C}$ executes \texttt{KeyGen($pp$)} on 
			representations of non-corrupted clients $\mathcal{HC}$ and 
			finally output $ (vk_{CT}, vk_{DK}, pk)$.
			\item \textbf{Corruption queries QCorrupt:} $\mathcal{A}$ can possesses 
			$q_C$ access to the corruption queries with submitting an index $id$ to play on behalf of client $id$ in the protocol.
			$\mathcal{C}$ sends $(sk_{id}, ek_{id}, vk_{CT}, vk_{DK})$ to $\mathcal{A}$ if $id$ is queried after QKeyGen.
			And $\mathcal{A}$ cannot play on behalf of client $id$ in the key generation process anymore.
			\item \textbf{Encryption queries QEncrypt:} 
			$\mathcal{A}$ has $q_E$ access to QEncrypt to obtain a correct encryption with submitting 
			an index $id$, message $\vec{x}_{id} \in \mathcal{M}^m$ and labels $\ell$ and 
			$\left(\ell_{Enc,j}\right)_{j \in [m]}$. $\mathcal{C}$ runs 
			\texttt{Encrypt}$(ek_{id}, \vec{x}_{id},$\allowbreak$ aux, \ell, (\ell_{Enc,j})_{j \in [m]}, pp)\to (C_{\ell, id}, \pi_{CT,id})$
			and sends them to $\mathcal{A}$.
			$\mathcal{A}$ obtaining the 
			validation key $vk_{CT}$ in the key generation phase can implement 
			the ciphertext validation itself without querying  \texttt{VerifyCT}($(C_{\ell,id})_{id \in [n]},$\allowbreak$\pi_{CT}, vk_{CT}$).
			\item \textbf{Functional key share queries QDKeyGen:} $\mathcal{A}$ has 
			$q_K$ access to QDKeyGen to gain a correct functional key share
			with sunbmitting $\left(id, \ell_y, y_{id}\right)$.
			$\mathcal{C}$ runs \texttt{DKeyGenShare}$\left(sk_{id}, pk, y_{id}, \ell_{y}, pp \right) 
			\to (dk_{y,id},$\allowbreak$ \pi_{DK,id})$ and sends them to $\mathcal{A}$. 
			$\mathcal{A}$ obtaining the decryption key validation 
			key $vk_{DK}$ in the key generation phase 
			can implement the validation itself without querying
			\texttt{VerifyDK}($({dk}_{y,id})_{id\in [n]},\pi_{DK}, vk_{DK}, pp$).
			\item \textbf{Finalize:} Let $\mathcal{CC}$ be the set of corrupted 
			clients and $\mathcal{HC}$ denotes the set of honest clients, 
			then $\mathcal{A}$ has to output verification keys $(vk_{CT}, vk_{DK})$ 
			and public key $pk$ from QKeyGen(), a set of labels $\ell$ and 
			$\left(\ell_{Enc,j}\right)_{j \in [m]}$ , malicious ciphertexts 
			$(C_{\ell,id}, \pi_{CT, id})_{id\in \mathcal{CC}}$, and malicious functional key shares 
			$(dk_{y,id}, \pi_{DK,id})_{id\in \mathcal{CC}}$ for a polynomially number $q_F$ of 
			function vectors $\{y_f\}_{ f \in [q_F]}$. 
			The ciphertexts of honest clients
			$\{C_{\ell,id}\}_{id\in \mathcal{HC}}$ and their functional key shares 
			$\{dk_{y,id}\}_{id\in \mathcal{HC}}$ are automatically completed by using the oracles 
			QEncrypt and QDKeyGen.
			$\mathcal{A}$ wins the game if one of the following cases happens:
			\begin{enumerate}
				\item If $\texttt{VerifyCT}((C_{\ell,id})_{id \in [n]},\pi_{CT, id}, vk_{CT}) \to 0$ or 
				$\texttt{VerifyDK}((dk_{id,y})_{id\in [n]}, \pi_{DK,i}, vk_{DK},  pp) \to 0$, 
				for some $y_f$, the union $\mathcal{CC}=\mathcal{CC}_{CT} \cup  \mathcal{CC}_{DK}$ 
				contains an honest sender $\mathcal{C}_i$.
				In other words,  an honest client is falsely accused of being a corrupt client, 
				which is noted as event E1.
				% \item If $\texttt{VerifyCT}((C_{\ell,i})_{i \in [n]},\pi_{CT, i}, vk_{CT})=1$ and 
				% $\texttt{VerifyDK}((dk_{y,i})_{i\in [n]}, \pi_{DK,i},vk_{DK},\\ pp) = 1$ hold, 
				% there does not exist a tuple of messages $(X_i)_{i\in[n]}$ such that, 
				% % for all $i\in [n]$,\\
				% $\texttt{Decrypt}(\left\{C_{\ell,i}\right\}_{i \in [n]} , h^{dk_{y_f}}, y_f) =
				% \{\langle \vec{x}^{\top}, y_f \rangle\}_{j \in [m]} $ 
				% with $dk_{y_f} = \texttt{DKeyComb}(\{dk_{i,y_f}\}_{i \in [n]}, \ell_{y},pk)$ 
				% for all function vector $y_f$ with $f=1,\cdots, q_F$,  which is noted as event E2.
				\item If $\texttt{VerifyCT}((C_{\ell,id})_{id \in [n]},\pi_{\text{CT},i}, vk_{\text{CT}}) \to 1 $
					and $ \texttt{VerifyDK}((dk_{id,y})_{id \in [n]}, \pi_{\text{DK},i}, vk_{\text{DK}}, pp) \to 1$
					there does not exist a tuple of messages $(\vec{x}_{id})_{id \in [n]}$ such that 
					$\texttt{Decrypt}((C_{\ell,id})_{id \in [n]} , {dk}, \vec{y}) \to( \langle \vec{x}^{\top}, \vec{y}\rangle)_{j \in [m]}$
					holds where $\vec{x}^{\top} = (x_{1,j}, \dots, x_{n,j}) $\allowbreak$\in \mathcal{M}^n$
					for all function vectors $y_f$ with $f \in [q_F]$, 
					where $ dk = \texttt{DKeyComb}(\{dk_{i,y_f}\}_{i \in [n]}, \ell_y, pk)$. 
					This defines event E2. 
			\end{enumerate}
			The first winning case guarantees that 
			no honest client will be misjudged when authentication fails.
			The second winning condition is determined statistically, and the scheme must ensure 
			that the adversary cannot generate a ciphertext and decryption key that can be 
			verified, causing the decryption to fail.
		\end{itemize}
		We say the CC-DVFE has $\epsilon(\lambda)$-verifiability 
		if there exists no adversary $\mathcal{A}$ who can win the above game  
		with advantage at least $\epsilon(\lambda)$.
		Namely,
		$$
			\left| \mathrm{Adv_{CC-DVFE}^{ver}} \right| = 
			\left| \text{Pr}(\text{E1}) + \text{Pr}(\text{E2}) \right| \le \epsilon(\lambda).
		$$
		\end{definition}

	% \section{Proof of Theorem \ref{theorem:ind}} \label{proof:ind}	

	% \section{Proof of Theorem \ref{theorem:ver}} \label{proof:ver}

		\section{Instantiation of $\pi_{CT}$}\label{zkp_enc}
		In the two algorithms \texttt{Encrypt} and \texttt{VerifyCT} of the above CC-DVFE scheme, 
		we implement the verification through the following zero-knowledge proof.
		For label $\ell$ and encryption label $\ell_{Enc}$, the vector $\vec{x}_i \in \mathcal{M}^m$, 
		auxiliary vector $\vec{x}_0 \in \mathcal{M}^m$ and commit of key $\mathrm{com}_i$,
		we define the relation:
		$$
			\mathcal{R}_{Encrypt}(V_i, \mathrm{com}_i, \ell, \ell_{Enc}, y_i, \vec{x}_0)=1 \leftrightarrow 
			\begin{cases}
				V_i = (\bar{u})^{s_{i}} \cdot \vec{w}^{\vec{x}_i}\\
				y_i = \frac{\langle \vec{x}_i, \vec{x}_0\rangle}{\langle \vec{x}_i, \vec{x}_i\rangle} \\
				\mathrm{com}_i = (v)^{s_{i}}
			\end{cases}.
		$$
		where $\bar{u}=\prod_{j}^m u_j$, $u_j = \mathcal{H}_1(\ell||j)$,  
		$\vec{w}^{{x}_i} = \prod_{j \in [m]} (w_j)^{x_{i,j}} $,
		$\{w_j\}_{j \in [m]} = \{\mathcal{H}_1^{'}(\ell_{Enc,j}||j)\}_{j \in [m]}$,
		$v = \mathcal{H}_1(\ell_{D})$.\\
		We proposes a zero knowledge proof system for the relation as described above:
		\begin{equation*}
            \text{PoK}
			\left\{(s_i, \vec{x}_i): 
                \begin{aligned}
				V_i = (\bar{u})^{s_{i}} \cdot \vec{w}^{\vec{x}_i} \wedge \\
				y_i = \frac{\langle \vec{x}_i, \vec{x}_0\rangle}{\langle \vec{x}_i, \vec{x}_i\rangle} \wedge \\
				\mathrm{com}_i = (v)^{s_{i}}.
                \end{aligned}
			\right\}
        \end{equation*}
		To prove the statement, the prover $\mathcal{P}$ and $\mathcal{V}$ engage 
		in the following zero knowledge protocol. \\
		$\mathcal{P}$ computes:
		\begin{align*}
			&\rho \overset{\$}{\leftarrow} \mathbb{Z}^2_p, \quad \vec{k} \overset{\$}{\leftarrow}  \mathbb{Z}_p^m \\
			&K = (\bar{u})^\rho \vec{w}^{\vec{k}}, \quad L(X) = \vec{x}_i + \vec{k}X  \\
			& t_0=\langle \vec{x}_i,\vec{x}_i \rangle,\quad t_1 = 2\langle \vec{k},\vec{x}_i \rangle, \quad t_2= \langle \vec{k},\vec{k} \rangle \\
			& t(X) = \langle L(X),  L(X)\rangle = t_0 + t_1 X + t_2 X^2 \\
			&\tau_0, \ \tau_1, \ \tau_2, \ \tau'_1 \gets \mathbb{Z}^2_p \\
			&T_0 = g^{t_0} (\bar{u})^{\tau_0}, \quad T_1 = g^{t_1} (\bar{u})^{\tau_1}, \quad T_2 = g^{t_2} (\bar{u})^{\tau_2}\\
			&t'_1 = \langle \vec{k}, \vec{x}_0 \rangle, \quad T'_1= g^{t'_1} (\bar{u})^{\tau'_1}\\
			&\mathrm{com}^* = (v)^{\rho}, \quad v = \mathcal{H}_1(\ell_{D})
		\end{align*}
		$\mathcal{P} \to \mathcal{V}: K, T_0, T_1, T_2, T'_1, \mathrm{com}^*_i$ \\
		$\mathcal{V}: \alpha \overset{\$}{\gets} \mathbb{Z}^*_p$ \\
		$\mathcal{V} \to \mathcal{P}: \alpha$ \\
		$\mathcal{P}$ computes:
		\begin{align*}
			&\tau_\alpha = \tau_0 + \tau_1 \alpha + \tau_2 \alpha^2,\quad \tau'_\alpha = \tau_0 y_i + \tau'_1 \alpha\\
			&\omega = s_i + \alpha \rho  \in \mathbb{Z}^2_p, \quad L = L(\alpha), \quad t = t(\alpha)
		\end{align*}
		$\mathcal{P} \to \mathcal{V}: \tau_\alpha, \tau'_\alpha, \omega, L, t$ \\
		$\mathcal{V}$ verifies:
		\begin{align*}
			g^t (\bar{u})^{\tau_{\alpha}} & \overset{?}{=} T_0 \left(T_1\right)^\alpha\left(T_2\right)^{\alpha^2} \\
			g^{\langle L, \vec{x}_0\rangle} (\bar{u})^{\tau'_{\alpha}} & \overset{?}{=} T_0^{y_i} (T'_1)^\alpha \\
			\vec{w}^L (\bar{u})^\omega & \overset{?}{=} V_i(K)^\alpha \\
			t & \overset{?}{=} \langle L, L\rangle \\
			(v)^{\omega} & \overset{?}{=} \mathrm{com}_{i} (\mathrm{com}^*)^{\alpha}
		\end{align*}
		
		\begin{theorem}\label{theorem:pi_enc}
			The Ciphertext Verification $\prod_{Encrypt}$  has perfect completeness, 
			perfect honest verifier zero-knowledge and computational special soundness.
		\end{theorem}
		\begin{proof}
			A formal proof are provided in Appendix \ref{proof:pi_enc}.
		\end{proof}

	\section{Instantiation of  $\pi_{DK}$}\label{zkp_dkg}
		In the two algorithms \texttt{DKeyGenShare} and \texttt{VerifyDK} of the above CC-DVFE scheme, 
		we implement the verification through the following zero-knowledge proof.
		To prevent malicious clients from interfering with the aggregation process 
		by uploading malicious keys, we define the relation:
        \begin{align*}
            \mathcal{R}_{DKeyGenShare}(\left(T_i\right)_{i \in [n]}, d_i,dk_i, \ell_y, y_i)=1 
            \\ 
            \leftrightarrow 
			\begin{cases}
				T_i = \left(h_p^{t_{i,b}}\right)_{b \in [2]} \\
				d_i = \left(f^{\hat{k}_b} K_{\sum,i,b}^{t_{i,b}}\right)_{b \in [2]} \\
				dk_i = \left( (\hat{v}_b)^{\hat{k}} \cdot h^{s_{i,b} \cdot y_i}\right)_{b \in [2]}\\
				\mathrm{com}_i = (v)^{s_{i}}
			\end{cases}.
        \end{align*}
		where $K_{\sum,i,b} = \prod_{i<i^*}^{n} T_{i^*,b} \cdot \left( \prod_{i>i^*}^n T_{i^*,b} \right)^{-1} \in G$, 
		$\hat{v}_b=\mathcal{H}_2(\ell_{y,b}), b \in [2], v= \mathcal{H}_1(\ell_{D})$.

		We guarantee the correctness of the decryption key share by constructing 
		zero-knowledge proofs as:
		\begin{equation*}
            \text{PoK}
			\left\{(s_i, t_i,\hat{k}):
            \begin{aligned}
                T_i = h_p^{t_i} \wedge \\
				d_i = (f^{\hat{k}_b} K_{\sum,i,b}^{t_{i,b}})_{b \in [2]} \wedge \\
				dk_i = \left( (\hat{v}_b)^{\hat{k}} \cdot h^{s_{i,b} \cdot y_i}\right)_{b \in [2]} \wedge \\
				\mathrm{com}_i = (v)^{s_{i}}
            \end{aligned} 
			\right\}
        \end{equation*}  
		To prove the statement, the prover $\mathcal{P}$ and $\mathcal{V}$ engage 
		in the following zero knowledge protocol.

		\noindent
		$\mathcal{V}$ verifes that $d_i, T_i \in \hat{G}^2$ for $j \in [n]$. \\
		$\mathcal{P}$ computes:
		\begin{align*}
			& r_{\hat{k},i}, r_{s,i} \overset{\$}{\gets} \mathbb{Z}_p^2,\quad r_{t,i} \overset{\$}{\gets} [0,2^\lambda pS]^2 \\
			& R_{T,i} = (h^{r_{t,i,b}}_p)_{b \in [2]} \\
			& K_{\sum,i,b} = \prod_{i<i^*}^{n} T_{i^*,b} \cdot \left( \prod_{i>i^*}^n T_{i^*,b} \right)^{-1} \in G \\
			& R_{d,i} = \left( f^{r_{\hat{k},i,b}}K_{\sum,i,b}^{r_{t,i,b}}\right)_{b \in [2]} \\
			& R_{dk,i} = \left( (\hat{v}_b)^{r_{\hat{k},i}} \cdot h^{r_{s,i.b} \cdot y_i}\right)_{b \in [2]} \\
			& R_{com,i} = (v)^{r_{s,i}}
		\end{align*}
		$\mathcal{P} \to \mathcal{V}: R_{T,i}, R_{d,i}, R_{dk,i}, R_{com,i}$ \\
		$\mathcal{V}: \beta \overset{\$}{\gets} \mathbb{Z}_p$ \\
		$\mathcal{V} \to \mathcal{P}: \beta$ \\
		$\mathcal{P}$ computes:
		\begin{align*}
			& z_{\hat{k},i} = r_{\hat{k},i} - \beta \cdot \hat{k}_i \\
			& z_{s,i} = r_{s,i} -  \beta \cdot s_{i} \\
			& z_{t,i} = r_{t,i} - \beta \cdot t_{i}
		\end{align*}
		$\mathcal{P} \to \mathcal{V}: z_{\hat{k},i}, z_{s,i}, z_{t,i} $\\
		$\mathcal{V}$ verifies:
		\begin{align*}
			R_{T,i} & \overset{?}{=} (T_{i,b}^{\beta} \cdot  h_p^{z_{t,i,b}})_{b \in [2]} \\
			K_{\Sigma, i} & {=} \left\{\prod_{i<i^*}^n T_{i^*, b} \cdot\left(\prod_{i>i^*}^n T_{i^*, b}\right)^{-1}\right\}_{b \in [2]} \in G\\
			R_{d,i} & \overset{?}{=} \left(\mathrm{d}_{i,b}^\beta \cdot f^{z_{\hat{k},i,b}} K_{\Sigma, i, b}^{z_{t, i, b}}\right)_{b \in [2]} \\
			R_{\mathrm{dk}_i} &\overset{?}{=}
						\left( (\mathrm{dk}_i)^{\beta} \cdot (\hat{v}_b)^{z_{\hat{k},i}} \cdot h^{z_{s,b} \cdot y_i}
						\right)_{b \in [2]} \\
			R_{\mathrm{com}, i} &\overset{?}{=} 
						(\mathrm{com}_{i})^{\beta}(v)^{z_{s,i}}
		\end{align*} 
		\begin{theorem}
			The functional key share verification $\prod_{DKeyGenShare}$ 
			has perfect completeness, statistically 
			zeroknowledge, and a computational soundness.
		\end{theorem}

		\begin{proof}
			A formal proof are provided in \cite{VMCFE}.
		\end{proof}

	\section{Proof of Theorem \ref{theorem:pi_enc}}\label{proof:pi_enc}
	\begin{proof}
		\textbf{Completeness.}
		Perfect completeness follows from the fact that for all valid witnesses,
		it holds that
		$y_i = \frac{\langle \vec{x}_i, \vec{x}_0\rangle}{\langle \vec{x}_i, \vec{x}_i\rangle}$
		and
		\begin{align*}
			g^t (\bar{u})^{\tau_{\alpha}} & = g^{t_0 + t_1 \alpha + t_2 \alpha^2} (u)^{\tau_0 + \tau_1 \alpha + \tau_2 \alpha^2} \\ 
                                      & = g^{t_0}u^{\tau_0}(g^{t_1}u^{\tau_1})^{\alpha}(g^{t_2}u^{\tau_2})^{\alpha^2} \\
                                      & = T_0 \left(T_1\right)^\alpha\left(T_2\right)^{\alpha^2} \\
			g^{\langle L, \vec{x}_0\rangle} (\bar{u})^{\tau'_{\alpha}} & = g^{\langle \vec{x}_i + k\alpha, \vec{x}_0\rangle}u^{\tau_0 y_i + \tau'_1 \alpha} \\
                                                             & = (g^{t_0}u^{\tau_0})^{y_i}(g^{t'_1} (u)^{\tau'_1})^\alpha \\
                                                             & = T_0^{y_i} (T'_1)^\alpha \\
			\vec{w}^L (\bar{u})^\omega & = \vec{w}^{\vec{x}_i + k\alpha}u^{s_i + \alpha \rho} \\ 
                             & = \vec{w}^{\vec{x}_i}u^{s_i}(\vec{w}^{k}u^{\rho})^\alpha = V_i(K)^\alpha \\
			t & = \langle L, L\rangle \\
			(v)^{\omega} & = v^{s_i + \alpha \rho} = v^{s_i} (v^{\rho})^{\alpha} = \mathrm{com}_{i} (\mathrm{com}^*)^{\alpha}.
		\end{align*}
		
		\textbf{Soundness.}
		To demonstrate special soundness, we construct an extractor $\mathcal{E}$ as outlined below.
		$\mathcal{E}$ rewinds the $\prod_{\texttt{Encrypt}}$ with $2$ different values of
		$\alpha$ to get 
		\begin{align*}
			L_1 = \vec{x}_i + k \alpha_1 &\text{ and } L_2 = \vec{x}_i + k \alpha_2 \\
			\omega_1 = s_i + \alpha_1 \rho &\text{ and } \omega_2 = s_i + \alpha_2 \rho \\
			\tau'_{\alpha_1} = \tau_0 y_i + \tau'_1 \alpha_1
				&\text{ and } \tau'_{\alpha_2} = \tau_0 y_i + \tau'_1 \alpha_2 \\
			\tau_{\alpha_1} = \tau_0 + \tau_1 \alpha_1 + \tau_2 \alpha_1^2 
				&\text{ and } \tau_{\alpha_2} = \tau_0 + \tau_1 \alpha_2 + \tau_2 \alpha_2^2,
		\end{align*}
		then $\mathcal{E}$ can compute 
		\begin{align*}
			k &=  \frac{L_1-L_2}{\alpha_1-\alpha_2}, \quad \rho = \frac{\omega_1-\omega_2}{\alpha_1-\alpha_2} \\
			\vec{x}_i &= L_1 - k \alpha_1 = L_1 - \alpha_1 \frac{L_1 - L_2}{\alpha_1 - \alpha_2} \\
			s_i &= \omega_1 - \alpha_1 \rho = \omega_1 - \alpha_1 \frac{\omega_1 - \omega_2}{\alpha_1 - \alpha_2} \\
			\tau'_1 &= \frac{\tau'_{\alpha_1}-\tau'_{\alpha_2}}{\alpha_1-\alpha_2}, \quad
				\tau_0 =  \frac{\tau'_{\alpha_1} - \tau'_1 \alpha_1}{y_i} \\
			\tau_1 &= \frac{\alpha_2(\tau'_{\alpha_1}-\tau_0)}{\alpha_1(\alpha_2-\alpha_1)} - 
				\frac{\alpha_1(\tau_{\alpha_2}-\tau_0)}{\alpha_2(\alpha_2-\alpha_1)} \\
			\tau_2 &= \frac{\tau_{\alpha_1}-\tau_0}{\alpha_1(\alpha_1-\alpha_2)} - 
				\frac{\tau_{\alpha_2}-\tau_0}{\alpha_2(\alpha_1-\alpha_2)}
		\end{align*}
		to extract $\vec{x}_i, s_i, \tau_0, \tau_1, \tau_2, \tau'_1$,
		and then further computation of 
		$$
			t_0=\langle \vec{x}_i,\vec{x}_i \rangle, \ t_1 = 2\langle k,\vec{x}_i \rangle, \ t_2= \langle k,k \rangle, \ t'_1 = \langle k, \vec{x}_0 \rangle
		$$
		reveals that $t_0, t_1, t_2, t_1'$.
		
		\textbf{Zero knowledge.}
		To demonstrate the protocol achieves perfect honest-verifier zero-knowledge,
		we construct a simulator that, given a statement 
		$\left(
			 V_i \in \mathbb{G}_1, y_i \in \mathbb{Z}_p, \mathrm{com}_i \in \mathbb{G}^2_1
		\right) $,
		produces a simulated proof whose distribution is indistinguishable from that 
		of a real proof generated through interaction between an honest prover and 
		an honest verifier. 
		All elements chosen by the simulator are either independently and 
		uniformly distributed or fully determined by the verification equations.
		The simulator randomly chooses $T_0, T_2, \tau_\alpha, \tau'_\alpha, \omega$  
		and $L$. Then, for the challenge $\alpha$ and the public value $\vec{x}_0$, the values $t, T_1, T'_1, K$, 
		and $\mathrm{com}^*$ are computed according to the verification equations.		
		\begin{align*}
			t &= \langle L, L\rangle \\
			T_1 &= (g^t(\bar{u})^{\tau_{\alpha}}T_0^{-1}T_2^{-\alpha^2})^{-\alpha}\\
			T'_1 & = \left(g^{\langle L, \vec{x}_0\rangle} (\bar{u})^{\tau'_{\alpha}}T_0^{-y_i}\right)^{-\alpha}\\
			K &= \left(\vec{w}^L (\bar{u})^\omega (V_i)^{-1}\right)^{-\alpha}\\
			\mathrm{com}^* &= (v^\omega (\mathrm{com}_i)^{-1} )^{-\alpha}
		\end{align*}
		The protocol remains zero-knowledge because we can successfully simulate 
		the proof procedure. Therefore, any information revealed or potentially 
		leaked during the simulation does not compromise the zero-knowledge property 
		of the overall protocol.
	\end{proof}
	
	\section{Lemma \ref{lemma_w0} and Proof}
	\begin{lemma}\label{lemma_w0}
		Assume the objective function \( F(W) \) is \( \mu \)-strongly convex and 
		\( L \)-smooth (Assumption 1). Let \( W^{opt} \) denote the optimal global model of \( F(W; D) \), 
		and \( W^t_0 \) denote the baseline model at iteration \( t \), obtained using 
		gradient descent on a clean dataset \( D_0 \) (Assumption 2) with a learning rate \( \eta \). 
		The following bound holds:  
		\[
			\|W_0^t - W^{opt}\| \leq \left(\sqrt{(1 - 2\eta \mu+\eta^2 L^2)}\right)^{lr} \|W_0 - W^{opt}\|.
		\]
		where $W_0$ is the initial model.
	\end{lemma}
	\begin{proof}
		Let \( lr \) denotes the total number of training iterations on \( D_0 \)  over $t$ rounds of iteration. 
		We define \( w^0 = W_0 \), 
		and after \( lr \) rounds of local training on dataset $D_0$, the final model is denoted by 
		\( w^{lr} = W_0^t \). During the \( i \)-th round of training, the gradient 
		descent update rule gives:
		$$
			w^{i+1} = w^i - \eta \nabla F(w^i; D_0).
		$$
		Subtracting \( W^{opt} \) from both sides, we obtain:
		$$
			w^{i+1} - W^{opt} = w^i - W^{opt} - \eta \nabla F(w^i; D_0).
		$$
		Taking the norm and square on both sides, we have:
		\begin{align*}
			&\|w^{i+1} - W^{opt}\|^2 \\
            = &\|w^i - W^{opt}\|^2 
			- 2\eta \langle \nabla F(w^i; D_0), w^i - W^{opt} \rangle \\
			& + \eta^2 \|\nabla F(w^i; D_0)\|^2.
		\end{align*}
		By Assumption 1 and $\nabla F(W^{opt})=0$, we have:
		\begin{align*}
			F(w^i) + \langle \nabla F(w^i; D_0), (W^{opt} - w^i) \rangle & \\
                + \frac{\mu}{2} \|W^{opt} - w^i\|^2 &\leq F(W^{opt}) \\ 
			F(W^{opt}) + \frac{\mu}{2} \|w^i - W^{opt}\|^2 &\leq F(w^i) \\
			\| \nabla F(w^i; D_0) - \nabla F(W^{opt}) \|^2 &\leq L^2 \| w^i - W^{opt}\|^2
		\end{align*}
		Then, it satisfies the following property:
		\begin{align*}
			- 2\eta \langle \nabla F(w^i; D_0), w^i - W^{opt} \rangle & \le 
			- 2\eta \mu \| w^i-W^{opt} \|^2 \\
			\eta^2 \|\nabla F(w^i; D_0)\|^2 &\le \eta^2 L\| w^i- W^{opt}\|^2
		\end{align*}
		Then, we have
		$$
			\|w^{i+1} - W^{opt}\|^2 \le (1 - 2\eta \mu+\eta^2 L^2)\| w^i- W^{opt}\|^2
		$$
		$$
			\|w^{i+1} - W^{opt}\| \le \sqrt{(1 - 2\eta \mu+\eta^2 L^2)}\| w^i- W^{opt}\|
		$$
		Recursively applying this inequality over \( lr \) rounds yields:
		$$
			\|w^{lr} - W^{opt}\| \leq \left(\sqrt{(1 - 2\eta \mu+\eta^2 L^2)}\right)^{lr} \|w^0 - W^{opt}\|.
		$$
		Since \( w^0 = W_0, w^{lr} = W_0^t \), the global model after \( lr \) local training rounds satisfies:
		$$
			\|W_0^t - W^{opt}\| \leq \left(\sqrt{(1 - 2\eta \mu+\eta^2 L^2)}\right)^{lr} \|W_0 - W^{opt}\|.
		$$
	\end{proof}
}

\vfill

\end{document}